\documentclass[11pt,letterpaper]{article}
\usepackage[margin=1in]{geometry}
\usepackage{datetime}
\usepackage{enumitem}
\usepackage{wrapfig}
\usepackage{subcaption}
\usepackage[T1]{fontenc}
\usepackage[utf8]{inputenc}
\usepackage{pgfplots}
\pgfplotsset{compat=1.15}
\usepackage{amsmath, amssymb}
\usepackage{amsthm}
\usepackage{color}
\usepackage{cases}
\usepackage{xparse}
\usepackage{xargs}
\usepackage{appendix}
\usepackage{verbatim, xspace}
\usepackage{tikz}
\usepackage{mathtools}
\usepackage[font=footnotesize]{caption}

\usepackage{framed}
\usepackage[framemethod=TikZ]{mdframed}

\usepackage[linesnumbered,lined,boxed,commentsnumbered,noend]{algorithm2e}

\newcommand{\LineIf}[2]{\textbf{if}\ {#1}\ \textbf{then} {#2} }

\usepackage[colorlinks,allcolors=blue]{hyperref}
\usepackage{pgfplots}

\usepackage[capitalise]{cleveref}
\usepackage{xcolor}
\newenvironment{claimproof}[1]{\par\noindent\emph{Proof.}\space#1}{\hfill $\lrcorner$\\}

\newcommand{\citet}[1]{\cite{#1}}

\usepackage[numbers]{natbib}
\newcounter{ctr}
\loop
 \stepcounter{ctr}
 \expandafter\edef\csname c\Alph{ctr}\endcsname{\noexpand\mathcal{\Alph{ctr}}}
 \expandafter\edef\csname b\alph{ctr}\endcsname{\noexpand\mathbf{\alph{ctr}}}
 \expandafter\edef\csname b\Alph{ctr}\endcsname{\noexpand\mathbf{\Alph{ctr}}}
\ifnum\thectr<26
\repeat

\usepackage[colorinlistoftodos,prependcaption,textsize=tiny,disable]{todonotes}
\newcommandx{\unsure}[2][1=]{\todo[linecolor=green,backgroundcolor=green!25,bordercolor=green,#1]{\normalsize #2}}
\newcommandx{\improvement}[2][1=]{\todo[inline,linecolor=blue,backgroundcolor=blue!05,bordercolor=blue,#1]{\normalsize #2}}
\newcommandx{\info}[2][1=]{\todo[linecolor=yellow,backgroundcolor=yellow!25,bordercolor=yellow,#1]{#2}}
\newcommandx{\floatmodel}[2][1=]{\todo[inline,linecolor=red,backgroundcolor=yellow!25,bordercolor=yellow,#1]{#2}}
\newcommandx{\thiswillnotshow}[2][1=]{\todo[disable,#1]{#2}}
\newcommandx{\karol}[2][1=]{\todo[inline,linecolor=blue,backgroundcolor=blue!25,bordercolor=blue,caption={\normalsize \textbf{Karol}},#1]{\normalsize #2}}
\newcommandx{\jesper}[2][1=]{\todo[linecolor=red,backgroundcolor=red!25,bordercolor=red,caption={\normalsize \textbf{Jesper}},#1]{\normalsize #2}}
\newcommandx{\sandor}[2][1=]{\todo[inline,linecolor=gray,backgroundcolor=red!25,bordercolor=red,caption={\normalsize \textbf{S\'andor}},#1]{\normalsize #2}}

\theoremstyle{plain}
\newtheorem{theorem}{Theorem}[section]
\newtheorem{definition}[theorem]{Definition}
\newtheorem{lemma}[theorem]{Lemma}
\newtheorem{corollary}[theorem]{Corollary}
\newtheorem{claim}[theorem]{Claim}

\newtheorem{observation}[theorem]{Observation}

\newcommand{\thistheoremname}{}
\newtheorem*{genericthm}{\thistheoremname}
\newenvironment{namedthm}[1]
  {\renewcommand{\thistheoremname}{#1}%
   \begin{genericthm}}
  {\end{genericthm}}
\theoremstyle{remark}
\newtheorem{remark}[theorem]{Remark}

\newcommand{\floor}[1]{\left\lfloor #1 \right\rfloor}

\newcommand{\eps}{\varepsilon}

\newcommand{\Oh}{\mathcal{O}}

\newcommand{\Otilde}{\widetilde{\Oh}}
\newcommand{\Ot}{\Otilde}

\newcommand{\Reals}{\mathbb{R}}

\newcommand{\poly}{\mathrm{poly}}
\newcommand{\polylog}{\,\textup{polylog}}

\newcommand{\Ex}[1]{\mathbb{E}\left[ #1 \right]}

\newcommand{\wt}{\mathrm{wt}}
\newcommand{\dist}{\mathrm{dist}}
\newcommand{\pro}{\mathrm{pro}}
\newcommand{\ex}{\mathrm{ex}}
\newcommand{\grid}{\mathrm{grid}}
\newcommand{\cost}{\mathrm{cost}}
\newcommand{\ori}{\mathrm{cor}}

\newcommand{\ins}{\mathrm{in}}
\newcommand{\out}{\mathrm{out}}
\newcommand{\empt}{\mathrm{empty}}
\newcommand{\dum}{\mathrm{dummy}}
\newcommand{\cdc}{\mathrm{cdc}}

\newcommand{\mst}{\mathrm{MST}}
\newcommand{\smt}{\mathrm{ST}}

\newcommand{\Tspan}{T_\mathrm{spanner}}
\newcommand{\PT}{\textsf{PF}_F}

\renewcommand{\leq}{\leqslant}
\renewcommand{\geq}{\geqslant}
\renewcommand{\le}{\leqslant}
\renewcommand{\ge}{\geqslant}

\newcounter{openquestion}

\newenvironment{insight}
{\mdfsetup{%
    nobreak=true,
	middlelinecolor=gray,
	middlelinewidth=1pt,
	backgroundcolor=gray!10,
    innertopmargin=7pt,
	roundcorner=5pt}
\begin{mdframed}}
{\end{mdframed}}

\newcommand{\defproblem}[3]{
	\vspace{2mm}
	\vspace{1mm}
	\noindent\fbox{
		\begin{minipage}{0.95\textwidth}
			#1 \\
			{\textbf{Input:}} #2  \\
			{\textbf{Task:}} #3
		\end{minipage}
	}
	\vspace{2mm}
}

\title{A Gap-ETH-Tight Approximation Scheme for Euclidean TSP}
\date{}

\author{
    S\'andor Kisfaludi-Bak\footnote{Aalto University, Espoo, Finland, \texttt{sandor.kisfaludi-bak@aalto.fi}}
    \and
    Jesper Nederlof\footnote{Utrecht University, The
    Netherlands, \texttt{j.nederlof@uu.nl}. Supported by
    the project CRACKNP that has received funding from the European
    Research Council (ERC) under the European Union’s Horizon 2020 research and
    innovation programme (grant agreement No 853234).}
    \and
    Karol W\k{e}grzycki\footnote{Saarland University and Max Planck Institute for Informatics,
        Saarbr\"ucken, Germany, \texttt{wegrzycki@cs.uni-saarland.de}. 
    This work is part of the project TIPEA that has
    received funding from the European Research Council (ERC) under the European Unions Horizon
    2020 research and innovation programme (grant agreement No. 850979).
    Author was also supported Foundation for Polish Science (FNP), by the grants
    2016/21/N/ST6/01468 and 2018/28/T/ST6/00084 of the Polish National Science
    Center and project TOTAL that has received funding from the European
    Research Council (ERC) under the European Union’s Horizon 2020 research and
    innovation programme (grant agreement No 677651).}
}

\begin{document}

\hypersetup{pageanchor=false}
\begin{titlepage}
\maketitle
\thispagestyle{empty}

\begin{abstract}
    We revisit the classic task of finding the shortest tour of $n$ points in $d$-dimensional Euclidean space, for any fixed constant $d \geq 2$. We determine the optimal dependence on $\eps$ in the running time of an algorithm that computes a $(1+\eps)$-approximate tour, under a plausible assumption.
    Specifically, we give an algorithm that runs in $2^{\Oh(1/\eps^{d-1})} n\log n$ time. This improves the previously smallest dependence on $\eps$ in the running time $(1/\eps)^{\Oh(1/\eps^{d-1})}n \log n$ of the algorithm by Rao and Smith~(STOC 1998). We also
    show that a $2^{o(1/\eps^{d-1})}\poly(n)$ algorithm would violate the
    Gap-Exponential Time Hypothesis (Gap-ETH).

    \medskip

	Our new algorithm builds upon the celebrated quadtree-based methods
    initially proposed by Arora (J. ACM 1998), but it adds a new idea
    that we call \emph{sparsity-sensitive patching}. On a high level this lets
    the granularity with which we simplify the tour depend on how sparse it is
    locally. We demonstrate that our technique
    extends to other problems, by showing that for Steiner Tree
    and Rectilinear Steiner Tree it yields the same running time. We complement our
    results with a matching Gap-ETH lower bound for Rectilinear Steiner Tree.
    
\end{abstract}

\begin{picture}(0,0)
\put(-70,-270)
{\hbox{\includegraphics[width=40px]{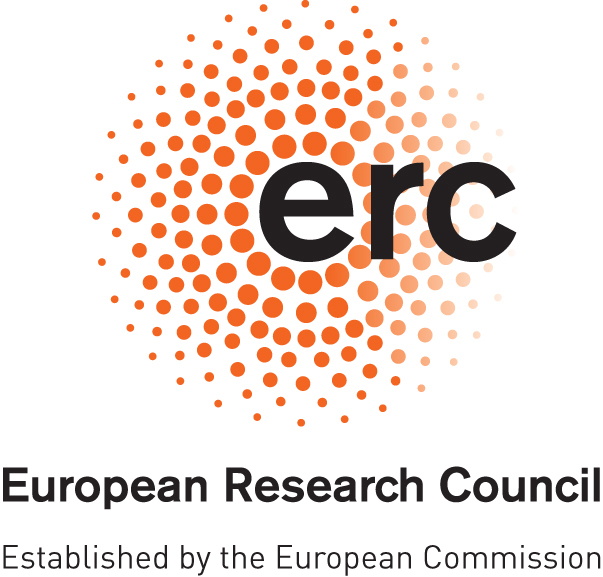}}}
\put(-80,-330)
{\hbox{\includegraphics[width=60px]{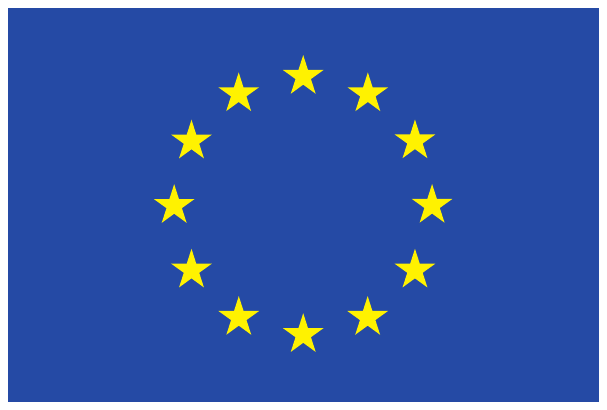}}}
\end{picture}

\end{titlepage}

\hypersetup{pageanchor=true}

\section{Introduction}

The Euclidean Traveling Salesman Problem (\textsc{Euclidean TSP}) is to find a
round trip of minimum length for a given set of $n$ points in $d$-dimensional
Euclidean space.  Its simple statement and clear applicability make the problem
very attractive, and work on it has been immensely influential and
inspirational.  In particular, the G\"odel-prize-winning approximation schemes
due to Arora~\cite{Arora98} and Mitchell~\cite{Mitchell99} are among the most
prominent results in approximation algorithms. Because of their elegance, they
serve as evergreens in graduate algorithms courses, textbooks on approximation
algorithms or optimization~\cite{vazirani-book,williamson2011design,
kortevygen12}, and more specialized textbooks~\cite{Har-Peled11,geomspannet}.

After the publication of these results, an entire research program with many
strong results consisting of improvements, generalizations and different
applications of the methods from~\cite{Arora98,Mitchell99} was conducted by many
authors (see e.g., the survey~\cite{Arora03}). The technique is now known to be
useful for a whole host of geometric optimization problems (see the related work
paragraph).

The most natural goal within this research direction is to improve the running
times to be \emph{optimal}, i.e. to improve and/or provide evidence that further
(significant) improvements do not exist.  In the last 25 years, only two such
results were obtained in $\mathbb{R}^d$:
\begin{enumerate}
	\item Rao and Smith~\cite{RaoS98} used geometric spanners to improve the
        $n(\log n)^{\Oh(1/\eps)^{d-1}}$ time approximation scheme of
        Arora~\cite{Arora98} to run in only $(1/\eps)^{\Oh(1/\eps)^{d-1}}\cdot n\log n$ time.\footnote{For dimension $d$, \cite{RaoS98} claimed $(1/\eps)^{\Oh(1/\eps)^{d-1}}n +
            \Tspan(n,\eps)$ time, where $\Tspan(n,\eps)$ is the spanner computation time.
            See~\cite[Chapter 19]{geomspannet} for a more detailed description of an
        $(1/\eps)^{\Oh(1/\eps^{d^2})}n\log n$ time algorithm.}
    \item Bartal and Gottlieb~\cite{BartalG13} gave a $2^{(1/\eps)^{\Oh(d)}}
            \cdot n$ time algorithm in the real-RAM model with atomic floor or mod operators. They
            give a truly linear algorithm in terms of $n$, however the
            dependence on $\eps$ is worse than the algorithm of Rao and
            Smith~\cite{RaoS98}.
\end{enumerate}
While these results determine the optimal\footnote{Depending on the model of
computation $\Omega(n \log n)$ time is required~\cite{Das97}.} dependence on
$n$, they do not yet settle the much faster growing exponential dependence on
$\eps$.  This is in contrast with the status of our knowledge of the complexity
of many other optimization problems: In the last decade a powerful toolbox for
determining (conditionally) optimal exponential running times has been
developed.

In the context of \textsc{TSP} in $d$-dimensional Euclidean space (henceforth
denoted by $\Reals^d$), this modern research direction culminated in an exact
algorithm with a running time of $2^{\Oh(n^{1-1/d})}$, which was matched by a
lower
bound of $2^{\Omega(n^{1-1/d})}$~\cite{BergBKK18} under the Exponential Time
Hypothesis (ETH).

In the context of approximation schemes for TSP, Klein~\cite{Klein08} improved
algorithms for the unweighted planar case from \cite{tsp-planar-1,tsp-planar-2}
with a $2^{\Oh(1/\eps)}n$ time approximation scheme. Subsequently,
Marx~\cite{marx07} showed that the dependence on $\eps$ in Klein's algorithm is
conditionally near-optimal. The tight exponential dependency of $\eps$ in
approximation schemes was also obtained in a plethora of other
problems, such as \textsc{Maximum Independent Set} in planar graphs~\cite{baker}
and a scheduling problem~\cite{jansen-scheduling} (see, e.g.,
    \cite{FeldmannSLM20}
for a survey).

Given the modern trend of fine-grained algorithm research and the prominence of
the discussed approximation schemes for \textsc{Euclidean TSP}, our goal
suggests itself:

\begin{insight}
    \label{main-goal}
    \textbf{Goal:} Conclude the research on approximation schemes for Euclidean
    TSP with a conditionally optimal algorithm.
\end{insight}

\subsection{Our contribution}

In this work, we achieve this goal for \textsc{Euclidean TSP} and
\textsc{(Rectilinear) Steiner Tree} in $\Reals^d$ and give algorithms with a
Gap-ETH-tight dependence on $\eps$. Our main result reads as follows.

\begin{theorem}[Main result]\label{thm:tsp}
    For any integer $d\geq 2$, there is a randomized $(1+\eps)$-approximation
    scheme for
    \textsc{Euclidean TSP} in $\Reals^d$ that runs in $2^{\Oh(1/\eps^{d-1})}n +
    \poly(1/\eps) n \log(n)$ time.
    Moreover, this cannot be improved to a $2^{o(1/\eps^{d-1})}\cdot \poly(n)$
    time algorithm, unless Gap-ETH fails.
\end{theorem}

Thus, we improve the previously best $(1/\eps)^{\Oh(1/\eps^{d-1})}$ dependence
of $\eps$ in the running time of~\cite{RaoS98} to $2^{\Oh(1/\eps^{d-1})}$. Note
that here and in the sequel, our big-$\Oh$ notation hides factors that depend
only on $d$ since it is assumed to be constant. Our running times are double
exponential in the dimension $d$, which is expected because of Trevisan's lower
bound~\cite{Trevisan00}.

Theorem~\ref{thm:tsp} improves the running time dependence on $\eps$ all
the way to conditional optimality: we show that an EPTAS with an asymptotically
better dependence on $\eps$ in the exponent is not possible under Gap-ETH, for
constant dimension~$d$. Note that our algorithms can be derandomized at the cost
of an extra $n^d$ factor in the running time, which maintains conditional
optimality in terms of $\eps$.

Our lower bound for \textsc{Euclidean TSP} is derived from a construction
for \textsc{Hamiltonian Cycle} in grid graphs~\cite{frameworkpaperjournal}, in
combination with
Gap-ETH~\cite{Dinur16,ManurangsiR17} (see Section~\ref{sec:lower-bounds}).

Our new algorithmic techniques enable us to improve approximation schemes for
another fundamental geometric optimization problem: the \textsc{Rectilinear
Steiner Tree} and \textsc{Euclidean Steiner Tree} problems.

\begin{theorem}\label{thm:st}
    For any integer $d\geq 2$, there is a randomized $(1+\eps)$-approximation
    scheme for \textsc{Euclidean Steiner Tree} and \textsc{Rectilinear Steiner
    Tree} in $\Reals^d$ that runs in $2^{\Oh(1/\eps^{d-1})}n + \poly(1/\eps) n
    \log(n)$ time. Moreover, the algorithm for \textsc{Rectilinear Steiner
    Tree} cannot be improved to a $2^{o(1/\eps^{d-1})}\cdot \poly(n)$ time
    algorithm, unless Gap-ETH fails.
\end{theorem}

This directly improves state-of-the-art algorithms by Rao and
Smith~\cite{RaoS98} and Bartal and Gottlieb~\cite{bartal-steiner-tree} in all
regimes of parameters $n$ and $\eps$.

The lower bound for \textsc{Rectilinear Steiner Tree} requires new ideas since
there is no known ETH-based lower bound for the exact version of the problem.
Our construction is based on a reduction from \textsc{Grid Embedded Connected
Vertex Cover} from~\cite{frameworkpaperjournal} and a combination of gadgets
proposed by~\cite{garey1977rectilinear} (see Section~\ref{sec:lower-bounds}). We
leave it as an open problem to give a matching lower bound on our algorithm for
\textsc{Euclidean Steiner Tree}.

\subsection{The existing approximation schemes and their limitations}

The approximation scheme from Arora~\cite{Arora98} serves as the basis of our
algorithm, and we assume that the reader is familiar with its basics
(see~\cite{williamson2011design,vazirani-book} for a comprehensive introduction
to the approximation scheme). In this section, we consider $d=2$ for simplicity.

In a nutshell, Arora's strategy in the plane is first to move the points to the
nearest grid points in an $L\times L$ grid where $L=\Oh(n/\eps)$. This grid is
subdivided using a hierarchical decomposition into smaller squares (\emph{a
quadtree}, see definition in Section~\ref{sec:etsp-in-plane}), where on each
side of a square $\Oh((\log n)/\eps)$ equidistant \emph{portals} are placed.
Arora proves a \emph{structure theorem}, which states that there is a tour of
length at most $(1+\eps)$ times the optimal tour length that crosses each square
boundary $\Oh(1/\eps)$ times and only through portals. This structure theorem
is based on a \emph{patching procedure}, which iterates through the cells of the
quadtree (starting at the smallest cells) and patches the tour such that the
resulting tour crosses all cell boundaries only $\Oh(1/\eps)$ times and only at
portals, and it does it in such a way that the new tour is only slightly longer.
While such a promised slightly longer tour does not necessarily exist for a
fixed quadtree, a randomly shifted quadtree works with high probability. The
algorithm thus proceeds by picking a randomly shifted quadtree and by
performing a dynamic programming algorithm on progressively larger squares and
the bounded set of possibilities in it to find a patched tour.

The first improvement to Arora's algorithm was achieved by Rao and
Smith~\cite{RaoS98} (see \cite[Chapter 16]{geomspannet} for a modern description
of their methods). Recall that Arora placed equidistant portals. Rao and Smith's
idea is to use \emph{light spanners} to "guide" the approximate TSP tour and
select portals on the boundary not uniformly. They show that it is sufficient to
look for the shortest tour within a spanner, or more precisely, they patch the
given spanner such that the resulting graph has $1/\eps^{\Oh(1)}$ crossings with
each quadtree cell, while still containing a $(1+\eps)$-approximate tour.
Similarly to Arora's algorithm, it is sufficient to consider tours that cross
each square boundary $\Oh(1/\eps)$ times, but now the number of portals is
$(1/\eps)^{\Oh(1)}$. Consequently, the algorithm of Rao and Smith needs only
$\left((1/\eps)^{\Oh(1)}\right)^{\Oh(1/\eps)}=2^{\Oh((1/\eps)\cdot \log
(1/\eps))}$ subsets of portals to consider for each square in their
corresponding dynamic programming algorithm.

\paragraph{Why do known techniques fail to get a better running time?}
To get the dependence on $\eps$ in the running time down to $2^{\Oh(1/\eps)}$,
the bottleneck is to get the number of candidate sets of where the tour crosses
a cell boundary down to $2^{\Oh(1/\eps)}$.\footnote{To properly solve all
    required subproblems, the dynamic programming algorithm also needs to
    consider all matchings on such a candidate set, but this can be circumvented
    by invoking the rank-based approach from~\cite{rank-based} that allows one
    to restrict attention to only $2^{\Oh(1/\eps)}$ matchings as long as the
candidate set has cardinality $\Oh(1/\eps)$.}

One could hope to improve Arora's algorithm by decreasing the number of portals
from $\Oh(\log n /\eps)$ to $\Oh(1/\eps)$, but this is not possible: the
structure theorem would fail even if the optimal tour is a rotated square with equally distributed points on its sides.

Another potential approach would be to improve the spanners and the spanner
modification technique of Rao and Smith to get a graph that contains a
$(1+\eps)$-approximate tour, while having only $\Oh(1/\eps)$ crossings on each
side of each square. Such an improvement seems difficult to accomplish as even
with Euclidean spanners~\cite{truly-optimal-spanners} of optimal lightness or
the more general \emph{Euclidean Steiner spanners}~\cite{steiner-spanners}, one
cannot get the required guarantee. Le and Solomon~\cite{truly-optimal-spanners}
gave a lower bound of $\tilde\Omega(1/\eps)$ on the lightness of Euclidean
Steiner Spanners in $d=2$, which was matched very recently by Bhore and
T\'oth~\cite{steiner-spanner-scg21}. Even with that optimal Steiner spanner,
the patching method of Rao and Smith yields a guarantee of only $\Ot(1/\eps^2)$
crossings per square and it is not clear if one can even get
$\Oh(1/\eps^{1.99})$ potential crossings per square.

\subsection{Our technique: Sparsity-Sensitive Patching}
We introduce a new patching procedure. Slightly oversimplifying and still focusing on $2$ dimensions, it iterates over the cells of the quadtree and processes a cell boundary as follows:
\begin{insight}
	\textbf{Sparsity-Sensitive Patching}: For a cell boundary that is crossed by
    a tour at $1 < k \leq \Oh(1/\eps)$ crossings,
	modify the tour by mapping each crossing to the nearest portal from the set of $g$ equidistant portals.
	Here $g$ is a granularity parameter \emph{that depends on $k$} as $g = \Theta(1/(\eps^2k))$.
\end{insight}
See Figure~\ref{fig:arora_vs_sparsity} for an illustration.
This can be used in combination with dynamic programming to prove the algorithmic part of
Theorem~\ref{thm:tsp} since it produces a tour for which the number of
possibilities for the set of crossings of the tour with a cell boundary is
$\sum_{k}\binom{\Oh(1/(\eps^2k))}{k} =2^{\Oh(1/\eps)}$ (see Claim~\ref{portal-bin-ineq}).

One notable aspect of our technique is that it allows to get running times faster than the ones of Arora~\cite{Arora03}, but without the use of spanners, see Remark~\ref{rem:spannerfree_alg}.

\begin{figure}[t]
\centering
\includegraphics[width=\textwidth]{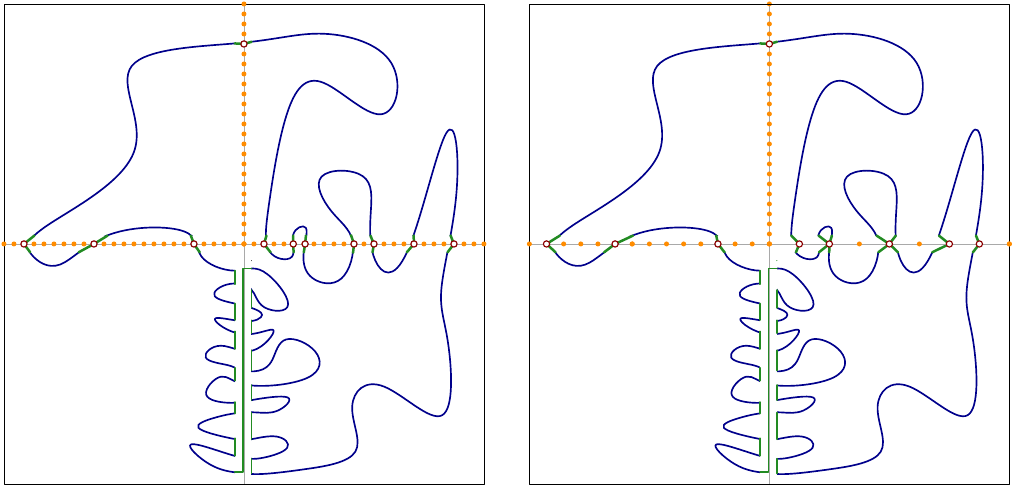}
\caption{On neighboring cells of the quadtree, one must ensure that the tour
    crosses at most $1/\eps$ times, chosen from a limited set of portals.
    Left: Arora's structure theorem snaps the tour to one of $\Oh(\frac{\log
    n}{\eps})$ equally spaced portals. Right: The number of possible portal
    locations depends on the number of crossings; the fewer portals are used, the more precisely they are chosen. Both techniques use the \emph{Patching Lemma} between the bottom two cells as their shared boundary is crossed more than $1/\eps$ times.
}\label{fig:arora_vs_sparsity}
\end{figure}

\paragraph{Bounding the patching cost.}

Our Sparsity-Sensitive procedure may seem quite similar to Arora's patching
procedure, so one may wonder why previous improvements to Arora's procedure
overlooked it. The reason is that our procedure is slightly counter-intuitive,
as it increases precision when the tour is already sparse, and it requires a
subtle analysis of the patching cost. This proof is the main contribution of
this paper.

We will informally describe how we achieve this next. Similar to the patching
cost analysis of Arora's patching procedure, our starting point is that the
total number of crossings that an optimal tour $\pi$ will have with all
horizontal and vertical lines aligned at integer coordinates is proportional to
the total weight of the tour (Lemma~\ref{lem:crossingsvslength}). Since we can
afford an additional cost of $\eps\cdot \wt(\pi)$, it is sufficient to show that
each crossing incurs, in an amortized sense, at most $\Oh(\eps)$ patching cost.

Let $\text{PC}(k,\ell)$ be the patching cost of a horizontal quadtree-cell side
of length $\ell$ with $k$
crossings. Since we connect each crossing to a portal that is at most $\ell/g$
distance away, and the total patching cost is never greater than $\Oh(\ell)$
(since
we can just ``buy'' an entire line, as illustrated in the bottom of Figure~\ref{fig:arora_vs_sparsity}), we obtain
$\text{PC}(k,\ell) \leq \Oh(\min\{\ell,k\ell/g\})$.
The amortized patching cost per crossing is then

\begin{equation}\label{eq:int}
    \frac{\text{PC}(k,\ell)}{k} = \Oh\left(\frac{\min\{\ell, k\frac{\ell}{g} \}}{k}\right) = \Oh \left(\frac{\ell}{k} \min\{1, (k\eps)^2 \}\right),
\end{equation}
and this is maximized when $k=1/\eps$, for which it is $\eps \ell$.

Because we consider a random shift of the quadtree, a crossing of $\pi$ with a fixed
horizontal line $h$ will end up in a cell side of
length $L/2^i$ with probability at most $2^{i-1}/L$, for each $0 \leq i \leq \log L$ (Lemma~\ref{levelprob}).
Letting $\alpha_i(x)$ be the (amortized) patching cost due to the crossing $x$ on line $h$ if $h$ has level $i$, $x$ incurs 
\begin{equation}\label{eq:badbound}
    \sum_{i=0}^{\log L+1}\frac{2^{i-1}}{L} \cdot \alpha_i(x)
\end{equation}
amortized patching cost in expectation.
Naively applying~\eqref{eq:int} for each $i$ to get $\alpha_i \le \eps L/2^i$ and putting this bound into
\eqref{eq:badbound}, gives an undesirably high cost of $\Oh(\eps \log L)$.

To get this cost down to $\Oh(\eps)$, we need to use a more refined argument. We
exploit the fact that the bound $\alpha_i(x) \le \eps L/2^i$ is tight only for a
single $i=i^\ast$ in the worst case. Subsequently, we show that for levels above
$i^\ast$ we have a geometrically decreasing series of costs, which will
demonstrate that the cost in \eqref{eq:badbound} is bounded by $\Oh(\eps)$.
Since we amortize the cost by the length of the tour inside a cell, we should
increase the precision as we move to levels below $i^\ast$. Intuitively, for
level $i^\ast+1$ the number of intersections is halved, while the tour length
(which should be proportional to the area of the cell) is divided by four. This
suggests that the number of portals should be inversely proportional to the
number of crossings.

In our proof we formalize this with a charging scheme based on the distance of the crossing to the
next crossing on the horizontal line.

\subsection{More related work}\label{subsec:relatedwork}
The framework of Arora~\cite{Arora98} and Mitchell~\cite{Mitchell99} was
employed for several other optimization problems in Euclidean space such as
\textsc{Steiner Forest}~\cite{euclidean-steiner-forest}, $k$-\textsc{Connectivity}~\cite{euclidean-mincostconnectivity}, $k$-\textsc{Median}~\cite{ekmedian,kmedian}, \textsc{Survivable Network Design}~\cite{survivable}. We hope our techniques will also find some applications in them.

The original results from~\cite{Arora98,Mitchell99} were also applied or generalized to different settings.
The state-of-the-art for the \textsc{Traveling Salesman Problem} in planar graphs
is now very similar to the Euclidean case. In~\cite{grigni95}, the authors gave the
first PTAS for TSP in planar graphs, which was later extended by 
\cite{weighted-planar-graphs} to weighted planar graphs. Klein~\cite{Klein08}
proposed a $2^{\Oh(1/\eps)}n$ time approximation scheme for TSP in unweighted planar
graphs, which later was proven by Marx~\cite{marx07} to be optimal assuming ETH. Klein~\cite{Klein06} also studied a weighted subset version of TSP that generalizes the planar Euclidean case and gave a PTAS for the problem.

The literature then generalized the metrics much further. Without attempting
to give a full overview, some prominent examples are the algorithms in
minor free graphs~\cite{DemaineHK11,BorradaileLW17,hung-le-soda20},
algorithms in doubling metrics~\cite{BartalGK16,ChanJ18}, and algorithms in
negatively curved spaces~\cite{KrauthgamerL06}, each of which is at least
inspired by the result of Arora~\cite{Arora98} and Mitchell~\cite{Mitchell99}.

Recently,  Gottlieb and Bartal~\cite{bartal-steiner-tree} gave a PTAS for
\textsc{Steiner Tree} in doubling metrics. Moreover, they proposed a
$2^{(1/\eps)^{\Oh(d^2)}}n \log{n}$ time algorithm for \textsc{Steiner Tree} in
$d$-dimensional Euclidean Space with a novel construction of banyan.

There is also a vast literature concerning Euclidean Spanners (see the
book~\cite{geomspannet} for an overview). Very recently Le and
Solomon~\cite{truly-optimal-spanners} proved that greedy spanners are optimal
and in \cite{steiner-spanners} they gave a novel construction of light Euclidean
Spanners with Steiner points. Many such results mention approximation schemes for \textsc{Euclidean TSP} as a major motivation.

\subsection{Organization}
This paper is organized as follows.
In Section~\ref{sec:etsp-in-plane} we define the building blocks of Arora's
approach that we use. 
Section~\ref{sec:pfstruct} proves the Structure Theorem, and in Section~\ref{sec:algorithm} we show how to use it in combination with dynamic programming to establish the algorithmic part of Theorem~\ref{thm:tsp}. Section~\ref{sec:steinertree} extends these techniques to prove the algorithmic results of Theorem~\ref{thm:st}. 
In Section~\ref{sec:lower-bounds} the matching lower bounds are presented, and in Section~\ref{conclusion} we conclude the paper.

\section{Preliminaries}
\label{sec:etsp-in-plane}

Throughout this paper, $\log$ denotes the logarithm of base $2$. We use standard graph notation, and the set $\{1,\dots,k\}$ is denoted by $[k]$. 

For a given set of points $S \subseteq \mathbb{R}^d$, a \emph{tour} is defined
to be a cycle $\pi = (s_1,\ldots,s_n,s_1)$ with vertex (multi)set $S$, which visits
each point and returns to its starting point. Note that in this definition, we
allow points to be visited multiple times. The length (sometimes called
\emph{weight}) $\wt(\pi)$ of a tour $\pi = (s_1,\ldots,s_n,s_1)$ is defined as
$\sum_{i=1}^n \dist(s_i,s_{i+1})$, where $s_{n+1} = s_1$. Hence, a tour $\pi$
consists of a sequence of segments that share endpoints consecutively. 
A \emph{geometric graph} is an embedding of a graph in $\Reals^d$ where vertices are points and edges are segments that connect the corresponding points. For technical reasons, we allow both the vertices and edges of a geometric graph to be a multiset of points, i.e., we allow vertices and edges to coincide in the geometric sense. For example a tour can be regarded as a connected geometric graph where the corresponding graph is a cycle. Occasionally, we will also think of embedded graphs where the edges are represented by a path of segments rather than a single segment.

A \emph{salesman tour} of the point set $P \subseteq \mathbb{R}^d$ is a tour of some points
$S \supseteq P$ (hence, a salesman tour is a closed polyline that passes through
each point in $P$ and is allowed to make some digressions). In the
\textsc{Euclidean Traveling Salesman Problem} (Euclidean TSP), one needs to
return the length of the minimum salesman tour of given points. If
$\pi^*$ is an optimal TSP tour, the standard $(1+\eps)$-approximation
scheme of the problem reports a length in the range $[\wt(\pi^*),
(1+\eps)\wt(\pi^*)]$ (throughout the paper, we assume that $\eps$ is a
real number with $0 < \eps < 1$).

A \emph{Steiner tree} of a point set $P \subseteq \mathbb{R}^d$ is a connected geometric graph that contains $P$ as vertices. In a \emph{rectilinear Steiner tree} we additionally require that each edge of the graph is axis-parallel.
In the \textsc{Euclidean} (resp. \textsc{Rectilinear}) \textsc{Steiner Tree} problem, for a given $P\subset \Reals^d$ the goal is to find a Steiner tree (resp., rectilinear Steiner tree) of $P$ with minimum total weight.

In the following we assume an instance of \textsc{Euclidean TSP} is given by a set $P$ of $n$ points in $\mathbb{R}^d$. By preprocessing the input instance in $\Oh(n
\log (n/\eps))$ time\footnote{By using a different computational model, this is
counted as $\Oh(n)$ time in~\cite{BartalG13}.} (see e.g.~\cite[Section
19.2]{geomspannet}), we may assume that $P \subseteq \{0,\ldots,L\}^d$ for
some integer $L = \Oh(n\sqrt{d} / \eps )$ that is a power of $2$.

\subparagraph*{Hyperplanes and crossings}
For a hyperplane $h$ we say that point $p \in h \cap \pi$ is a \emph{crossing} of
the tour $\pi = (s_1,\ldots,s_n,s_1)$ if there exist $i \in [n]$ such that $p \in
s_is_{i+1}$ (where $s_{n+1} = s_1$) and the endpoints $s_i$ and $s_{i+1}$ are
separated by $h$.

A \emph{grid hyperplane} is a point set of the form $\{ (x_1,\ldots,x_d) \in \mathbb{R}^d \mid
x_i = 1/2 + k)\}$ for some integer $i \in [d]$ and $k \in \mathbb{Z}$.  For a set of line segments, tour, or geometric graph $S$ we
define $I(S,h)$ (respectively, $I(\pi,h)$) the set of intersection points of the segments of $S$ and $h$. We remark that the tour $\pi$ may cross a given point of $h$ several times, but we still think of $I(\pi,h)$ as a set rather than a multiset. Similarly, we will often refer to the set of crossing points with a closed $(d-1)$-dimensional hypercube $F$ as $\pi\cap F$, and $|\pi \cap F|$ does not count the multiplicity of these crossings.

The following simple lemma relates the number of crossings with grid hyperplanes
with the total length of the line segments.

\begin{lemma}[c.f., Lemma 19.4.1 in~\cite{geomspannet}]\label{lem:crossingsvslength} If $S$ is a set of line segments with endpoints in $\mathbb{Z}^d$, then
    \begin{displaymath}
        \sum_{h \text{ is a grid hyperplane}} |I(S,h)| \le \sqrt{d} \cdot \wt(S).
    \end{displaymath}
\end{lemma}

The following folklore lemma is typically used to reduce the number of ways a
salesman tour can cross a given hyperplane:

\begin{lemma}[Patching Lemma~\cite{Arora98}]\label{lem:patch}
    Let $h$ be a hyperplane, $\pi$ be a tour or Steiner tree, and let $I(\pi,h)$ be the set
    of intersections of $\pi$ with $h$. Assume that $h$ does not contain any
    endpoints of segments that define $\pi$. Let $T$ be a tree on the hyperplane $h$ that
    spans $I(\pi,h)$. Then, for any point $p$ in $T$ there exist line segments
    contained in $h$ whose total length is at most $\Oh(\wt(T))$ and whose
    addition to $\pi$ changes it into a tour (resp. Steiner tree) $\pi'$ that crosses $h$ at
    most twice and only at~$p$.
\end{lemma}

We refer to \cite[Section 19.6]{geomspannet} for a proof of Lemma~\ref{lem:patch}. Note that in typical presentations, the resulting patched tour will contain two copies of $T$, each infinitesimally close to $h$ but outside it; in our variant we allow overlapping segments inside $h$ which allows the patching to happen in $h$. See also our definition of dissection-aligned multigraphs~\Cref{def:multigraph}.

\paragraph{Dissection and Quadtree.} Now, we introduce a commonly used hierarchy to decompose $\mathbb{R}^d$ that will be instrumental to guide our algorithm.
Pick $a_1,\ldots,a_d \in \{1,\ldots,L\}$ independently and uniformly at random
and define $\ba \coloneqq (a_1,\ldots,a_d)$. Consider the hypercube
\[
    \mathrm{C}(\ba) \coloneqq \bigtimes_{i=1}^d [-a_i + 1/2,\,2L-a_i+1/2].
\]
Note that $\mathrm{C}(\ba)$ has side length $2L$ and each point from $P$
is contained in $\mathrm{C}(\ba)$ by the assumption $P \subseteq \{0,\ldots,L\}^d$.

For a cutoff parameter $\mu \in \mathbb{Z}$ let the \emph{dissection} $\mathrm{D}(\ba)$ of $\mathrm{C}(\ba)$
to be a rooted tree that is recursively defined as follows. With each vertex of
the tree we associate a hypercube in $\mathbb{R}^d$. For the root this is
$\mathrm{C}(\ba)$ and for the leaves of the tree this is a hypercube of 
side length $2^\mu$. Typically we will have $\mu=0$
and unit side-length cubes associated with the leaves. Each non-leaf vertex $v$
of the tree with associated closed hypercube
$\bigtimes_{i=1}^d [l_i,u_i]$ has $2^d$ children with which we associate
$\bigtimes_{i=1}^d I_i$, where $I_i$ is either $[l_i,(l_i+u_i)/2]$ or
$[(l_i+u_i)/2,u_i]$. We refer to such a hypercube that is associated with a
vertex in the dissection as a \emph{cell} of the dissection. The level of a cell is the distance from the corresponding vertex to the root of the tree. 

The \emph{quadtree} $\mathrm{QT}(P,\ba)$ is obtained from
$\mathrm{D}(\ba)$ by terminating the subdivision whenever a cell has at
most $1$ point from the input point set $P$. This way, every cell is either a
leaf that contains $0$ or $1$ input points, or it is an internal vertex of the
tree with $2^d$ children, and the corresponding cell contains at least $2$ input
points. We say that a cell $C \in \mathrm{QT}(P,\ba)$ is \emph{redundant}
if it has a child that contains the same set of input points as the parent of
$C$. A redundant path is a maximal ancestor-descendant path in the tree whose
internal vertices are redundant. The \emph{compressed quadtree}
$\mathrm{CQT}(P,\ba)$ is obtained from $\mathrm{QT}(P,\ba)$ by
removing all the empty children of redundant cells, and replacing the redundant
paths with single edges. In the resulting tree some internal cells may have a single
child; we call these \emph{compressed cells}. It is well-known and easy to check
that compressed quadtrees have $\Oh(n)$ vertices (note that compressed quadtrees
can be computed even in $\Oh(n)$ time on a word RAM model~\cite{DBLP:journals/ipl/Chan08}).

For every face $F$ of a cell in $\mathrm{D}(\ba)$ there exists
a unique grid hyperplane that contains $F$. For a grid hyperplane $h$ we define
the \emph{level of $h$} to be the smallest integer $i$ such that
$\mathrm{D}(\ba)$ contains a cell with sides of length $2L/2^i$, one of
whose faces is contained in $h$.

We say that two distinct cells of a dissection or quadtree with the same side-length are \emph{neighboring} if they share a facet, and they are \emph{siblings} if they also have the same parent cell.

\begin{lemma}[Lemma 19.4.3~\cite{geomspannet}]\label{levelprob}
	Let $h$ be a grid hyperplane, and let $i$ be an integer satisfying $0 \leq i
    \le 1+ \log L$. Then the probability that the level of $h$ is equal to $i$ is at most $2^{i-1}/L$.
\end{lemma}

\paragraph{Building blocks of Arora's technique}
We now briefly describe the building blocks from~\cite{Arora98} that we will use.

\begin{definition}[$m$-regular set]
    An $m$-regular set of portals on a $d$-dimensional hypercube $\mathrm{C}$ is an
    orthogonal lattice $\grid(\mathrm{C},m)$ of $m$ points in the cube. If the cube
    has side length $\ell$, then the spacing between the portals is set to
    $\ell/(m^{1/d}-1)$.
\end{definition}

We will normally have $m$ be chosen as $k^d$ for some integer $k\geq 2$, and as a consequence, $\grid(\mathrm{C},m)$ will always contain the corners of $\mathrm{C}$.

\begin{definition}[$r$-light]
	A set of line segments $S$ is $r$-light with respect to the dissection
    $\mathrm{D}(\ba)$ if it crosses each face of each cell of $\mathrm{D}(\ba)$ at most $r$ times.
\end{definition}

\begin{theorem}[Arora's Structure Theorem]\label{thm:arora}
    Let $P \subseteq \{0,\ldots,L\}^d$, and let $\wt(OPT)$ be the minimum length of a salesman tour visiting $P$. Let the shift vector $\ba$ be picked randomly.
	Then with probability at least $1/2$, there is a salesman tour of cost at
    most $(1 + \eps)\wt(OPT)$ that is $r$-light with respect to $\mathrm{D}(\ba)$ such
    that it crosses each facet $F$ of a cell of $\mathrm{D}(\ba)$ only at points
    from $\grid(F,m)$, for some $m = (\Oh((\sqrt{d}/\eps) \log L))^{d-1}$ and $r=(\Oh(\sqrt{d}/\eps))^{d-1}$.
\end{theorem}

\section{Structure Theorem} \label{sec:pfstruct}

Now we present and discuss the main structure theorem that allows us to prove the algorithmic part of
Theorem~\ref{thm:tsp}. We state the theorem for a general dimension $d$.

For a $(d-1)$-dimensional hypercube $F$ let $F^*$ denote $F$ without its $2^{d-1}$ corner points.

\begin{definition}[$r$-simple geometric graph]\label{def:simple}
	Let $\pi$ be a geometric graph in $\mathbb{R}^d$ such that the grid hyperplanes of the dissection $\mathrm{D}(\ba)$ do not contain any edge of $\pi$.  We say that $\pi$ is \emph{$r$-simple} if for every facet $F$ shared by a pair of sibling cells in $\mathrm{D}(\ba)$: 
	\begin{enumerate}[label=(\alph*)]
		\item $\pi$ crosses $F^*$ through at most one point (and some subset of $2^{d-1}$ corners of~$F$), or
        \item $\pi$ crosses $F$ only through the points from $\grid(F,g)$ for some $2^{d-1}\leq g \leq  r^{2d-2}/|\pi\cap F^*|$.
	\end{enumerate}
	Moreover, for any point $p$ on a hyperplane $h$ of $\mathrm{D}(\ba)$, $\pi$ crosses $h$ at most twice via $p$.
\end{definition}

One can see that in case (b) we always have $|\pi\cap F^*| \le g$ (or when including multiplicities, there are at most $2g$ crossings). Consequently, $|\pi\cap F^*|\leq r^{2d-2}/|\pi\cap F^*|$ and thus $|\pi\cap F^*|\leq r^{d-1}$. Moreover, we recall that $\grid(F,g)$ contains all corners of $F$.

\begin{definition}[$r$-simplification]\label{def:simplification}
We say that a geometric graph $\pi'$ is an \emph{$r$-simplification of $\pi$} if $\pi'$ is an
$r$-simple geometric graph, and in each facet $F$ where $|\pi\cap F^*|=1$, the single non-corner crossing is a point from $\pi \cap F^*$.
\end{definition}

\begin{theorem}[Structure Theorem]\label{thm:struct}
	Let $\ba$ be a random shift and let $\pi$ be a tour or Steiner tree of $P \subseteq \mathbb{R}^d$. Then for any positive integer $r$ there is a tour (resp., Steiner tree) $\pi'$ of $P$ that is an $r$-simplification of  $\pi$ such that
    \[\mathbb{E}_{\ba}[\wt(\pi')-\wt(\pi)] \leq \Oh(d^{5/2}\cdot \wt(\pi) / r).\]
\end{theorem}

Theorem~\ref{thm:struct} is the main contribution of this paper. 

\subsection{Sketch for the 2-dimensional case}

Before we present the proof of Theorem~\ref{thm:struct}, let us informally describe the construction of the tour $\pi'$ of Theorem~\ref{thm:struct} for $d=2$ and how it can be used give a $2^{\Oh(1/\eps)} n\polylog(n)$ algorithm for \textsc{Euclidean TSP} when $d=2$. The full proof of Theorem~\ref{thm:struct} will start at~\cref{sec:base-line-tree}.

\paragraph{Sketch of the algorithm}

We set $r = \Oh(1/\eps)$. If we find an $r$-simple tour of
lowest weight, then property (i) of Theorem~\ref{thm:struct} guarantees that
this tour is a $(1+\eps)$-approximation of an optimal salesman tour.  Similarly to
Arora~\cite{Arora98} we can use dynamic programming to find such a tour.  The
number of possible ways in which the tour can enter and leave a cell of the
quadtree is at most $\binom{\Oh(1/(\eps^2m))}{m}2^{\Oh(m)} \cdot \poly(n)$,\footnote{\label{catalan}This uses the well-known fact that the number of non-crossing matchings on $r$ endpoints is at most
	$2^{\Oh(r)}$.}
since there are at most $\poly(n)$ possibilities for the location of crossing if
there is at most one crossing. The number of table entries can then be upper bounded with
$2^{\Oh(1/\eps)}\poly(n)$ via the following claim:

\begin{claim}
	\label{portal-bin-ineq}
    For every $1\leq a \leq b$, it holds that $\binom{b/a}{a} \leq e^{\sqrt{b/e}}$.
\end{claim}
\begin{claimproof}
    If $a > \sqrt{b}$, then $\binom{b/a}{a} = 0$ and the
	inequality follows. If $a \le \sqrt{b}$, then by the standard upper bound
	$\binom{n}{k} \leq (\frac{n\cdot e}{k})^k$ we have that $\binom{b/a}{a} \leq
	\left(\frac{b\cdot e}{a^2}\right)^{a}$. In the interval $a\in [1,\sqrt{b}]$, the latter expression is maximized
    for $a=\sqrt{b/e}$, where it equals $e^{\sqrt{b/e}}$.
\end{claimproof}

To get the $\poly(n)$ factor in the running time down to $\polylog(n)$, note
that we can first apply Theorem~\ref{thm:arora} with smaller $\eps$ to ensure
there are $\log^{\Oh(1)}(n)$ possibilities for the case where the tour crosses a
cell edge at a single point.

\paragraph{Sketch of the patching}
The proof of Theorem~\ref{thm:struct} uses a so-called \emph{patching procedure}
that modifies an (optimal) tour to a tour with the desired properties, but
without increasing the length by too much. Here, we sketch the procedure for the
case of $d=2$, hence for simplicity assume that $h$ is a horizontal
line, and $c_1< c_2 < \ldots < c_k$ are the $x$-coordinates of the $k=|I(G,h)|$
crossings.
We define the \emph{proximity} of the $j$-th crossing as $\pro(c_j) =
c_{j}-c_{j-1}$ (for $j=1$, use $c_{0}=-\infty$).\footnote{In the later formal proof that also handles $d>2$ we use a more complicated version of \emph{proximity} defined in terms of the \emph{base-line tree}.}

Our Sparsity-Sensitive Patching considers each cell $C$
of the dissection and each side $F$ of $C$ with at least two crossings, and
connects each crossing $x$ on $F$ as follows (see Figure~\ref{fig:constructiontc}).
\begin{enumerate}
    \itemsep0em 
    \item Let $N$ be the set of ``\emph{near}'' crossings, that is, $N$ is the set of
        crossings of $\pi$ and $F$ satisfying $\pro(x) \leq \frac{L}{2^i r}$,
        where $i$ is the level of the line of $F$ in the dissection.
    \item Let $G$ be the set of remaining crossings of $\pi$ with $F$.
	\item Create a set of line segments $\PT$ (the \emph{patching forest} of $F$) by connecting each vertex from
    $N$ to its successor and, if $|G| > 1$, connecting each vertex from $G$ to
    the closest point in $\grid(F,r^2/|G|)$.
	\item Apply Lemma~\ref{lem:patch} to each set of touching line segments of
        $\PT$ to obtain a new tour $\pi'$ that crosses $F$ only at $|G|$ points
        of $\grid(F,r^2/|G|)$, and at most twice at each of these points.
\end{enumerate}
\begin{figure}[t!]
    \centering
    \includegraphics[width=0.8\textwidth]{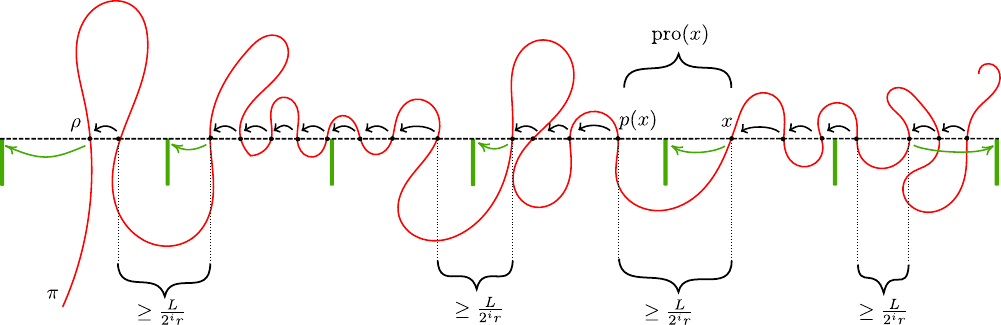}
    \caption{Construction of a set of line segments $\PT$ in $d=2$. The tour $\pi$ is colored red. Green portals denote the
        points in $\grid(F,r^2/|G|)$. The leftmost point and the points with
    $\pro(x) > L/(2^i/r)$ form the set $G$ and are connected to the closest portal from the grid by a green arrow. Points with $\pro(x)
    \le L/(2^ir)$ form the set $N$ and they are connected to their parent with black arrows. The
    set of line segments $\PT$ is indicated with a collection of
    black and green arrows.}
    \label{fig:constructiontc}
\end{figure}

\begin{remark}
    One difference between Arora's and our patching procedure is that each line (or in higher dimensions, each hyperplane) is patched only once. In other words, we only do patching between neighboring sibling cells. Arora's algorithm uses bottom-up patching, that is, it first patches along the shared boundary of neighboring leaf cells of the quadtree. (These leaves need not be siblings.) The procedure then goes up one level, and a patching may happen again if the number of crossings in the new (and larger) facet exceeds some threshold. Thus, Arora's patching procedure is iterative, and several patching steps may occur on any given cell boundary. In contrast, our patching is not iterative and it is done independently in each hyperplane (or on each cell boundary) only once.
\end{remark}

The rest of this section is dedicated to the proof of the Structure Theorem in
$d$-dimensional Euclidean space. Before we prove it, we first show the existence
of a certain \emph{base-line tree} in $d$-dimensional Euclidean space. This tree
will be a subset of a hyperplane, and parts of it will be used via the
invocation of the patching routine from Lemma~\ref{lem:patch} to reduce the
number of crossings of the tour with the hyperplane. In $\mathbb{R}^2$, this
tree is just an entire line segment and the construction of the base-line tree in Subsection~\ref{sec:base-line-tree} and the analysis of using it for patching the tour in Subsection~\ref{sec:struct-proof} can be skipped over by readers only interested in a proof sketch for $d=2$. Trees similar to the base-line tree were also used for the case $d>2$ by a previous algorithm
(see~\cite{RaoS98,geomspannet}), but we need a more delicate construction. Crucially, our
base-line tree determines the proximities (i.e., the amortized patching costs) of the crossings and whether a given
crossing point will be connected to a point from a grid or not. It will also play a crucial role in avoiding a large number of new crossings in perpendicular hyperplanes that could arise as a result of patching.

\subsection{The base-line tree}
\label{sec:base-line-tree}

In order to construct a good base-line tree we will need to align it with a fixed dissection.

\begin{definition}[Dissection-aligned geometric graph]\label{def:multigraph}
    A graph $G$ is a
    $\mathrm{D}(\ba)$-aligned geometric graph if the following hold:
    \begin{itemize}
        \item the vertices and edges of $G$ are represented by a multiset of
            points and segments in $\Reals^d$, respectively,
        \item each vertex of $G$ is assigned to a unique cell containing the point representing this vertex,
        \item every edge of positive length connects two vertices of the same cell, and
        \item every edge of length $0$ is connecting a pair of vertices that are assigned to a cell and its child cell, respectively.
    \end{itemize}
\end{definition}

We say that an edge of a dissection-aligned geometric multigraph forms a \emph{crossing} of a cell $C$ if it has one end vertex assigned to $C$ and the other end-vertex assigned to the parent of $C$. Note that by the above definition, a crossing edge is always of length $0$, and its geometric location is at the location of the edge.

The following lemma is based on \cite[Lemma 19.5.1]{geomspannet}, but there are
two important differences.  First, we do not need an efficient
construction and only need to prove the existence of such a tree $T$.  Second,
\cite[Lemma 19.5.1]{geomspannet} does not guarantee our Property (c) in
Lemma~\ref{lem:lighttree}.

\begin{lemma}\label{lem:lighttree}
    Let $d\geq 1$ be a constant and let $K \subseteq \mathbb{R}^d$ and 
    let $\mathrm{D}(\ba)$ be a dissection in $\mathbb{R}^d$ where each bottom-level cell contains at most two vertices of $K$.
    Then there exists
    a rooted tree $T$ spanning $K$ that is a $\mathrm{D}(\ba)$-aligned multigraph with the following properties.
	\begin{enumerate}[label=(\alph*)]
    \item $T$ is $1$-light, i.e., each cell has at most one crossing edge, which leads to a parent cell.
    \item Each cell $C$ has its crossing vertex at a designated corner $\ori(C)$ of $C$, and
    \item For each cell $C$ of $\mathrm{D}(\ba)$ with
        side length $\ell$ and $Q \subseteq K \cap C$, it holds that the minimum subtree $T'$
        of $T$ that spans $Q$ satisfies $\wt(T') \le 8 d \ell |Q|^{1-1/d}$.
	\end{enumerate}
\end{lemma}
\begin{proof}

We construct a $\mathrm{D}(\ba)$-aligned geometric tree as follows.
The \emph{skeleton} of a cell $C$ of $\mathrm{D}(\ba)$ is the graph whose
vertex set consists of the $2^d$ corners of $C$ and whose edge set consists of
of all $d2^{d-1}$ edges of the hypercube $C$ (i.e. each pair of corners of $C$ that differ in one coordinate shares an edge). If $C$ is the top-level cell of the dissection, then we define $\ori(C)$ to be the corner of
$C$ at which all coordinates are the smallest possible, that is, the lexicographically minimum corner. If $C$ is the child cell of $C^*$, then let $\ori(C)$ be the unique vertex of $C$ that is also a vertex of $C^*$.
We construct a tree $T_0$
that only crosses each cell~$C$ of the dissection at~$\ori(C)$. To do so, for each
cell $C$ we add a spanning tree of the skeleton of~$C$ rooted at~$\ori(C)$ with
depth at most~$d$; this tree is denoted by~$T_C$, see Figure~\ref{fig:baseline}. Such a tree could be constructed for example by breadth-first-search on the edges of the hypercube; note that the diameter of the edge graph is~$d$ thus the resulting BFS tree has depth~$d$. 
Observe that trees~$T_C$ in different cells~$C$ have some shared vertices, but no shared edges, although edges can have overlaps: if~$C'$ is a child of~$C$, then $C'_\tau$~and~$C_\tau$ will have edges that are intervals of the same line. For each parent and child cell pair $C$~and~$C'$, let~$e_{CC'}$ be an edge of length $0$ connecting the vertex of~$C'$ located at~$\ori(C')$ to the vertex of~$C$ located at the same place~$\ori(C')$. Now let~$T_0$ be the union of all the trees~$T_C$ (with vertices assigned to $C$) and edges~$e_{CC'}$ for each cell~$C$ and for each parent-child cell pair~$C,C'$ of~$\mathrm{D}(\ba)$, respectively.

\begin{figure}
\centering
\includegraphics{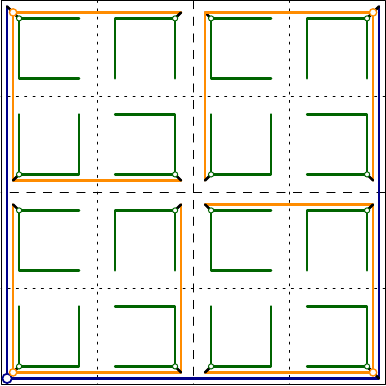}
\caption{Constructing the tree $T_0$. The spanning tree of the skeleton of a cell $C$ is connected to a vertex of its parent's spanning tree with a length-0 edge at $\ori(C)$ (denoted by a circles). The spanning trees at levels 1,2 and 3 are drawn in blue, orange and green, respectively.}\label{fig:baseline}
\end{figure}

Now let $T$ be the tree $T_0$ plus the edge that connects each point $x \in K$ to the vertex of $C$ at $\ori(C)$ where $C$ is the bottom-level cell of the dissection (of side length $2^\mu$) that contains $x$. Naturally, the vertex at $x$ is assigned to $C$. The
tree $T$ remains $\mathrm{D}(\ba)$-aligned and $1$-light. It
is thus enough to show that it satisfies property $(c)$.

To prove (c) we introduce an auxiliary tree $DT$ as follows. The set of
vertices of $DT$ is the multiset defined by $\{\ori(C) \mid C \text{ is a cell
of } \mathrm{D}(\ba)\} \cup K$. Recall that a cell $C$ of the dissection
has a well-defined level $\text{lvl}_C \in \{1,\ldots,\ell\}$.  For each vertex
$v \in V(DT)$ we say that $v$ has level $\ell+1$ if $v \in K$ or has level $i$
if the corresponding cell $C$ has level $i$. Note that when $C$ has level $\ell$, then the only vertex of $T_0$ that was assigned to $C$ and has been added to $DT$ is $\ori(C)$.

Finally, we add edges between
$\ori(C)$ and $\ori(C')$ of weight $2\cdot Ld/2^i$ if $C'$ is a child of $C$. Moreover,
each vertex $x \in K$ is connected to $\ori(C)$ with an edge of weight
$2\cdot Ld/2^{\ell+1}$, where $C$ is a cell at the
bottom level that contains $x$. Notice that when $C$ is at the bottom level, then $|K\cap C|\leq 2$ implies that $\ori(C)$ has at most $2\leq 2^d$ children in $DT$.

Now we show that the length of each subtree of $DT$ is at least the length
of the corresponding subtree (i.e., the subtree spanned by the same vertex set) of $T$. 

Note that in $DT$ the weight of an edge between $\ori(C)$ and $\ori(C')$ from level $i$ to level $i-1$ is exactly $2\cdot Ld/2^i$.
Consider the corresponding path between these vertices in $T$. 
Notice that
$\ori(C)$ and $\ori(C')$ are both in the skeleton of $C$ and their distance in
$T$ is at most $Ld/2^i$. For vertices $x \in K$, the distance between $x$
and $\ori(C)$, where $C$ is the bottom-level cell containing $x$ is at most
$\sqrt{d}\cdot L/2^\ell < 2\cdot Ld/2^{\ell+1}$.

The lemma is now a consequence of applying the following claim on general
weighted trees to the subtree $T'$ of~$DT$ rooted at~$\ori(C)$.
	
\begin{claim}\label{cl:smallsubtree}
    Let $T'$ be a rooted tree in which each vertex has at most
    $2^d$ children and each edge from level $i-1$ to level $i$ has weight at most $1/2^i$. Then for any set
    of vertices $Q$, the minimum subtree of $T'$ that spans $Q$ has weight at most $4\cdot |Q|^{1-1/d}$.
\end{claim}
\begin{claimproof}
	Let $k$ be the integer such that 
	\[
		2^k \leq |Q|^{1/d} < 2^{k+1}.
	\]
	From level $i-1$ to level $i$ we have at most $2^{di}$ edges and each such
    edge has weight $1/2^{i}$. We have that the total weight of all edges from
    level $0$ to $k$ is at most:
	\[
		\sum_{i=1}^{k} \frac{2^{di}}{2^{i}} \leq 2\cdot 2^{k(d-1)} \leq 2 \cdot (|Q|^{1/d})^{d-1} = 2 \cdot |Q|^{1-1/d}.
	\]
    On the other hand, the length of a path from a vertex $q \in Q$ that has a
    level at least $k$ to its ancestor at level $k$ is at most
    $\sum_{i=k+1}^{\infty} 1/2^{i}= 1/2^k$. Thus, the total length of all such
    paths is at most $|Q|/2^k < 2 |Q|^{1-1/d}$. Therefore, the weight of the
    subtree is less than $4|Q|^{1-1/d}$ in total.
\end{claimproof}

This concludes the proof of Lemma~\ref{lem:lighttree}.
\end{proof}

With Lemma~\ref{lem:lighttree} in hand, we are ready to start the proof of
Theorem~\ref{thm:struct}. We start with describing the desired traveling
salesman tour $\pi'$.

\subsection{Constructing the patched tour \texorpdfstring{$\pi'$}{pi'} and analyzing its crossings}
\label{sec:struct-proof}

In the proof of our structure theorem (\Cref{thm:struct}), we may assume without loss of generality that $r\geq 128d$, as otherwise the claim can be satisfied by any known constant-approximation. We will call $\pi$ a tour; the proof is analogous when $\pi$ is a Steiner tree.

We construct the tour $\pi'$ by iteratively processing all crossings per grid
hyperplane.  For a grid hyperplane $h$ let $I(\pi,h)$ be the set of
intersections of $\pi$ with $h$.

Now fix a grid hyperplane $h$ and suppose that $h$ fixes the $j$-th coordinate
(so $h = \{(x_1,\ldots,x_d) : x_j = 1/2+z_h \}$ for some integer $z_h$).
Without loss of generality, we assume that $j=1$. Let $a_1$ be the first
coordinate of $\ba$.  Let $\mathbf{a'}$ be obtained from $\ba$ by
omitting the first coordinate.  Therefore $\ba = (a_1,\ba')$.
 
\begin{remark}\label{rem:dep_i}
    The level of a hyperplane perpendicular to the first axis depends only\footnote{The orthogonality expressed in Remarks~\ref{rem:dep_i} and~\ref{rem:depT} are crucial to the success of our analysis. This orthogonality will be instrumental to the swapping of sums in \eqref{eq:sumswap}. We observe that the same orthogonality is used in an analogous manner in the proof of Arora's structure theorem: it is required to be able to swap two similar sums.
    
The iterative patching of Arora raises a related problem.
Consider the variable $c_{\ell,j}$ in \cite{Arora98}, and the sentence stating that $c_{\ell,j}$ is independent of $i$ above formula (3) on page 766 of~\cite{Arora98} (due to orthogonality). Strictly speaking this is not true: $c_{\ell,j}$ is undefined when $j<i$ (or it could be set to $0$). Thus in lines 6 and 7 of page 767 of~\cite{Arora98} the swapping of the sums over $i$ and $j$ is formally incorrect.
This inaccuracy can be easily fixed by setting $c_{\ell,j}$ based on the case $i=1$, thus making it independent from $i$, and treating $c_{\ell,j}$ as an upper bound on the true crossing count $c_{\ell,i,j}$ at stage $j\geq i$, where $c_{\ell,i,j}$ is set to $0$ when $j<i$. We note furthermore that the sums are swapped without the mention of orthogonality on page~770 of~\cite{Arora98}. The same inaccuracy appears in almost all published proofs of Arora's structure theorem: see page 460 of~\cite{geomspannet}, page 268 of~\cite{williamson2011design}.    
    The only full analysis that the authors are aware of which manages to avoid this matter is by Har-Peled~\cite{Har-Peled11} for $d=2$.} on $a_1$.
\end{remark}

Suppose that $h$ has level $i$, and let us fix a facet $F$ of a cell at
level $i$ in $\mathrm{D}(\ba)$ (with cutoff $\mu=0$) where $F\subset h$. Note that $F$ is a $(d-1)$-dimensional hypercube, so $F$ is
actually a \emph{cell} in the dissection $\mathrm{D}(\mathbf{a'})$. The side length of $F$ is $L/2^i$.
Next, we will change $\pi$ so that the resulting tour satisfies Definition~\ref{def:simple} on $F$:
If the tour already satisfies (a) we do not have to do anything, so let us
assume for now that it does not satisfy (a) for the facet $F$.

We apply the $(d-1)$-dimensional version of Lemma~\ref{lem:lighttree} to the
set of crossings $I(\pi,h) \subseteq h$ with dissection
$\mathrm{D}(\mathbf{a'})$ with cutoff $\nu$ where $\nu$ is the largest integer such that $2^\nu < \frac{1}{r^{2d-2}}$ and all vertices of $I(\pi,h)$ fall in different cells of $\mathrm{D}(\mathbf{a'})$. (We reiterate that during this proof of the structure theorem is not intended to be algorithmically efficient.)
We obtain a rooted tree $T_0$ that spans $I(\pi,h)$ and that is $1$-light with
respect to $\mathrm{D}(\mathbf{a'})$. Let $\rho$ denote the root vertex of $T_0$. Let $T$ be the tree whose vertices are the leaves and branching
points of $T_0$, and its edges are the maximal paths of $T_0$ whose internal
vertices have degree $2$. (That is, the drawing of $T$ and $T_0$ consists of the
same set of segments.) Let $X$ be the set of vertices in $T$; note that
$X\supset I(\pi,h)$ and $|X|\leq 2|I(\pi,h)|-2$. We orient every edge in $T$
away from $\rho$. Hence $\rho$ is an ancestor of every vertex in $X$ and leaves
of $T$ have only themselves as descendants.

\begin{remark}\label{rem:depT}
    The tree $T$ and the set $X$ depend only on $\ba'$ and $I(\pi,h)$.
\end{remark}

For a point $x\in X$ let $\cdc(x)$ denote the \emph{closest descendant crossing} of $x$, that is, the closest vertex among $I(\pi,h)$ in the subtree of $x$ in $T$, where the distance is measured along the edges of $T$. In particular, if $x\in I(\pi,h)$ then $\cdc(x)=x$. Notice moreover that the leaves of $T$ are all vertices of $I(\pi,h)$ and therefore $\cdc(.)$ assigns leaves to themselves. See Figure~\ref{fig:cdcpro}.
		
	\paragraph{Construction of $\PT$.}
For a point $x\in I(\pi,h)\setminus\{\cdc(\rho)\}$ we define the \emph{proximity} of $x$, denoted by $\pro_{\ba'}(x)$, to be the sum of arc lengths in $T$ whose targets are in $\mathrm{cdc}^{-1}(x)$. We set $\pro_{\ba'}(\cdc(\rho))=\infty$. We note that $\pro_{\ba'}$ is positive, and we will occasionally use $1/\pro_{\ba'}(\cdc(\rho))=0$. By \Cref{rem:depT} we have that $\pro_{\ba'}(x)$ depends only on $\ba'$ and it is independent of $a_1$.
	Let $T_F$ be the subtree of $T$ induced by $F\cap X$, and let $\rho_F$
	denote the root of $T_F$. We remark that we are using the $1$-lightness of $T$ here: it guarantees that $F\cap X$ induces a subtree in $T$ which is \emph{contained} in $F$. We define $G\subset I(\pi,h)\cap F$ to include a set of (intuitively distant)
    vertices whose proximity is large: we add $\text{cdc}(\rho_F)$ to $G$ as well as the vertices $x \in I(\pi,h)
    \cap F \setminus \{\text{cdc}(\rho_F)\}$ such that $\pro_{\ba'}(x)>L /(2^ir)$. Let $G'\coloneqq G \setminus \{\cdc(\rho_F)\}$.

\begin{figure}[t]
\centering
\includegraphics{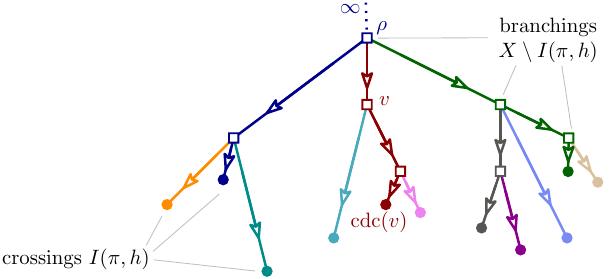}
\caption{Closest descendant crossings ($\cdc(.)$) and proximity in $T$, where crossing nodes are denoted by disks. The closest descendant crossing of each node $x\in X$ (be it a crossing or a branching) is the crossing node of the same color. The proximity of a crossing node $v\in I(\pi,h)$ is the total length of the tree path of the same color, i.e., the total length of the arcs entering $\cdc^{-1}(v)$, except for the dark blue crossing $\cdc(\rho)$ whose proximity is $\infty$.}\label{fig:cdcpro}
\end{figure}
  
  \begin{lemma}\label{lem:pro_sum}
  Let $T^G_F$ be the minimum subtree of $T$ spanned by $G$. Then $\sum_{z \in G'}\pro_{\ba'}(z) \leq \wt(T^G_F)$.
  \end{lemma}
  \begin{proof}  
        For $z \in G$ the arcs ending in $\cdc^{-1}(z)$ form a directed path inside $T_F$; let $P_z$ be this path. Notice that these paths are edge-disjoint.
     Let $\rho^G_F$ denote the root of $T^G_F$. The path $P_{\cdc(\rho_F)}$ contains the arc with target $\rho_F$ as well as the arc with target $\rho^G_F$. By the disjointness of the paths $P_z$ for $z\in G$, we have that every other path stays in $F$ and under $\rho^G_F$. Consequently, each path $P_z$ for $z\in G'$ is contained in $T^G_F$. Thus we have that $\sum_{z \in G'}\pro_{\ba'}(z) = \sum_{z \in G'} \wt(P_z)\leq \wt(T^G_F)$.
       \end{proof}    
    
    If
    $|G|=1$, that is, when $G=\{\cdc(\rho_F)\}$,
	we set $\PT^*=T_F$. When $|G|\geq 2$, then
	we change $T_F$ to get a forest $\PT^*$ as follows.
    We say that an arc $(u,v)$ of $T_F$, where $v$ is a child of $u$, is \emph{bad} if $\cdc(u) \neq
    \cdc(v)$ and $\cdc(v) \in G$.
    Delete every bad edge from $T_F$, and for any rooted tree in the remaining forest, iteratively remove a root if it has degree one and it is not in $I(\pi,h)$.
    The resulting forest $\PT^*$ has exactly $|G|$ connected components, and for each component the root is either in $I(\pi,h)$ or it has at least two children.
    Finally, we will connect each vertex of $G$ to a point of a grid as follows.
    To define the grid,
	let $q$ be a positive integer such that $(q-1)^{d-1} < (r/2)^{2d-2}/|G|
	\leq q^{d-1}$, and let $g=q^{d-1}$ (since $r \ge 128d$ such $q$ exists).
	Thus 	
    \begin{equation}\label{eq:gbound}
        \frac{(r/2)^{2d-2}}{|G|} \leq g < \frac{r^{2d-2}}{|G|},
    \end{equation}
    where the second inequality follows, since we assumed $r \ge 128d$.
    
\begin{claim}\label{cl:gGbound}
$|G|<(16dr)^{d-1}$ and $g> \left(\frac{r}{64d}\right)^{d-1}$, in particular, $g> 2^{d-1}$.
\end{claim}
\begin{claimproof}
Recall from \Cref{lem:pro_sum}
that $T^G_F$ is the subtree of $T$ spanned by $G$. By \Cref{lem:lighttree}(c) in $T$ with $Q=G$ gives $\wt(T^G_F)\leq \frac{8dL}{2^i}|G|^{1-1/(d-1)}$. \Cref{lem:pro_sum} implies that
$\sum_{z \in G'}\pro_{\ba'}(z) \leq \frac{8dL}{2^i}|G|^{1-1/(d-1)}$. By the definition of $G$, we have that for any $z\in G'$ the proximity $\pro_{\ba'}(z)>\frac{L}{2^i r}$. Thus we get
\begin{equation}\label{eq:simpleGbound}
(|G|-1)\frac{L}{2^i r} < \frac{8dL}{2^i}|G|^{1-1/(d-1)}.
\end{equation}
For the second inequality, if $|G|=1$ then \eqref{eq:gbound} implies that $g\geq (r/2)^{2d-2}>\left(\frac{r}{64d}\right)^{d-1}$ as $r\geq 8$.
Assuming $|G|\geq 2$, we have $2(|G|-1)\geq |G|$, thus after simplifying \eqref{eq:simpleGbound} we get 
\[|G|^{1/(d-1)}< 16dr \quad \Rightarrow \quad |G|< (16dr)^{d-1}.\]
Using \eqref{eq:gbound} we get $g>\frac{(r/2)^{2d-2}}{(16dr)^{d-1}}=\left(\frac{r}{64d}\right)^{d-1}$.
Since $r\geq 128d$, this directly implies $g>2^{d-1}$.
\end{claimproof}    
    
\begin{figure}[t]
\centering
\includegraphics{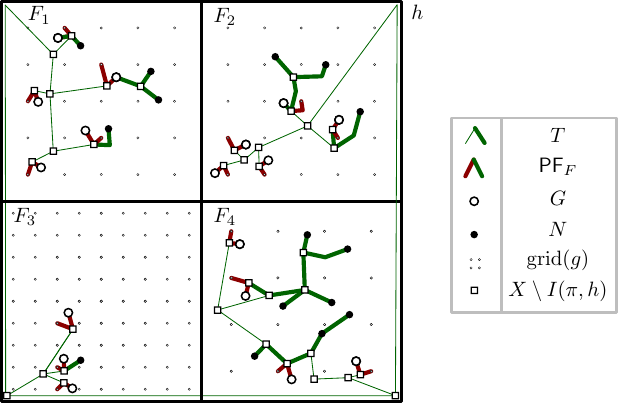}
\caption{Construction of the forests $\PT$ in four faces in a plane $h$ of level
$1$ in the quadtree. The green (thin and thick) edges are a schematic picture of
$T$ (note that the actual edges consist only of axis-parallel segments), and the
thick (red and green) edges indicate the forests $\PT$.}\label{fig:3dpatching}
\end{figure}

     Consider
    now the grid $\grid(F,g)$ in $F$. Let $T_{\grid(F,g)}$ be the subtree of $T$
    induced by all of its vertices in $F$ plus a new edge from each point $x$ of
    $\grid(F,g)$ to $\ori(C_x)\in V(T_{\grid(F,g)})$ where $C_x$ is the bottom
    level cell of side length $2^\nu$ containing $x$. Observe that our setting of the cutoff $\nu$ ensures that each cell of $\mathrm{D}(\ba')$ contains at most two vertices from $\grid(F,g)\cup I(\pi,h)$. Notice that
    $T_{\grid(F,g)}$ is a supergraph of $T_F$ and it is 1-light
    with respect to $\mathrm{D}(\ba)$, as it is the tree one gets by applying Lemma~\ref{lem:lighttree} for the set $I(\pi,h)\cup\grid(F,g)$ on the subdivision $\mathrm{D}(\ba')$.
    
    For a vertex $x\in G$ let $C_x$ be the smallest cell of $\mathrm{D}(\ba')$ that contains $x$ and at least one point $y$ from $\grid(F,g)$. We connect $x$ and $y$ along the tree path in $T_{\grid(F,g)}$. Let $\PT$ be the resulting forest.     See Figure~\ref{fig:3dpatching} for a schematic illustration of the
    construction for $d=3$.

\begin{lemma}\label{lem:gridconnect_cost}
The length of the path connecting $x\in G$ to the point $y$ in $\grid(F,g)$
along $T_{\grid(F,g)}$ is less than $\frac{128dL|G|^{1/(d-1)}}{r^2 2^i}$.
\end{lemma}    

\begin{proof}
For a vertex $x\in G$ consider the cell $C_x$ that is the smallest cell of $\mathrm{D}(\ba')$ that contains $x$ and at least one point $y$ from $\grid(F,g)$. Recall that $\grid(F,g)$ is a grid with minimum point distance $s:=\frac{L}{(g^{1/(d-1)}-1)2^i}$. 
Then the side length of $C_x$ is at most $2s$. \Cref{cl:gGbound} implies that $g>2^{d-1}$, thus $g^{1/(d-1)}-1\geq g^{1/(d-1)}/2$, so $s\leq \frac{2L}{g^{1/(d-1)}2^i}$.  By \Cref{lem:lighttree} with $Q=(x,y)$ in the tree $T_{\grid(F,g)}$, we have that the distance of $x$ and $y$ in $T_{\grid(F,g)}$ is at most 
\[8ds\cdot2^{1-1/d}<\frac{32dL}{g^{1/(d-1)}2^i}\leq \frac{128dL|G|^{1/(d-1)}}{r^2 2^i},\] by the first inequality in~\eqref{eq:gbound}. 
\end{proof}

    \paragraph{Patching along $\PT$.} Next, we change the salesman tour $\pi$ by
    applying Lemma~\ref{lem:patch} on each connected component of $\PT$
    that has a vertex from $I(\pi, h)$. 
    This restricts the tour to cross
    the hyperplane $h$ only at some vertex of $I(\pi,h)\cap F$ when
    $G=\{\cdc(\rho_F)\}$ or in at most $|G|$ points from $\grid(F,g)$.
    Additionally, if there is a point $p\in h$ and the patched tour
    crosses $h$ more than twice at $p$, we apply Lemma~\ref{lem:patch}
    with $X=T=\{p\}$ and reduce the number of crossings at $p$ to at most two
    without increasing the length of the tour. This finishes the description of
    the construction of $\pi'$ promised by the theorem. We note that the cost
    of the patching in $F$ is at most $\Oh(\wt(\PT))$ by
    Lemma~\ref{lem:patch}. After processing $F$ the resulting tour restricted to $F$ is a $\mathrm{D}(\ba')$-aligned geometric multigraph, as each vertex in the relative interior of $F$ is assigned to the level of $F$, and we
    can add length-$0$ edges at crossings of the relative boundary of $F$.

    Observe that this patching can introduce new crossings in hyperplanes
    perpendicular to $h$. Let $C$ be a cell with face $F$. New crossings can
    only be introduced between two descendants of $C$ that are incident to $F$.
    Recall that the patching of $F$ occurs along some subforest of the $\mathrm{D}(\ba')$-aligned base-line tree, so in particular we have the following.
    \begin{claim}
        Let $C_1,C_2$ be a pair of $d$-dimensional sibling cells in $\mathrm{D}(\ba)$ that are descendants of $C$ and there
        is a new crossing introduced between them by $\PT$. Then, this crossing
        is in a shared corner of $C_1$ and $C_2$.
    \end{claim}
    \begin{claimproof}
        Let $x$ be a point of $T_{\grid(F,g)}$ that gives an intersection between $C_1$ and $C_2$ on the shared boundary hyperplane $h'$ of $C_1$ and $C_2$. Observe that the level of $h'$ is more than $i$, therefore $C_1$ and $C_2$ also have level more than $i$, and they are thus descendants of $C$.
        The point $x\in F\cap h'$ is a point of $\PT$, thus
        \Cref{lem:lighttree}(b) implies that $x$ is a shared corner of the cells
        $C_1\cap F$ and $C_2\cap F$ of the dissection $\mathrm{D}(\ba')$. This corner is also a shared corner of $C_1$ and $C_2$.        
    \end{claimproof}

To construct the patched tour $\pi'$, apply the above patching on $\pi$ in each
hyperplane of the dissection $\mathrm{D}(\ba)$. We emphasize that the patching is not iterative, it is always applied on the original tour $\pi$.

We note that $\pi'$ may in fact have edges within grid hyperplanes due to the patching. We make infinitesimal perturbations to $\pi'$ to ensure that any positive-length edge assigned to a cell is shifted to the interior of the cell. As a result, the intersections of $\pi'$ with any grid hyperplane is a set of points.

    To see that the obtained tour $\pi'$ is an $r$-simplification of $\pi$,
    note that if $\pi'$ crosses a facet $F$ outside a grid\footnote{Recall that the corners of $F$ are included in grids of each possible granularity.}, then $|G|=1$
    and the crossing of $\pi$ must have happened at some point of
    $I(\pi,h)\cap F=\pi\cap F$.
    Otherwise $\pi'$ crosses in at most $|G|$ non-corner points from $\grid(F,g)$.
    The inequality chain \eqref{eq:gbound} gives $g < r^{2d-2}/|G|$,  thus $g < r^{2d-2}/|G| \leq r^{2d-2}/|\pi'\cap F^*|$. Finally, if some point of $\pi'\cap F^*$ appears more than $3$ times along $\pi$, then we can use~Lemma~\ref{lem:patch} at this single point to get a new tour of the same length whose multiplicity at this point is reduced to at most $2$. Thus $\pi'$ is an $r$-simplification of $\pi$.

\subsection{Analysis of the expected length of \texorpdfstring{$\pi'$}{pi'}}

	Let $\cost(h)$ denote the increase of the salesman tour during the iteration corresponding to the hyperplane $h$.
    Our main effort will lie in proving that $\Ex{\cost(h)} \leq
    \Oh(\frac{d^2}{r} \cdot  |I(\pi,h)|)$. This would be sufficient to prove the theorem since it allows us to conclude that
	\begin{equation}\label{eq:OPTtogridhyp}
        \sum_{h\text{:grid hyperplane}} \Ex{\cost(h)} \leq
        \sum_{h\text{:grid hyperplane}} 
        \Oh\left(d^2 \cdot \frac{|I(\pi,h)|}{r}\right)
= \Oh(d^{5/2}\cdot \wt(\pi)/r),
	\end{equation}
	where the second inequality is by Lemma~\ref{lem:crossingsvslength}.

We note here that in case of $d=2$ the following analysis can be simplified. First, the patching forest $\PT$ is a line segment, and the proximity of a point is $\pro_{\ba'}$ is independent of the shift, i.e., one could omit the subscript $a'$. The set $X$ is equal to the set $I(\pi,h)$ of crossings in $h$. Finally, the function $\cdc(.)$ is the identity function.

\paragraph{Amortized patching costs.}

By \Cref{rem:dep_i}, we have that the level $i=i_{a_1}$ of $h$ depends only on $a_1$ and it is independent of $\ba'$. With each $x \in
    I(\pi,h)$ we associate the following coefficients $\alpha_{i,\ba'}(x)$ that
    represent the \emph{amortized expected patching cost} due to $x$ if the level of $h$ is $i$:
	\[
		\alpha_{i,\ba'}(x) =
		\begin{cases}
            \displaystyle \pro_{\ba'}(x) \cdot 2^i/L,	& \textbf{if } \pro_{\ba'}(x)\leq L/(2^ir),\\[0.75em]
            \displaystyle  \frac{d^2 L}{\pro_{\ba'}(x) \cdot 2^i r^2},	& \textbf{if } L / (2^ir) < \pro_{\ba'}(x) \leq L/2^i,\\[1.25em]
			\displaystyle 0,	& \textbf{if }  L / 2^i < \pro_{\ba'}(x).
		\end{cases}
	\]
	In case of $d=2$, we note that $\alpha_{i,\ba'}(x)$ is independent of $\ba'$ and could simply be written as $\alpha_i(x)$.
	
	Next, we will show that the expected cost of patching from
    Lemma~\ref{lem:patch} for a fixed  cell $F$ whose hyperplane has level $i$ is at most
    $\Oh(\frac{L}{2^i} \cdot \sum_{x
    \in F \cap X} \alpha_{i,\ba'}(x))$.
	
	\paragraph{The weight of $\PT$.}
	Now we use the special properties of $T$ guaranteed by Lemma~\ref{lem:lighttree} to show the following.
	\begin{lemma}
		It holds that:
        \[\wt(\PT) = \Oh\left( \frac{L}{2^i}\cdot \sum_{x \in F
        \cap I(\pi,h)}\alpha_{i,\ba'}(x)\right).\]
	\end{lemma}
	\begin{proof}
	Let $\rho_F$ be the root of $T_F$.
	Consider the following subset of $F \cap X$:  
	\begin{align*}
        N &\coloneqq \left\{ x \in F \cap I(\pi,h) : \pro_{\ba'}(x) \leq \frac{L}{2^ir}  \right\} \setminus\{\cdc(\rho_F)\}.
	\end{align*}

    Therefore $N$ is a set of (near) vertices with \emph{small proximity}. As
    $G$ consists of vertices with large proximity (i.e., it consists of $\cdc(\rho_F)$ and all vertices in $z \in I(\pi,h) \cap F$
    that satisfy $\pro_{\ba'}(z) > L /(2^ir) $) we have that $I(\pi,h) \cap F
    = N \cup G$.
    
    \begin{claim}\label{cl:PTweight}
    \[\wt(\PT)\leq \begin{cases}
    2\cdot \sum_{y \in N} \pro_{\ba'}(y) & \text{ if }|G|=1\\
    2\cdot \sum_{y \in N} \pro_{\ba'}(y) + |G|^{d/(d-1)}\cdot \frac{128dL}{r^2 2^i} & \text{ otherwise.}
    \end{cases}\]
    \end{claim}
\begin{claimproof}
If $d=2$, then each component of $\PT$ consists of a segment connecting the crossings whose length is $\sum_{y \in N} \pro(y)$, and a segment connecting an endpoint of this segment to the nearest grid point, which is of length at most $\frac{L|G|}{r^2 2^i}$. If $|G|=1$, then this latter connection to the grid is not made. When $|G|\geq 2$, then the total cost of these grid connections is therefore  $|G|^2\frac{L}{r^2 2^i}$. This concludes the proof for $d=2$.

For general $d$, recall that $\PT$ consists of two types of edges, those that connect
vertices of $G$ to points of $\grid(F,g)$ and the edges of $\PT^*$.
To bound $\wt(\PT^*)$, let $T_x$ be the tree of $\PT^*$ containing a
given $x\in G$, and let $\rho_x$ be the root of the subtree $T_x$.
We claim that $\wt(T_x[\rho_x,x]) + \sum_{y\in V(T_x)\cap N}
\pro_{\ba'}(y) = \wt(T_x)$, where $T_x[\rho_x,x]$ is the path in
$T_x$ connecting $\rho_x$ to $x$. To see this, notice that all
vertices of $T_x$ are covered by $\bigcup_{y\in N\cap V(T_x)}
\cdc^{-1}(y)$ except those in $V(T_x[\rho_x,x])$, as
$V(T_x[\rho_x,x])\subseteq \cdc^{-1}(x)$. Thus by the
definition of $\pro_{\ba'}(.)$ we have that arcs whose
lengths are included in the sum $\sum_{y\in V(T_x)\cap I(\pi,h)}
\pro_{\ba'}(y)$ include each arc of $E(T_x)\setminus E(T_x[\rho_x,x])$.

Notice that $x$ is the vertex among $V(T_x)\cap I(\pi,h)$ that is
closest to $\rho_x$ in $T_x$, and $\pro_{\ba'}(x)$ is an upper bound
on this distance.

We claim that $\wt(T_x) \leq 2\cdot \sum_{y \in N\cap V(T_x)} \pro_{\ba'}(y)$. This trivially holds when $N\cap V(T_x)=\emptyset$ since that implies that $T_x$ is a single-vertex tree. 

Suppose now that $N\cap V(T_x)\neq \emptyset$. It is sufficient to show that $\wt(T_x[\rho_x,x])\leq \max_{y}\pro_{\ba'} (y)$. This is immediate when $\rho_x=x$ because then $\wt(T_x[\rho_x,x])=0$. Otherwise, the definition of $\PT^*$ implies that $\rho_x$ has a child $z$ such that $z\not \in T_x[\rho_x,x]$. Then $\pro_{\ba'}(\cdc(z))=\wt(T_x[\rho_x,z])$, and by the definition of $\cdc(.)$, we have $T_x[\rho_x,x]\leq \wt(T_x[\rho_x,z])$. Consequently, $\wt(T_x[\rho_x,x])\leq \max_{y}\pro_{\ba'} (y)$ holds, as the right hand side includes the term for $\cdc(z)\in N\cap V(T_x)$, and thus $\wt(T_x) \leq 2\cdot \sum_{y \in N\cap V(T_x)} \pro_{\ba'}(y)$ holds. Summing over each tree $T_x$ of $\PT^*$, we get
\[\wt(\PT^*)\leq 2\cdot \sum_{y \in N} \pro_{\ba'}(y).\]

To bound the connections of $G$ to the grid points, recall that no connections are made when $|G|=1$. When $|G|\geq 2$, then \ref{lem:gridconnect_cost} implies that each connection is of length less than $\frac{128dL|G|^{1/(d-1)}}{r^2 2^i}$, so the total length of these connections is less than $|G|^{d/(d-1)}\cdot \frac{128dL}{r^2 2^i}$, which concludes the proof of our claim.
\end{claimproof}    
    
    We consider the case $|G|=1$ (so $G=\{\cdc(\rho_F)\}$). By \Cref{cl:PTweight} we have 
    \[
        \wt(\PT) \leq 2\sum_{x \in N} \pro_{\ba'}(x) = 2\cdot \frac{L}{2^i} \sum_{x \in
        N}\alpha_{i,\ba'}(x),
    \]
    and the lemma follows. Thus, from now we assume $|G| \geq 2$.    

    Next, we upper bound the size
    of $G$.

    \begin{claim}\label{eq:boundz}
    If $|G|\geq 2$, then
        \begin{equation*}
            |G|^{d/(d-1)} \leq 32 \frac{d L}{2^i} \sum_{z \in G}\frac{1}{\pro_{\ba'}(z)}.
        \end{equation*}
    \end{claim}
    \begin{claimproof}
        Recall that $T^G_F$ is the smallest subtree of $T$ that contains $G$ and $G'\coloneqq G \setminus \{\cdc(\rho_F)\}$.
        By property (c) of Lemma~\ref{lem:lighttree} of $T$ with set $Q\coloneqq G$
        and cell $F$, we have that $\wt(T^G_F)\leq 
        \frac{8d L}{2^i}|G|^{1-1/(d-1)}$.
        Thus by \Cref{lem:pro_sum} we have that $\sum_{z \in G'}\pro_{\ba'}(z)\leq \wt(T^G_F)\leq \frac{4dL}{2^i}|G|^{1-1/(d-1)}$.
        Next, we apply the 
        Cauchy-Schwartz inequality to the vectors
        $\big(\sqrt{\pro_{\ba'}(z)}\big)_{z\in G'}$ and $\big(\sqrt{1/\pro_{\ba'}(z)}\big)_{z\in G'}$, which gives
        \[
            |G'|^2 \leq  \left(\sum_{z \in G'} \pro_{\ba'}(z)\right) \cdot  \left(\sum_{z \in
            G'}\frac{1}{\pro_{\ba'}(z)}\right) \leq \frac{8d L}{2^i}|G|^{1-\frac{1}{d-1}} \cdot \left( \sum_{z \in
            G'}\frac{1}{\pro_{\ba'}(z)}\right).
        \]
        Since $|G| \geq 2$, we have that $|G|^2 = (|G'|+1)^2 \leq 4|G'|^2$
        and conclude the proof of the claim.
    \end{claimproof}

    Now, we can upper bound $\wt(\PT)$ by bounding the weight of each the edge from $x$
    to the parent of $x$ in $\PT$ as follows:
	\begin{align*}
        \wt(\PT) &\le 2\cdot \sum_{y \in N} \pro_{\ba'}(y) + \sum_{z \in G}
        \frac{128 dL}{2^ir^2}\cdot |G|^{1/(d-1)}& \hfill \text{\textcolor{gray}{(by~\Cref{cl:PTweight})}}\\
        &= 2\cdot \sum_{y \in N} \pro_{\ba'}(y) + \frac{128 dL}{2^i r^2}\cdot|G|^{\frac{d}{d-1}}&\\
        &\le 2\cdot  \sum_{y \in N} \pro_{\ba'}(y) +
                2^{12}d^2\left(\frac{L}{2^ir}\right)^2 \sum_{z \in
                G}\frac{1}{\pro_{\ba'}(z)}& \hfill \text{\textcolor{gray}{(by~\cref{eq:boundz})}}\\
				&= \frac{L}{2^i} \cdot \Oh\left(\left(\sum_{y \in N}
                    \frac{2^i}{L} \cdot \pro_{\ba'}(y)\right) + \left(\sum_{z \in
                G}\frac{d^2}{r^2}\frac{L}{2^i}\frac{1}{\pro_{\ba'}(z)}\right)\right)				\\
                &=  \Oh \left( \frac{L}{2^i}\sum_{x \in I(\pi,h) \cap F}\alpha_{i,\ba'}(x)\right).&\qedhere
	\end{align*}
	\end{proof} 
	
	\paragraph{Expected patching cost analysis.}
	Recall that by Lemma~\ref{levelprob}, the hyperplane $h$ gets level $i$ with probability
    $2^{i-1}/L$, where $i=i(a_1)$ is the level of $h$.
    Thus we have the following.
\begin{align}\label{eq:sumswap}
		\mathbb{E}[\cost(h)] &= 
		\frac{1}{L^d}\sum_{\ba \in [L]^d} \cost_{\ba}(h)\nonumber\\	
		&= \frac{1}{L^d}\sum_{a_1 \in [L]} \sum_{\ba' \in [L]^{d-1}} \sum_{\substack{F\text{ is facet of $C$ in } h,\\\text{level of }
        C \text{ is } i=i(a_1)}} \wt(\PT)\nonumber\\
        &= \frac{1}{L^d}\sum_{a_1 \in [L]} \sum_{\ba' \in [L]^{d-1}} \sum_{\substack{F\text{ is facet of $C$ in } h,\\\text{level of }C \text{ is } i=i(a_1)}}
        \Oh\left( \frac{L}{2^i}\sum_{x\in I(\pi,h)\cap F} \alpha_{i,\ba'} (x)\right)\\
        &= \Oh\left( \frac{1}{L^d} \sum_{\ba' \in [L]^{d-1}} \sum_{a_1 \in [L]} \frac{L}{2^i} \sum_{\substack{F\text{ is facet of $C$ in } h,\\\text{level of }C \text{ is } i=i(a_1)}}
       \left( \sum_{x\in I(\pi,h)\cap F} \alpha_{i,\ba'} (x)\right)\right) \nonumber\\
         &= \Oh\left( \frac{1}{L^d} \sum_{\ba' \in [L]^{d-1}} \sum_{a_1 \in [L]} \frac{L}{2^i} \sum_{x\in I(\pi,h)} \alpha_{i,\ba'} (x)\right)\nonumber
\intertext{Notice that for any fixed $\ba'$, the value of $\alpha_{i,\ba'}$ depends only on the level $i=i(a_1)$ of $h$, and there are $2^j$ values of $a_1$ where $i=j$. Thus, we can write}
         \mathbb{E}[\cost(h)]
         &= \Oh\left( \frac{1}{L^d} \sum_{\ba' \in [L]^{d-1}} \sum_{a_1 \in [L]} \frac{L}{2^i} \sum_{x\in I(\pi,h)} \alpha_{i,\ba'} (x)\right)\nonumber\\
         &= \Oh\left( \frac{1}{L^d} \sum_{\ba' \in [L]^{d-1}} \sum_{j =0}^{1+\log L} 2^j \frac{L}{2^j} \sum_{x\in I(\pi,h)} \alpha_{j,\ba'} (x)\right)\nonumber\\
         &= \Oh\left( \frac{1}{L^{d-1}} \sum_{\ba' \in [L]^{d-1}} \sum_{x\in I(\pi,h)} \sum_{j =0}^{1+\log L} \alpha_{j,\ba'} (x)\right)\nonumber
\end{align}
    In the innermost sum, for a fixed $\ba'$ and $x$ we have that
	\[
		\sum_{j=0}^{1+\log L} \alpha_{j,\ba'}(x)  = \sum_{j=0}^{\theta_1}
        \alpha_{j,\ba'}(x)+\sum_{j=\theta_1+1}^{\theta_2} \alpha_{j,\ba'}(x) \leq \frac{2}{r}+\frac{2d^2}{r} < \frac{3d^2}{r},
\]
    where we have 
    set $\theta_1 \coloneqq \log\left(\frac{L}{ \pro_{\ba'}(x)\cdot r}\right)$ and
    $\theta_2 \coloneqq \log\left(
    \frac{L}{ \pro_{\ba'}(x)} \right)$, and used the bound on the sum of both geometric series.
	Consequently,
    $\mathbb{E}[\cost(h)]= \Oh\big(\frac{d^2|I(\pi,h)|}{r}\big)$. We can now conclude the proof of \Cref{thm:struct} with~\eqref{eq:OPTtogridhyp}.

\section{Approximate TSP in \texorpdfstring{$\mathbb{R}^d$}{d-dimensional space}}
\label{sec:algorithm}

In this section we prove the algorithmic part of Theorem~\ref{thm:tsp}. The first few steps of the
algorithm are the same as in Arora's algorithm~\cite{Arora98}, as outlined in Section~\ref{sec:etsp-in-plane}.

In \emph{Step 1}, we perturb points and assume that $P
\subseteq [L]^d$ for some integer $L = \Oh(n \sqrt{d}/\eps)$ that is
a power of $2$. Then in \emph{Step 2} we pick a uniform random shift $\ba
\in \{0,\ldots,L\}^d$ and construct a compressed quadtree. 

In \emph{Step 3} we use the following result by Rao and
Smith~\cite{RaoS98} (the result can also be obtained by applying the procedure \textsc{PATCH} from~\cite{geomspannet} to the graph obtained from~\cite[Lemma 19.3.2]{geomspannet}). 

\begin{lemma}[\cite{RaoS98}, see also \cite{geomspannet}]\label{lem:spann}
    Let $P\subset \Reals^d$ be a set of $n$ points and let $\ba$ be a  random shift. There is a $\poly(1/\eps)n \log(n)$
    time algorithm that, given $P$ and $\ba$, computes a set of line segments $S$ such that
	\begin{enumerate}
        \item $\mathbb{E}_{\ba}[\wt(\pi^S)-\wt(\pi)] =
            \Oh(\eps \cdot  \wt(\pi))$, where $\pi$ is a shortest tour of $P$ and $\pi^S$ is a shortest tour of $P$ among the tours that use
            only edges from $S$. 
        \item $S$ is $1/\eps^{\Oh(d)}$-light with respect to $\mathrm{D}(\ba)$.
	\end{enumerate}
Moreover, in $\poly(1/\eps) n \log n$ time we can store for each facet $F$ of every cell $C$ of $\mathrm{CQT}(P,\ba)$ all the crossings of $S$ and $F$.
\end{lemma}

Let $\pi$ be the optimum TSP tour on the perturbed point set $P$. Lemma~\ref{lem:spann} gives 
us the set $S$ with the property that (i) there exist $\pi^S$ that uses only
edges from $S$, and (ii) in expectation the extra weight of $\pi^S$ is only
$\Oh(\eps\cdot \wt(\pi))$, and (iii) $\pi^S$ crosses every cell of the quadtree at
most $1/\eps^{\Oh(d)}$ times.
This summarizes all the steps from previous work that we will use in the
algorithm.

We apply Theorem~\ref{thm:struct} to $\pi^S$, which guarantees the existence of
an $r$-simple tour that is a good approximation in expectation and has the
property that quadtree facets that are crossed only once outside the corners will be crossed at the same place
as in $\pi^S$. In \emph{Step 4}, which we describe in full detail in the next section, we find
the optimal $r$-simplification of $\pi^S$ with this property. Similar to Arora's algorithm and
the description in Section~\ref{sec:etsp-in-plane}, we use a dynamic programming
algorithm for this. However, there are two crucial changes. First, we cannot
bound the number of matchings in $d>2$ the same way as we did for $d=2$ since we
used the non-crossing property for this. For this reason, we combine the dynamic
programming with the rank-based approach~\cite{rank-based}. In order to achieve
a more efficient $\Oh(n\log n)$ running time dependence on $n$, we use a portal
set consisting of all crossings with the cell boundary and the set of line
segments $S$ from Lemma~\ref{lem:spann} when the tour has only one crossing
point in a given facet.

\subsection{Dynamic Programming}

We use a dynamic programming algorithm to find an $\Oh(1/\eps)$-simple salesman tour
with respect to the shifted quadtree (in a similar fashion to
Arora~\cite{Arora98}). With high probability, this salesman tour has weight
$(1+\eps)\cdot \wt(\text{OPT})$. The running time of this step is $2^{\Oh(1/\eps^{d-1})}n$.

For a given quadtree cell $C$ let $\partial C$ denote its boundary. Our Structure Theorem, (Theorem~\ref{thm:struct}) guarantees the existence of some
set of active portals $B \subseteq \partial C$ that will be traversed by the tour. Technically, each portal may be used at most twice: for the sake of convenient set notation, we will think of each portal as having two copies that are infinitesimally close, and selecting a subset of these modified portals. In
our subproblem we will fix such a set~$B$ and we are also given a perfect matching
$M$ on $B$. We say that a collection $\cP \coloneqq
\{\pi_1,\ldots,\pi_{|B|/2}\}$ of paths \emph{realizes} $M$ on $B$ if for each
$(p,q) \in M$ there is a path $\pi_i \in \cP$ with $p$ and $q$ as
endpoints.

For each facet $F$ there is a unique maximum facet $\ex(F)$ that is the boundary
of a cell in the compressed quadtree that contains~$F$. Note that when
considering the cell $C$ and one of its facets~$F$, we will place the portals
according to the grids of $\ex(F)$, or potentially at some point in $F\cap S$.
We say that~$B$~is \emph{fine} with respect to~$S$ if for all facets~$F$ of~$C$
we have that either (i) $B\cap F^*=\{p\}$ and $p \in S\cap F$, or (ii) $B$ contains each point of $F \subset
\grid(\ex(F),r^{2d-2}/m_F)$ at most twice, where $|B\cap F^*|\leq 2m_F \leq 2r^{d-1}$.  Note that
the first option is needed because in Definition~\ref{def:simple} we often need
a perfect precision on faces that are crossed at exactly one point apart from their corners.

The subproblems are defined as follows (cf., $r$-multipath problem in Arora~\cite{Arora98}).

\defproblem{$r$-Multipath Problem}
{A nonempty cell $C$ in the shifted quadtree, a portal set $B \subseteq \partial C$ that is fine with $S$, and a perfect matching $M$ on $B$.}
{Find an $r$-simple path collection $\cP_{B,M}$ of minimum total length that satisfies the following properties.
\vspace{-0.5em}
\begin{itemize}
    \setlength\itemsep{0.0em}
    \item The paths in $\cP_{B,M}$ visit all input points inside $C$.
    \item $\cP_{B,M}$ crosses $\partial C$ only through portals from $B$.
    \item $\cP_{B,M}$ realizes the matching $M$ on $B$.
\end{itemize}
}

Arora~\cite{Arora98} defined the multipath problem in a similar way. The main
difference is that he considers all $B \subseteq \bigcup_F \grid(\ex(F), (r\log n)^{d-1})$,
while our structural Lemma enables us to select $B$ from $\bigcup_F \grid(\ex(F),
r^{2d-2}/m_F)$ of size at most $2\sum_F m_F$ apart from the special case with $1$ crossing on $F^*$. (In this comparison and the subsequent sketch we ignore the corner crossings.)  Arora~\cite{Arora98} showed how to use dynamic programming to solve the
$r$-Multipath problem in time $\Oh(n \cdot m^{\Oh(r)})$ (for $d=2$) which is
already too expensive in our case for $m,r = \Oh(1/\eps)$.
Here and below, the union is taken over all facets $F$ of the given cell $C$.

Before we explain our approach in detail, let us first discuss the natural
dynamic programming algorithm for the $d=2$ case and why it is not fast enough
for $d>2$. The dynamic programming builds a lookup table that contains the costs
of all instances of the multipath problem that arise in the quadtree, which is
exactly the same as in Arora~\cite{Arora98}. When the table is built, it is
enough to output the entry that corresponds to the root of the quadtree. The
number of non-empty cells in the compressed quadtree is $\Oh(n)$. For each facet
$F$ of the cell $C$, we guess an integer $m_F \leq 1/\eps^{d-1}$ that is the
number of times the $\Oh(1/\eps^{d-1})$-simple salesman tour crosses $\ex(F^*)$. Then,
we select a set $B$ of size at most $2m$ by choosing from $\bigcup_F
\grid(\ex(F),r^{2d-2}/m_F)$, where $\sum_F m_F=m$. There are at most 
\[\prod_{F\text{ facet of }C}\;\; \sum_{m_F=2}^{r^{d-1}}
\left( 3^{m_F}\cdot \binom{r^{2d-2}/m_F}{m_F}\right) \leq 2^{\Oh(r^{d-1})}\]
possible choices for the  set of active portals $B$ by Claim~\ref{portal-bin-ineq}, since
the number of facets~$F$ of~$C$ is only~$2d$. Unfortunately, the number of
perfect matchings on the $m$ points is $2^{\Oh(m \log m)}$. Since $m=\Oh(r^{d-1})=\Oh(1/\eps^{d-1})$, this would lead to a
running time of $2^{\Oh(1/\eps^{d-1} \log(1/\eps))}$, which has an extra
$\log(1/\eps)$ factor in the exponent compared to our goal. If $d=2$
we could use the fact that an optimal TSP tour is crossing-free, which allows
one to look for ``crossing-free matchings'' (and their number is at most
$2^{\Oh(m)}$). To reduce the number of possible matchings in $d > 2$, we will
use the \emph{rank-based approach}.

\paragraph{Rank-based approach}

Now we describe how the rank-based
approach~\cite{rank-based,fast-hamiltonicity-jacm} can be applied in this
setting.
We will heavily build upon the methodology and terminology from~\cite{rank-based}, and describe the basics here for the unfamiliar reader.
We follow the notation from~\cite{BergBKK18}.

Let $C$ be the cell of the quadtree and let $B \subseteq \partial C$ be the
set of portals on its boundary with $|B| = m$ (note that $m$ is even). We
define the \emph{weight} of a perfect matching $M$ of $B$ to be the total
length of the solution to the multipath problem on $(C,B,M)$, and denote it by
$\wt(M)$. A weighted matching on $B$ is then a pair $(M,\wt(M))$ for some
perfect matching $M$. Let $\mathcal{M}(B)$ denote the set of all weighted
matchings on $B$. 

We say that two perfect matchings $M_1,M_2$ \emph{fit} if their union is a
Hamiltonian Cycle on $B$. For some set $\mathcal{R}[B] \subseteq \mathcal{M}(B)$
of weighted matchings and a fixed perfect matching $M$ we define
\begin{displaymath}
    \text{opt}(M,\mathcal{R}[B]) \coloneqq \min \left\{ \text{weight}(M') \;
    : \; (M',\text{weight}(M')) \in \mathcal{R}[B] \text{ and } M' \text{ fits } M \right\}
\end{displaymath}
Finally, we say that the set $\mathcal{R}[B] \subseteq \mathcal{M}(B)$ is
\emph{representative} if for any matching $M$, we have
$\text{opt}(M,\mathcal{R}[B]) = \text{opt}(M, \mathcal{M}(B))$. 
The following results is the crucial theorem behind the rank-based approach.

\begin{lemma}[Theorem 3.7 in~\cite{rank-based}]
    \label{lemma:reduce}
    There exists a set $\mathcal{R}^\star[B]$ of $2^{|B|/2-1}$ weighted matchings
    that is representative of $\mathcal{M}(B)$. There is an algorithm
    $\mathtt{Reduce}$ that given some representative set $\mathcal{R}[B]$ of
    $\mathcal{M}(B)$ computes a set $\mathcal{R}^\star$ in $|\mathcal{R}^\star[B]| \cdot 2^{\Oh(|B|)}$
    time.
\end{lemma}
In the following, $\mathcal{R} \coloneqq \bigcup_{B} \{\mathcal{R}[B]\}$ for $B
\subseteq \partial C$ that are fine with $S$. For convenience, we say that the family
$\mathcal{R}$ is representative if every $\mathcal{R}[B] \in \mathcal{R}$ is
representative.

Now, we are ready to describe the solution to the $r$-Multipath problem (see
Algorithm~\ref{rank-based-dp} for global pseudocode). The algorithm is given a
quadtree cell $C$ and a set of line segments $S$. The task is to output the
union of sets $\mathcal{R}^\star[B]$ for every $B \subseteq \partial{C}$, where
$B$ has size $m$ and it is fine with $S$, and $\mathcal{R}^\star[B]$ is
representative of $\mathcal{M}(B)$. We start the description of the algorithm
with a case distinction based on the type of the given cell in the quadtree.

In the \emph{base case}, we consider a cell that has one or zero points. Next,
we consider another special case, i.e., the \emph{compressed case}, when the
given cell has only one child in the compressed quadtree. After that, we show
how to combine $2^d$ children in the \emph{non-compressed non-leaf case}
paragraph.

\subparagraph*{Base case} We start with the base case, where the cell $C$ is a
leaf of the quadtree and contains at most one input point. Consider all possible sets $B$ that are fine with $S$. This gives an instance of at most $|B|+1$ points and we can use an exact algorithm to get
a set $\mathcal{R}^\star[B]$ in time $2^{\Oh(|B|)}$. We can achieve that with
a standard dynamic programming procedure: Let us fix $B$ and let $p$ be the
only input point inside $C$ (if it exists).  For every $X \subseteq B$ we will
compute a table $\mathtt{BC}[X]$ that represents $\mathcal{M}(X)$ for every $X
\subseteq B$. Initially $\mathtt{BC}[\emptyset] = \{(\emptyset,0)\}$ and if
$p$ exists, then for every $a,b \in B$ let $\mathtt{BC}[a,b] \coloneqq \{\{(a,b)\},
\dist(a,p) + \dist(p,b)\}$, which means that $p$ is connected to the portals
$a,b \in B$.  Next, we compute $\mathtt{BC}[X]$ for every $X \subseteq B$ with
the following dynamic programming formula.
\begin{displaymath}
    \mathtt{BC}[X] \coloneqq \mathtt{reduce}\left(
        \bigcup_{\substack{u,v \in X\\u\neq v}} \left\{ \big(M \cup \{(u,v)\}, \wt(M) + \dist(u,v)\big)
        \; \Big| \; 
        (M,\wt(M)) \in \mathtt{BC}[X \setminus \{u,v\}]
        \right\}  \right)
\end{displaymath}
For a fixed $B$ this algorithm runs in $\Oh(|\mathcal{R}^\star[B]| \cdot
2^{\Oh(|B|)})$ time and correctly computes $\mathtt{BC}[B]=\mathcal{R}^\star[B]$ (cf.,
\cite[Theorem 3.8]{rank-based} for details of an analogous dynamic programming
subroutine).

\subparagraph*{Compressed case} 

In this case, we are given a large cell $C_\out$ and its only child $C_\ins$.
From the dynamic programming algorithm, we know the solution to $C_\ins$ for all
relevant $B_\ins \subseteq \partial C_\ins$, and the task is to connect these
portals to the portals $B_\out \subseteq C_\out$. We do dynamic programming
similar to the one seen in the base case. We say that a pair $B_\ins, B_\out$
where $B_\out \subset \partial C_\out$ and $B_\ins \subset \partial C_\ins$ are
fine with $S$ if they are individually fine with $S$, and if $\partial C_\out
\cap \partial C_\ins \neq \emptyset$, then $B_\out \supset B_\ins \cap \partial
C_\out$. For each fixed pair $B_\out, B_\ins$ that are fine with $S$, we compute
a table $\mathtt{DBC}[X]$ (mnemonic for \emph{dummy base case}) that represents
$\mathcal{M}(X)$ (where the paths \emph{need not cover} any input points) for
every multiset $X \subseteq B_\out \uplus B_\ins$. Note that the cell $C_\out$
can be regarded as the disjoint union of $C_\ins$ and a \emph{dummy leaf cell}
that has region $C_\out \setminus C_\ins$. Initially, we set
$\mathtt{DBC}[\emptyset] = \{(\emptyset,0)\}$. We can then compute the values for
the dummy base cases $\mathtt{DBC}$ with the same formula as for the base case.
\begin{displaymath}
    \mathtt{DBC}[X] \coloneqq \mathtt{reduce}\left(
        \bigcup_{\substack{u,v \in X\\u\neq v}} \left\{ \big(M \cup \{(u,v)\}, \wt(M) + \dist(u,v)\big)
        \; \Big| \; 
        (M,\wt(M)) \in \mathtt{DBC}[X \setminus \{u,v\}]
        \right\}  \right)
\end{displaymath}

Let $\mathcal{R}^\star$ be the table of these
sets for all $B_\out ,B_\ins$ that are fine with $S$, i.e.,
$\mathcal{R}^\star[X] = \mathtt{DBC}(X)$ for all $X \subseteq
B_\out \cup B_\ins$. In order to get representative
sets $\mathcal{R}[B_\out ]$ of $\mathcal{M}(B_\out )$ for every $B_\out \subset
\partial C_\out$ that
is fine with $S$, we can combine the representative set $\mathcal{R}_\ins^\star$ of
$C_\ins$ and $\mathcal{R}^\star$ (see Algorithm~\ref{alg:compresed}).

\begin{algorithm}
	\SetAlgoLined
	\DontPrintSemicolon
	\SetKwInOut{Input}{Algorithm}
	\SetKwInOut{Output}{Output}
    \Input{$\mathtt{CompressedCase}(C_\ins, C_{\out},\mathcal{R}^\star_\ins)$.
    $C_\out$ is a compressed cell and $C_\ins$ its child}
    Let $\mathcal{R}^\star_\dum[X] \leftarrow \mathtt{DST}(X)$ for every relevant $X \subseteq B_\ins\ \cup B_\out$\\
    \ForEach{$M_\ins \in \mathcal{R}^\star_\ins, M_{\dum} \in \mathcal{R}^\star_{\dum}$}{
            \If{$M_\ins$, $M_{\dum}$ are compatible}{
                Let $M_\out \leftarrow \mathtt{Join}(M_\ins,M_{\dum})$\\
                Let $B_\out \leftarrow$ ground set of $M_\out$  \tcp*{Note that $B_\out \subset \partial C_\out$}
                \If{$B_\out$ is fine with respect to $S$}{
                    Insert $\big(M_\out,\,\wt(M_\ins) + \wt(M_{\dum})\big)$ into $\mathcal{R}[B_\out]$
                }
            }
        }
        \ForEach{$B_\out\subset \partial C_\out$ that is fine with $S$}{ 
            $\mathcal{R}[B_\out] \leftarrow \mathtt{reduce}(\mathcal{R}[B_\out])$
        }
    \Return $\mathcal{R}$
	\caption{Pseudocode for compressed cells}
	\label{alg:compresed}
\end{algorithm}

Observe that for a fixed $B_\out $ and $B_\ins$ this algorithm runs in
$\Oh(|\mathcal{R}^\star[B_\out]| \cdot |\mathcal{R}^\star[B_\ins]| \cdot
2^{\Oh(|B_\out |+|B_\ins|)})$ time and correctly computes the distances and
matchings for every that are $B_\out  \subseteq \partial C_\out$, $B_\ins
\subseteq \partial C_\ins$  that are fine with $S$.

\subparagraph*{Non-compressed non-leaf case} For non-compressed non-leaf cells $C$ we combine
the solutions of cells of one level lower. Let $C_1,\ldots,C_{2^d}$ be the
children of $C$ in the compressed quadtree. Also, let $\mathcal{R}_i$ be the solution to
the $r$-Multipath problem in cell $C_i$ that we get recursively. Next we iterate
over every $M_1 \in \mathcal{R}_1,\ldots,M_{2^d} \in \mathcal{R}_{2^d}$ and
check if matchings $M_1,\ldots,M_{2^d}$ are \emph{compatible}. By this, we mean
that (i) for every neighboring cell $S_i,S_j$ the endpoints of matchings on
their shared facet are the same and (ii) combining $M_1,\ldots,M_{2^d}$ results in
a set of paths with endpoints in $\partial C$.

\begin{algorithm}
	\SetAlgoLined
	\DontPrintSemicolon
	\SetKwInOut{Input}{Algorithm}
	\SetKwInOut{Output}{Output}
    \Input{$\mathtt{MultipathProblem}(C,S,r)$}
    \Output{Family $\mathcal{R}$, which is the union of sets $\mathcal{R}[B]$ of weighted matchings that represent $\mathcal{M}(B)$ for each $B$ that is fine with $S$}
    \LineIf {$|C\cap P| \le 1$}{$\mathcal{R} \leftarrow$ base case with one or no points}\\
    \ElseIf{$C$ is compressed}{
        Let $C_{+}$ be the only child of $C$ and $\mathcal{R}_{+}$ solution on $C_{+}$\\
        $\mathcal{R} \leftarrow \mathtt{CompressedCase}(C,C_{+},\mathcal{R}_{+})$ 
    }
    \Else{
        \label{uca}Let $C_1,\ldots,C_{2^d}$ be the children of $C$\\
        Let $\mathcal{R}_i \leftarrow \mathtt{MultipathProblem}(C_i,S,r)$\\
        \ForEach{$M_1 \in \mathcal{R}_1, \ldots, M_{2^d} \in \mathcal{R}_{2^d}$\label{line:loop}}{
            \If{$M_1,M_2$\ldots and $M_{2^d}$ are compatible}{
                Let $M \leftarrow \mathtt{Join}(M_1,\ldots,M_{2^d})$, let $B \leftarrow$ ground set of $M$\label{line:insert}\\
                \If{ $B$ is fine with respect to $S$}{
                    Insert $\big(M,\,\wt(M_1) + \ldots + \wt(M_{2^d})\big)$ into $\mathcal{R}$
                }
            }
        }
        \ForEach{$B\subset \partial C$ that is fine with $S$\label{line:loop2}}{ 
           \label{ucb} $\mathcal{R}[B] \leftarrow \mathtt{reduce}(\mathcal{R}[B])$
        }
    }
    \Return $\mathcal{R}$
	\caption{Pseudocode of the dynamic programming for the Multipath problem}
	\label{rank-based-dp}
\end{algorithm}

Next, if the matchings $M_1, \ldots, M_{2^d}$ are compatible, we \emph{join}
them (join can be thought of as $2^d - 1$ joins of matchings defined in
\cite{rank-based}). This operation will give us the matching obtained from $M_1
\cup \ldots \cup M_{2^d}$ by contracting degree-two edges if no cycle is
created, and gives us matchings on the boundary (i.e., set $B$) and information
about the connection between these points (i.e., a matching $M$ on the set $B$).

If $B$ is fine with $S$, then we insert $M$ into $\mathcal{R}[B]$ with weight
being the sum of the weights of matchings $M_1, \ldots, M_{2^d}$. At the end, we
will use the operation \emph{reduce} to decrease the sizes of all
$\mathcal{R}[B]$ and still get a representative set of size $2^{\Oh(|B|)}$. The
corresponding pseudo-code is given in Lines~\ref{uca} to~\ref{ucb} of
Algorithm~\ref{rank-based-dp}.

\subparagraph*{Overall Algorithm}
\begin{lemma}
    For a cell $C$ and a fixed $B$ that is fine with
    respect to $S$, the set $\mathcal{R}[B]$ computed in
    Algorithm~\ref{rank-based-dp} is representative of $\mathcal{M}(B)$.
\end{lemma}
\begin{proof}
    The proof is by induction on $|C\cap P|$. For $|C\cap P| \le 1$ the lemma
    follows from the correctness of the base case. Next we assume that $|C\cap
    P| > 1$ and has some children $C_1,\ldots,C_{2^d}$ in the quadtree. Let us
    fix some $B \subseteq \partial C$ of size $m$ that is fine with respect to
    $S$, a matching $M$ on $B$ and an optimal solution, i.e., collection of
    $r$-simple paths $\text{OPT}(S,B,M,r) = \{\pi_1,\ldots,\pi_{|B|/2}\}$ with
    distinct endpoints in $B$ that realize matching $M$. Because
    $\text{OPT}(S,B,M,r)$ is $r$-simple, there exists $B_1 \subseteq \partial
    C_1,\ldots, B_{2^d} \subseteq \partial C_{2^d}$ that are fine with respect
    to $S$ and matchings $M_1,\ldots,M_{2^d}$ on $B_1,\ldots,B_{2^d}$ such
    that $\text{OPT}(S,B,M,r)$ crosses boundaries between $C_1,\ldots,C_{2^d}$
    exactly in $B_1,\ldots,B_{2^d}$ and the matchings $M_1,\ldots,M_{2^d}$ are
    compatible and their join is $M$. Hence in Line~\ref{line:insert},
    Algorithm~\ref{rank-based-dp} finds $B$ and the matching $M$. Next we will
    insert it with the weight $\wt(\text{OPT}(S,B,M,r))$ to the set $\mathcal{R}[B]$.
    Since the join operation preserves representation (see \cite[Lemma
    3.6]{rank-based}), the set $\mathcal{R}[B]$ is a representative set.
    Finally, by Lemma~\ref{lemma:reduce} we assert that the $\mathtt{reduce}$
    algorithm also outputs a representative set. An analogous argument shows
    that the sets $\mathcal{R}[B]$ computed for compressed cells are also
    representative.
\end{proof}

\begin{lemma}
    Algorithm~\ref{rank-based-dp} runs in time $\Oh(n \cdot
    |\mathcal{R}|^{2^{\Oh(d)}} \cdot 2^{\Oh(|B|)})$, where $n$ is the number
    of points in~$C$.
\end{lemma}
\begin{proof}
    In the algorithm, we use a compressed quadtree, therefore the number of
    cells to consider is $\Oh(n)$. Algorithm~\ref{rank-based-dp} in the base
    case runs in time $\sum_B
    |\mathcal{R}|2^{\Oh(|B|)}=|\mathcal{R}|^{\Oh(1)} 2^{\Oh(|B|)}$. The for loop in
    Line~\ref{line:loop} of Algorithm~\ref{rank-based-dp} has $|\mathcal{R}_1|
    \cdots |\mathcal{R}_{2^d}| = \Oh(|\mathcal{R}|)^{2^d}$ many
    iterations, and the analogous for loop in the compressed case has
    $|\mathcal{R}_{\ins}| \cdot |\mathcal{R}^\star_\empt|\cdot  2^{\Oh(|B|)} = 2^{\Oh(|B|)}|\mathcal{R}|^{\Oh(1)}$
    iterations. Checking whether matchings $M_1,\ldots,M_{2^d}$ are compatible
    and joining them takes $\poly(r,2^d)$ time. Moreover, checking
    whether $B$ is fine with respect to the set $S$ can be checked in
    $r^{\Oh(d)}$ time because Lemma~\ref{lem:spann} guarantees us that access
    to these points can be achieved through the lists. The for loop in
    Line~\ref{line:loop2} of Algorithm~\ref{rank-based-dp} has at most
    $|\mathcal{R}|$ many iterations. In each iteration, we invoke the
    $\mathtt{reduce}$ procedure that takes $|\mathcal{R}| \cdot 2^{\Oh(|B|)}$
    time according to Lemma~\ref{lemma:reduce}. Note that the running times in
    the compressed case can be bounded the same way. This yields the claimed
    running time.
\end{proof}

\begin{lemma}
    $|\mathcal{R}| \cdot 2^{\Oh(|B|)} \le 2^{\Oh(r^{d-1})}$
\end{lemma}
\begin{proof}
    First, recall that $|B| < 4d r^{d-1} + 2^d=\Oh(r^{d-1})$ as
    Definition~\ref{def:simple} implied that $|B\cap F^*|\leq 2r^{d-1}$ for each of the $2d$ faces. Next, we bound the number of possible sets~$B$.
    We select sets $B\cap F^*$ of size at most $m_F$, where the points can be chosen from
    $\grid(\ex(F),r^{2d-2}/m_F)\cap F^*$, each with multiplicity $0,1$ or $2$, so $|B\cap F^*| \leq 2m_F\leq 2r^{d-1}$ or
    (when some facet $F^*$ is crossed exactly once) it can also
    be chosen from $S \cap F$ with multiplicity at most $2$. Including the choice of some subset of the $2^d$ corners of $C$, each of multiplicity at most $2$, there are at most
    \[3^{2^d}\cdot 3^{2d}\binom{|S\cap\partial C|}{2d}\cdot \prod_F
	\left(\sum_{m_F=1}^{r^{d-1}} 3^{m_F}
	\binom{r^{2d-2}/m_F}{m_F}\right)\]
    possible choices for $B$. Recall that there are at most $2d$ possible
    facets $F$. Moreover,  Lemma~\ref{lem:spann} guarantees that $S$ crosses
	each face at most  $1/\eps^{\Oh(d)}$ times, hence $|S \cap F| \le
	1/\eps^{\Oh(d)}=r^{\Oh(d)}$ and $\binom{|S\cap\partial C|}{2d} \le r^{\Oh(d^2)}$. By
	Claim~\ref{portal-bin-ineq}, $\binom{r^{2d-2}/m_F}{m_F}$ is bounded by
    $2^{\Oh(r^{d-1})}$, and $3^{m_F}=2^{\Oh(r^{d-1})}$. Therefore, the number of possible choices for $B$ is at most
	\[3^{2^d}\cdot r^{\Oh(d^2)} \cdot \prod_F 	\left(\sum_{m_F=1}^{r^{d-1}}
	2^{\Oh(r^{d-1})}\right) = 2^{\Oh(r^{d-1})}.\] Next
	we bound $\mathcal{R}[B]$ for a fixed $B$. Note that in
	Algorithm~\ref{rank-based-dp}, we always use the subroutine $\mathtt{reduce}$ to
	reduce the size of $\mathcal{R}[B]$. Lemma~\ref{lemma:reduce} guarantees
	that this procedure outputs a set $\mathcal{R}^\star[B]$ of size at most
	$2^{|B|-1}$. Multiplying all of these factors together gives us the desired property.
\end{proof}
Combining all of the above observations gives us the following Corollary.
\begin{corollary}
    \label{cor:algorithm}
    Suppose we are given a compressed quadtree $Q$, a point set $P$ with $n$
    points, and a set $S$ of segments that cross each facet of $Q$ at most
    $r^{\Oh(d)}$ times. Let $\pi^S \subseteq S$ be the shortest salesman tour of
    $P$ within $S$. Then, we can find the shortest $r$-simple salesman tour
    $\pi'$ that visits all points in $P$ and, if $\pi^S$ crosses any facet of
    $Q$ exactly once, then $\pi'$ crosses it at the same point, in $n\cdot
    2^{\Oh(r^{d-1})}$ time.
\end{corollary}

We can now proceed with the proof of the algorithm's existence from Theorem~\ref{thm:tsp}.

\begin{proof}[Proof of the algorithmic part of Theorem~\ref{thm:tsp}]
For the running time observe that Step 1, 2, 3 and 5 take
$\poly(1/\eps) \cdot n \log{n}$ time. In  Step~4 
we set $r$ to $\Oh(d^{5/2}/\eps)$ and by Corollary~\ref{cor:algorithm} we get an
extra $n \cdot 2^{\Oh(d^{5/2}/\eps)^{d-1}}$ factor. Overall, this gives the claimed
running time.

For the approximation ratio, assume that $\pi$ is the optimal solution. Note
that Step 1 perturbs the solution by at most $\Oh(\eps\cdot \wt(\pi))$. In Step 3,
by the Lemma~\ref{lem:spann} we are guaranteed that there exists a tour $\pi^S$ of
weight $\Oh(\eps\cdot \wt(\pi))$ larger than $\pi$ (in expectation). Next in Step 4,
Corollary~\ref{cor:algorithm} applied to the set $S$, guarantees that that we
find a salesman tour $\pi'$ that satisfies the condition of Structural
Theorem~\ref{thm:struct} for $\pi^S$. It means that $\mathbb{E}_\ba[\wt(\pi') -
\wt(\pi^S)] = \Oh(\eps\cdot \wt(\pi))$ and $\wt(\pi') = (1+\Oh(\eps)) \wt(\pi)$. Applying
Step 5 on $\pi'$ can only decrease the total weight of $\pi'$. This concludes
the proof of the algorithmic part of Theorem~\ref{thm:tsp}.
\end{proof}

It is easy to see that the algorithm can be derandomized by trying all possibilities for $\ba$.

\begin{remark}\label{rem:spannerfree_alg}
One can swap out the patched spanner of Rao and Smith in \emph{Step 3} with Arora's structure theorem in order to avoid using spanners, which can be beneficial for certain problems. As a result, when the patched tour has a single crossing in a facet, its location would have to be guessed from Arora's portals, which are $\grid\left(\ex(F),\Oh\left(\frac{\log(1/\eps)}{\eps^{1/(d-1)}\log n}\right)\right)$. This results in $\poly(1/\eps)\cdot (\log n)^{d-1}$ potential locations for the crossing in the facet rather than just $\poly(1/\eps)$ as with spanners. For a given cell $C$, there would be $\prod_F
\left(\binom{r^{2d-2}/m_F}{m_F}+(\poly(1/\eps)\log n)^{d-1}\right) = 2^{\Oh(1/\eps^{d-1})}(\log n)^{(d-1)\cdot 2d} $ options. This results in a spanner-free algorithm with a slightly slower running time of $2^{\Oh(1/\eps^{d-1)}}n(\log n)^{2d^2-2d}=2^{\Oh(1/\eps^{d-1)}}n \,\poly(\log n)$.
\end{remark}

\section{Algorithm for \textsc{Euclidean} and \textsc{Rectilinear Steiner Tree}}
\label{sec:steinertree}
In this subsection, we consider extensions to two variants of \textsc{Steiner
Tree}: \textsc{Euclidean Steiner Tree} and \textsc{Rectilinear Steiner Tree}. As
most techniques work the same way for these problems, we only sketch the
differences compared to our algorithm for  \textsc{Euclidean TSP}.

The notion of spanners for the Steiner tree problems is more complicated (one
requires so-called banyans). We summarize this notion at the end of this
section. Similarly to how we used spanners for the algorithm for TSP, here we
use banyan to determine a set $\tilde{S}$ of points. This set will consist of
$(1/\eps)^{\Oh(d)}$ portals for each facet that we use in the case of single
crossing. Consequently, we say that a portal set $B\subset \partial C$ is
\emph{valid} with respect to $\tilde{S}$ if for each facet $F$ of $C$, we have
that either (i) $B \cap F = \{p\}$ and $p \in \tilde{S} \cap F$, or (ii) $B
\subset \grid(\ex(F),r^{2d-2}/m_F)$ where $|B\cap F|\leq m_F \leq r^{d-1}$.

Additionally, we need to track connectivity requirements with partitions instead
of matchings. A partition $M$ of $B$ is realized by a forest if for any $b,b'\in
B$, we have that $b$ and $b'$ are in the same tree of the forest if and only if
they are in the same partition class of $M$. The problem we need to solve in
cells is the following.

\defproblem{$r$-Simple Steiner Forest Problem}
{A nonempty cell $C$ in the shifted quadtree, a portal set $B \subseteq \partial
C$ that is valid with $\tilde{S}$, and a partition $M$ on $B$}
{Find an $r$-simple forest $\cP_{B,M}$ of minimum total length that satisfies the following properties.
\vspace{-0.5em}
\begin{itemize}
    \setlength\itemsep{0.0em}
    \item The forest $\cP_{B,M}$ spans all input points inside $C$.
    \item $\cP_{B,M}$ crosses $\partial C$ only through portals from $B$.
    \item $\cP_{B,M}$ realizes the partition $M$ on $B$.
\end{itemize}
}

The rank-based approach~\cite{rank-based,fast-hamiltonicity-jacm} was originally
conceived with partitions in mind, and therefore we can still use representative
sets and the reduce algorithm as before (although the upper bound on
    $\mathcal{R}^\star[B]$ of $2^{|B|/2-1}$ from Lemma~\ref{lemma:reduce} needs
to be increased to $2^{|B|-1}$).
The main difference between TSP and Steiner Tree is the handling of leaf and
dummy leaf cells (i.e., the base case and the dynamic programming in the
compressed case).

\subparagraph*{Leaves and dummy leaves for \textsc{Rectilinear Steiner Tree}.}

Consider now a point set $Q\subset \Reals^d$. The \emph{Hanan-grid} of $Q$ is
the set of points that can be defined as the intersection of $d$ distinct
axis-parallel hyperplanes incident to $d$ (not necessarily distinct) points of
$Q$. By Hanan's and Snyder's results~\cite{hanan1966steiner,Snyder92}, the
optimum rectilinear Steiner tree for a given point set $Q$ lies in the
Hanan-grid of $Q$. In particular, in a leaf cell, our task is to find a
representative set for a fixed set of $k=\Oh(1/\eps)^{d-1}$ terminals (which
includes the input point in case of a non-empty leaf cell). Note that for each
fixed $B$, this task can be done in the graph $G$ defined by the Hanan-grid of
$B\cup (C\cap P)$, where the edge weights correspond to the $\ell_1$ distance.
The graph $G$ has $\poly(1/\eps)$ vertices and edges. Let $H$ be the set of
vertices in the Hanan-grid of $B$, i.e., the set of possible Steiner points for
the terminal set $B$, where $B$ is the set of portals on the boundary of the
cell.

To solve the base case efficiently, we will use a dynamic programming subroutine
inspired by the classical Dreyfus-Wagner algorithm~\cite{dw71}.
Let $\mathtt{ST}[D,v]$ be the minimum possible weight of a Steiner Tree for $D \cup
\{v\}$, for all $D \subseteq B$ and $v \in H$. In the base case
$\mathtt{ST}[\{b\},v] = \lVert b-v\rVert_1$. We can compute it efficiently
with the following dynamic programming formula:
\begin{displaymath}
    \mathtt{ST}[D,v] \coloneqq \min_{\substack{u \in H\\\emptyset \neq D' \subset
    D}} \Big\{ \mathtt{ST}[D',u] + \mathtt{ST}[D \setminus D', u] + \lVert u - v\rVert_1 \Big\}
    .
\end{displaymath}
This algorithm correctly computes a minimum weight Steiner Tree that connects
$D \subseteq B$ and the running time of this algorithm is $2^{\Oh(|B|)} \cdot \poly(1/\eps)$
(see~\cite{dw71}). Finally, let $\mathtt{ST}[X] \coloneqq \min_{v} \mathtt{ST}[X,v]$. 

Next, we take care of all partitions of $B$. This involves a similar dynamic
programming as in the base case of TSP. Let $\mathtt{SF}(X)$
be the set $\mathcal{R}^\star[X]$ that represents every partition of $X
\subseteq B$. Namely, for every Steiner forest $F$ with connected components
$B_1,\ldots,B_k$, such that $B_1\uplus\ldots\uplus B_k = X$ there exists $M \in
\mathcal{R}^\star[X]$, such that the union of $F$ and a forest $F_M$ whose connected components correspond to $M$ gives a tree that spans $X$.
At the beginning, we set $\mathtt{SF}[\emptyset] = \{\emptyset,0\}$. Next, we use the following dynamic
programming to compute $\mathtt{SF}[X]$ for all $X\subset B$:
\begin{displaymath}
    \mathtt{SF}[X] \coloneqq \mathtt{reduce}\left(
        \bigcup_{Y \subseteq X} \Big\{ (M \cup \{Y\}, \wt(M) + \mathtt{ST}[Y])\;
            \Big| \;
    (M, \wt(M)) \in \mathtt{SF}[X \setminus \{Y\}]\Big\}
    \right)
\end{displaymath}
The number of table entries $\mathtt{SF}[X]$ is $2^{|B|}$. To compute each entry we
need $2^{\Oh(|B|)} |\mathcal{R}^{\star}|$ time. Because $\mathtt{reduce}$
guarantees that $|\mathcal{R}^{\star}| \le 2^{\Oh(|B|)}$ we can bound the running time of the
dynamic programming algorithm by $2^{\Oh(|B|)}$. We know that $|B| \le
\Oh(1/\eps^{d-1})$ and the running time bound for the base case follows. The correctness follows from the
correctness of the procedure $\mathtt{reduce}$ for partitions (see~\cite{rank-based})
and the fact that $\mathtt{ST}[Y]$ is an optimal Steiner tree on the terminal set $Y \subseteq B$.

\subparagraph*{Leaves and dummy leaves for \textsc{Euclidean Steiner Tree}.}
In the case of \textsc{Euclidean Steiner Tree}, we can pursue a similar line of
reasoning. First, notice that in leaf and dummy leaf cells, it is sufficient to
compute a $(1+\Oh(\eps))$-approximate forest for $\tau'\cap C$, as these forests
are a subdivision of $\tau'$. By the grid perturbation argument within $C$, it
is sufficient to consider forests where the Steiner points lie in a regular
$d$-dimensional grid of side length $\Oh(1/\eps)$. Let $V_C$ be the set of
$\Oh(1/\eps)^d$ grid points obtained this way, and let $G$ be the complete graph
on $V_C$ where the edge weights are defined by the $\ell_2$ norm. Then the
minimum Steiner forest of $B$ for a given partition $M$ is equal to the
corresponding forest within $G$. In particular, it is sufficient to compute the
representative set of all partitions of $B$ in $G$. To achieve that, we use
exactly the same dynamic programming as in the base case for rectilinear Steiner
Tree. We only need to change the distance in the procedure $\mathtt{ST}$ to be
$\ell_2$ distance. Note that $\mathtt{ST}$ works in $2^{\Oh(|B|)} \cdot
\text{poly}(|V_C|)$ and the running time of the dynamic programming for
$\mathtt{SF}$ is bounded by $2^{\Oh(|B|)}\cdot \text{poly}(|V_C|)$.

\subparagraph*{Single crossings and banyans}

Similarly to \emph{Step 3} in our algorithm for Euclidean TSP, we use the
structure theorem based on the results of Rao and Smith~\cite{RaoS98} and Czumaj
et al.~\cite{survivable}. We include the details of the proofs in
Appendix~\ref{sec:filtering} for completeness. Note that  
this is taken almost verbatim from~\cite{survivable}.
We modify their approach to make
it work in the required time for both the rectilinear and the Euclidean case.
Let $\smt(P;E)$ be the minimum length Steiner tree with terminal set~$P$ which is
allowed to use only segments from $E$ as edges.

\begin{lemma}\label{lem:banyan}
    There is a $\poly(1/\eps)n\log(n)$
    time algorithm that, given point set $P$ and
    the random offset $\ba$ of the dissection, computes a set of segments
    $\tilde{S}$ such that:
	\begin{enumerate}
        \item
            $\mathbb{E}_{\ba}[\wt(\smt(P;\tilde{S}))-\wt(\smt(P;\mathbb{R}^d \times \mathbb{R}^d))]
                    = \Oh\big(\eps \cdot  \wt(\smt(P;\mathbb{R}^d \times
                        \mathbb{R}^d))\big)$.
                    \item for every facet $F$ of $\mathrm{D}(\ba)$ it holds that $|F \cap \tilde{S}|
            =1/\eps^{\Oh(d)}$.
	\end{enumerate}
\end{lemma}

Lemma~\ref{lem:banyan} is analogous to Lemma~\ref{lem:spann}. It gives us the
set $\tilde{S}$ of segments such that (i) there exists a Steiner Tree that uses
only edges from $\tilde{S}$, (ii) in expectation the excess weight of the tree
is only $\Oh(\eps \wt(\text{OPT}))$, and (iii) the segments of $\tilde{S}$ cross
every cell of the quadtree at most $(1/\eps)^{\Oh(d)}$ times.
Lemma~\ref{lem:banyan} works for both Rectilinear and Euclidean Steiner Tree
(see Appendix~\ref{sec:filtering}).

We use Lemma~\ref{lem:banyan} analogously to Lemma~\ref{lem:spann}. Namely, when
our dynamic programming procedure guesses that an optimum solution of the
$r$-simple Steiner Forest Problem crosses a cell facet $F$ exactly once, then we
guess the crossing exactly from $\tilde{S}\cap F$. Lemma~\ref{lem:banyan}
guarantees that the number of candidates is $(1/\eps)^{\Oh(d)}$, which is less
than $2^{\Oh(1/\eps)^{d-1}}$.

Putting the above ideas together proves the algorithmic part of Theorem~\ref{thm:st}.

\section{Lower bounds} \nocite{SKBdoktori}
\label{sec:lower-bounds}

Our starting point is the gap version of the Exponential Time Hypothesis~\cite{Dinur16,ManurangsiR17}, which is normally abbreviated as Gap-ETH. The hypothesis is about the \textsc{Max 3SAT} problem, where one is given a 3-CNF formula with $n$ variables and $m$ clauses, and the goal is to satisfy the maximum number of clauses.

\begin{namedthm}{Gap Exponential Time Hypothesis (Gap-ETH)}[Dinur~\cite{Dinur16}, Manurangsi and Raghavendra~\cite{ManurangsiR17}]
There exist constants $\delta,\gamma>0$ such that there is no $2^{\gamma m}$ algorithm which, given a 3-CNF formula $\phi$ on $m$ clauses, can distinguish between the cases where (i) $\phi$ is satisfiable or (ii) all variable assignments violate at least $\delta m$ clauses.
\end{namedthm}

Let \textsc{Max-(3,3)SAT} be the problem where we want to maximize the number of satisfied clauses in a formula $\phi$ where each variable occurs at most $3$ times and each clause has size at most $3$. (Let us call such formulas (3,3)-CNF formulas.) Note that the number of variables and clauses in a (3,3)-CNF formula are within constant factors of each other. Papadimitriou~\cite[pages 315--318]{PapadimitriouBook} gives an $L$-reduction from \textsc{Max-3SAT} to \textsc{Max-(3,3)SAT}, which immediately yields the following:

\begin{corollary}\label{cor:max33sat}
There exist constants $\delta,\gamma>0$ such that there is no $2^{\gamma n}$ algorithm that, given a (3,3)-CNF formula $\phi$ on $n$ variables and $m$ clauses, can distinguish between the cases where (i) $\phi$ is satisfiable or (ii) all variable assignments violate at least $\delta m$ clauses, unless Gap-ETH fails.
\end{corollary}

\subsection{Lower bound for approximating \textsc{Euclidean TSP}}

In this subsection we prove the following Theorem, which will conclude the proof of Theorem~\ref{thm:tsp}.

\begin{theorem}\label{thm:tsplower}
    For any $d$ there is a $\gamma>0$ such that there is no
    $2^{\gamma/\eps^{d-1}}\poly(n)$ time $(1+\eps)$-approximation algorithm for \textsc{Euclidean TSP} in
    $\Reals^d$, unless Gap-ETH fails.
\end{theorem}

We show that the reduction given in~\cite{frameworkpaperjournal} from
\textsc{(3,3)-SAT} to \textsc{Euclidean TSP} (see also the equivalent reduction for
\textsc{Hamiltonian Cycle} in~\cite{SKBdoktori})  can also be regarded as a reduction
from \textsc{Max-(3,3)SAT}, and it gives us the desired bound. We start with the
short summary of the construction from~\cite{frameworkpaperjournal}.

The construction of \cite{frameworkpaperjournal} heavily builds on the construction
of~\cite{itai1982hamilton} for \textsc{Hamiltonian Cycle} in grid graphs and
\cite{plesn1979np} for \textsc{Hamiltonian Cycle} in planar graphs. A basic familiarity with the lower bound framework \cite{frameworkpaperjournal} as well as the reductions in \cite{plesn1979np} and \cite{itai1982hamilton} is recommended for this section.

Overall, the construction of~\cite{frameworkpaperjournal} takes a $(3,3)$-CNF
formula $\phi$ as input, and in polynomial time creates a set of points
$P\subset \Reals^d$ where each point has integer coordinates, and $P$ has a
tour of length $|P|$ if and only if $\phi$ is satisfiable. The set $P$ can be
decomposed into \emph{gadgets}, which are certain smaller subsets of $P$.

We will use a notation proposed by~\cite{marx07} to describe properties of
gadgets. We say that a set of walks is a \emph{traversal} if
each point of a given gadget is visited by at least one of the walks. Note that for a
given gadget a TSP tour induces a traversal simply by taking edges adjacent to
the points of a gadget.

Each gadget $G$ has a set of \emph{visible points} $S \cup T \subseteq G$. A \emph{state} $q$ of gadget $G$ with visible points  $S \cup T \subseteq G$ is a collection of pairs $(s^q_i,t^q_i)$ for $i\in [k]$ where $s^q_i\in S$ and $t^q_i\in T$. 
We say that a traversal $\mathcal{W} = \{W_1,\ldots,W_k\}$ \emph{represents} state
$q$ if walk $W_i$ starts in $s^q_i$ and ends in $t^q_i$ for each $i\in [k]$.
The set of allowed states form the \emph{state space} $Q$ of the gadget. Finally, for a given TSP tour $\pi$ and a gadget $G$ we say that a traversal
induced on $G$ by an (bidirected) tour $\pi$ has the following \emph{weight}.

\begin{displaymath}
    \wt(T,G) \coloneqq  \sum_{\substack{p \in G,\\ (p,x) \in \pi}} \frac{\lVert p - x \rVert_2}{2}  - |G|,
\end{displaymath}
where $(p,x)$ are ordered pairs, i.e., edges induced by $G$ are counted twice.
Hence if a traversal visits all vertices exactly once and all edges are of
length $1$, then the weight of the traversal is $0$. Note that the input/output
edges contribute $1/2$ to the weight of the traversal.

Recall that in the construction developed in~\cite{frameworkpaperjournal}, the points are
placed on a grid of integral coordinates. The tour that traverses a gadget in
a ``bad'' way (i.e., in a way that does not correspond to a state of the gadget)
has to either visit a point more than once or it must use some diagonal edge of length
at least $\sqrt{2}$. Hence a weight of a traversal that does not represent any state of
the gadget needs to have weight at least $\frac{\sqrt{2} - 1}{2}$.

\begin{observation}\label{obs:traversal}
    Every gadget $G$ with state space $\mathcal{Q}$ developed in
    \cite{frameworkpaperjournal} has two properties:
    \begin{enumerate}[align=left, font=\normalfont, label=(\roman*)]
        \item for every state $q \in \mathcal{Q}$ there exists a traversal $\mathcal{W}_q$ of weight $0$ that represents $q$, and 
        \item any traversal $\mathcal{W}$ that does not represent any $q \in \mathcal{Q}$ is of weight at least $\frac{\sqrt{2} -1}{2}$.
    \end{enumerate}
\end{observation}

Now, we are ready to describe more concretely the gadgets in
\cite{frameworkpaperjournal}. There are size-$3$ and size-$2$ clause gadgets
with state spaces $\mathcal{Q}_3 = \mathbb{Z}_2^3 \setminus (0,0,0)$ and
$\mathcal{Q}_2 = \mathbb{Z}_2^2 \setminus (0,0)$ respectively. Both types of
clause gadgets consist of a constant number of points with integer coordinates
in $\{0,\ldots,c_0\}^d$ (translated appropriately). Clause gadgets are used to
encode clauses of the \textsc{Max-(3,3)SAT} instance. Similarly,
\cite{frameworkpaperjournal} developed a \emph{Variable Gadget} for state
space $\mathcal{Q} = \mathbb{Z}_2$, that is used to encode the values of variables
in the \textsc{Max-(3,3)SAT} instance.

A \emph{wire} is a constant width grid path with state space
$\mathbb{Z}_2$. It is used to transfer information from a
Variable gadget to a Clause Gadget (note that a wire does have a constant number of
points). In the $2$-dimensional case \cite{frameworkpaperjournal} define a \emph{crossover gadget} that has
state space $\mathbb{Z}_2 \times \mathbb{Z}_2$ that is able to transfer information
both horizontally and vertically. It is added in the junction of two crossing wires
in order to enable transfer of two independent bits of information. We and
\cite{frameworkpaperjournal} do not need crossing gadgets in higher dimensions.

In the reduction of \cite{frameworkpaperjournal} Clause Gadgets are connected
with Variable Gadgets by wires. When a gadget and a wire are connected, then
they always share a constant number of points. Clauses are connected to
Variables in the natural way: Namely, let $T_\text{OPT}$ be the optimal TSP
path. Any clause gadget $\phi = x^*_1 \vee x^*_2 \vee x^*_3$ where
$x^*_i\in \{x_i,\neg x_i\}\; (i=1,2,3)$ is connected by a wire to the variable
gadgets $x_1,x_2,x_3$. Moreover if a subpath of the optimal tour
$T_\text{OPT}$ goes through a clause gadget $\phi$ and represents a state
$(y_1,y_2,y_3)$, then the traversal of $P$ inside the wire connecting $\phi$
with $x_i$ represents state $y_i$, and the traversal of $T_\text{OPT}$
inside the gadget of $x_i$ represents state $y_i$ if $x_i$ has a positive
literal in this clause and $\neg y_i$ if it has a negative literal there.

The final detail that \cite{frameworkpaperjournal} needs is to place all the
gadgets along a cycle. They add a ``snake'' (a width 2-grid path based
on~\cite{itai1982hamilton})  through variable and clause gadgets
(see~\cite[Figure 8.10]{SKBdoktori} for a schematic picture of construction in
$3$-dimensions). A snake is used to represent a long graph edge. It has two
states, corresponding to the long edge being in the Hamiltonian Cycle or not.
Note that every point in the construction is part of one gadget or gadget and
a wire or a gadget and a snake, and distinct gadgets have distance more than $1$.

We are now ready to prove the lower bound for \textsc{Euclidean TSP}.

\begin{proof}[Proof of Theorem~\ref{thm:tsplower}]
We use a reduction from \textsc{Max-(3,3)SAT}. For an input point set $P$ let
$OPT$ denote the minimum tour length. Suppose for the sake of contradiction
that for every $\gamma>0$ there is an algorithm that for any point set $P$ and
$\eps>0$ returns a traveling salesman tour of length at most $(1+\eps)OPT$ in
$2^{\gamma/\eps^{d-1}}\poly(n)$ time. Fix some integer $d\geq 2$, and let
$\phi$ be a (3,3)-CNF formula, and apply the construction
of~\cite{frameworkpaperjournal} to obtain a point set $P\subset \Reals^d$
satisfying Observation~\ref{obs:traversal}. Note that $P$ has a TSP tour of
length $|P|$ if and only if $\phi$ is satisfiable. Let $c$ be such that
$|P|=cn^{d/(d-1)}$.

Suppose now that $P$ has a TSP tour $T_\text{apx}$ of length $(1+\eps)|P|$.  
We mark a gadget $G$ to be \emph{destroyed} if the traversal of $T_\text{apx}$ does
not represent any state of the gadget. By the properties of the gadget, such a
traversal has weight at least $\frac{\sqrt{2} - 1}{2}$. Therefore, there can be at most 
$\frac{4\eps|P|}{\sqrt{2}-1}$ destroyed gadgets (note that $1$ edge of
$T_\text{apx}$ can be a part of at most $2$ traversals).

For a fixed variable $x$ let $V_x$ be the variable gadget that encodes it.
We will mark a variable $x$ ``bad'' if one of the following conditions holds:

\begin{itemize}
    \item the gadget $V_x$ is destroyed, or
    \item any of the wires or snakes connecting to $V_x$ is destroyed, or
    \item a crossover gadget on one of the wires of $V_x$ is destroyed, or
    \item any clause that is connected to $V_x$ is destroyed.
\end{itemize}

Consequently, one destroyed gadget may result in up to $3$ variables being
marked bad (in case the destroyed gadget corresponds to a size-3 clause).
Since there are at most $\frac{4\eps|P|}{\sqrt{2}-1}$ destroyed gadgets, we
can have at most $12\frac{\eps|P|}{\sqrt{2}-1}<30\eps|P|$ bad variables.
Therefore, for all variables that have not been marked ``bad'', as well as the
connected wires, snakes, crossovers, and clause gadgets only have incident
edges of length $1$. Just as in the original construction, we can use the
length $1$ edges of the tour in these variable gadgets to define a partial
assignment for the non-bad variables. This partial assignment is guaranteed to
satisfy all the clauses that contain only non-bad variables. Since each
variable occurs in at most $3$ clauses, we have at most $90\eps|P|$ clauses
that have a bad variable, so the partial assignment for the non-bad variables
will satisfy at least $m-90\eps|P|=\left(1-90\eps c
\frac{n^{d/(d-1)}}{m}\right) m$ clauses.

We can now set $\eps=\frac{\delta m}{90 c n^{d/(d-1)}}$. Since $m=\Theta(n)$,
we have that $\eps=\Theta(1/n^{1/(d-1)})$. We can now apply the approximation
algorithm for \textsc{Euclidean TSP} with the above $\eps$ on $P$. As a
result, we can distinguish between a satisfiable formula (and thus a tour of
length $|P|$) and a formula in which all assignments violate at least $\delta
m$ clauses, where therefore any tour has length more than $(1+\eps)|P|$. Since
the construction time of $P$ is polynomial in $n$, the total running time of
this algorithm is $2^{\gamma/\eps^{d-1}}\poly(n)= 2^{\gamma c' n}$ for some
constant $c'$. The existence of such algorithms for all $\gamma>0$ would
therefore violate Gap-ETH  by Corollary~\ref{cor:max33sat}.
\end{proof}

\subsection{Lower bound for approximating \textsc{Rectilinear Steiner tree}}
\label{rectilinear-steiner-tree}

We will now prove the lower bound of Theorem~\ref{thm:st}, which can be stated as follows.

\begin{theorem}\label{thm:rectsteinerlower}
For any $d\geq 2$ there is a $\gamma>0$ such that there is no $(1+\eps)$-approximation algorithm for
\textsc{Rectilinear Steiner Tree} in $\Reals^d$ that has running time
$2^{\gamma/\eps^{d-1}}\poly(n)$, unless Gap-ETH fails.
\end{theorem}

The proof of Theorem~\ref{thm:rectsteinerlower} has three stages. In the first
stage, we give a reduction (in several steps) from \textsc{Max-(3,3)SAT},
which converts a $(3,3)$-CNF formula $\phi$ to a variant of connected vertex
cover on graphs drawn in a $d$-dimensional grid. In the second stage, given
such a connected vertex cover instance, we create a point set $P\subset
\Reals^d$ in polynomial time. A satisfiable formula $\phi$ will correspond to
a minimum connected vertex cover, which will correspond to minimum rectilinear
Steiner tree. The harder direction will be to show that from a
$(1+\eps)$-approximate rectilinear Steiner tree $T$ we can find a good
connected vertex cover and therefore a good assignment to $\phi$. Before we
can define a connected vertex cover based on the tree $T$, we need to show
that we can \emph{canonize} $T$, i.e., to modify parts of $T$ in a manner that
does not lengthen $T$, and at the same time makes its structure much simpler.
In the final stage, we use the canonized tree $T$ and an argument similar to
the one seen for \textsc{Euclidean TSP} above to wrap up the proof.

\subsubsection{From \textsc{Max-(3,3)SAT} to \textsc{Connected Vertex Cover}}

The construction begins in a slightly different manner for $d=2$ and for
$d\geq 3$, but the resulting constructions will share enough properties so
that we will be able to handle $d\geq 2$ in a uniform way in later parts of
this proof.

Let $\phi$ be a fixed $(3,3)$-CNF formula on $n$ variables, and let $G$ be its
\emph{incidence graph}, i.e., $G$ has one vertex for each variable and one
vertex for each clause of $\phi$, and a variable vertex and clause vertex are
connected if and only if the variable occurs in the clause.

A \emph{grid cube of side length $\ell$} is a graph with vertex set $[\ell]^d$
where a pair of vertices is connected if and only if their Euclidean distance
is $1$. We say that a graph is \emph{drawn in a $d$-dimensional grid cube of
side length $\ell$} if its vertices are mapped to distinct points of
$[\ell]^d$ and its edges are mapped to vertex disjoint paths inside the grid
cube.

Given a graph $G=(V,E)$, a vertex subset $S\subset V$ is a \emph{vertex cover}
if for any edge $e\in E$ there is a vertex incident to $e$ in $S$. The set $S$
is a \emph{connected vertex cover} if $S$ is a vertex cover and the subgraph
induced by $S$ is connected. The \textsc{Vertex Cover} problem is to find the
minimum vertex cover of a given graph on $n$ vertices, while \textsc{Connected
Vertex Cover} seeks the minimum connected vertex cover. If $G$ is restricted
to be in the class of graphs that can be drawn in an $n\times n$ grid, then the
corresponding problems are called \textsc{Grid Embedded Vertex Cover} and
\textsc{Grid Embedded Connected Vertex Cover}. (Note that the graph itself may
have up to $n^2$ vertices in these grid embedded problems.)

\paragraph{Grid embedding in $\Reals^2$}
Given $\phi$, \cite{frameworkpaperjournal} constructs a CNF formula $\phi'$ on
$\Oh(n^2)$ variables such that the incidence graph $G'$ of $\phi'$ is planar and
it can be drawn in $[cn]^2$ for some constant $c$, and each variable of
$\phi'$ occurs at most $3$ times, and each clause has size at most $4$. By
introducing a new variable for each clause of size $4$, we can replace a
clause $(x_1\vee x_2 \vee x_3 \vee x_4)$ with the clauses $(x_1\vee x_2 \vee
y) \wedge (\neg y \vee x_3 \vee x_4)$, and this corresponds to dilating the
original clause vertex and subdividing it with the variable vertex of $y$ in
$G'$. One can then modify the drawing of $G'$ accordingly. (The drawing of
$G'$ may need to be \emph{refined}, i.e., scaled up by a factor of $3$, while keeping the underlying grid unchanged, to provide enough space for
the new vertices.) As a result, we get a $(3,3)$-CNF formula $\phi_2$ whose
incidence graph $G_2$ is planar and has a drawing in the grid $[cn]^2$ for
some constant~$c$.

\begin{lemma}\label{lem:phi2dim}
The formula $\phi$ is satisfiable if and only of $\phi_2$ is satisfiable.
If $\phi_2$ has an assignment that satisfies all
but $t$ clauses, then $\phi$ has an assignment that satisfies all but
$6t$ clauses.
\end{lemma}

\begin{proof}
The first statement follows from the construction. See
also~\cite{frameworkpaperjournal,Lichtenstein82}. For the second statement, we
simply restrict the assignment to the set of variables that are also present
in $\phi$; let us call these original variables. Note that an unsatisfied
clause within a crossing gadget of $\phi_2$ might make two variables ``bad''
(see the proof of Theorem~\ref{thm:tsplower} for a similar argument). Since
each variable occurs at most $3$ times, this means that up to $6$ clauses may
become unsatisfied. As all clauses either occur in a crossing gadget or are
inside $\phi$ itself, having $t$ unsatisfied clauses in $\phi_2$ means that
there can be at most $6t$ clauses that are unsatisfied by the assignment.
\end{proof}

\begin{figure}[t]
\centering
\includegraphics[width=\textwidth]{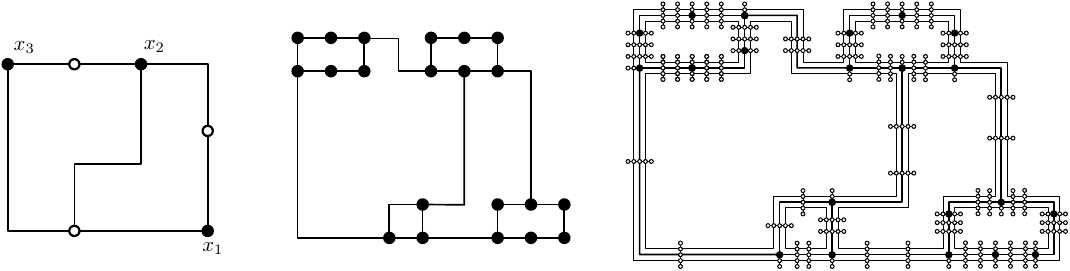}
\caption{Left: incidence graph of $\phi=(x_1\vee \neg x_2 \vee x_3) \wedge (x_1 \vee x_2) \wedge (x_2 \vee \neg x_3)$ drawn in a grid. Middle: an instance of \textsc{Vertex Cover} where variables are replaced with length-$6$ variable cycles, and size-$3$ clauses are replaced with triangles. Right: adding a skeleton (see~\cite{garey1977rectilinear}) to get an instance of \textsc{Connected Vertex Cover}.}\label{fig:cvcskeleton}
\end{figure}

Initially, we mostly follow the first few steps of the connected vertex cover
construction in \cite{frameworkpaperjournal}. Namely, refine this
grid drawing by a factor of $4$, that is, we scale the drawing from the origin
by a factor of $4$ while keeping the underlying grid unchanged. This allows us
to replace the vertices corresponding to variables with cycles of length $6$,
where the selection of odd or even vertices on the cycle corresponds to
setting the variable to true or false. We connect graph edges corresponding to
true literals to distinct even vertices and graph edges corresponding to false
literals to distinct odd vertices. Each vertex corresponding to the size $3$
clause is replaced by a cycle of length $3$, with one incoming edge for each of the
literals. For vertices corresponding to the size $2$ we subdivide one of the incident
edges into a path of length $3$. Finally, vertices corresponding to the clauses of size one can be removed
in a preprocessing step. Let $G^*_2$ be the resulting graph, which is drawn in
an $\Oh(n)\times \Oh(n)$ grid. In particular, if $G^*_2$ has $n^*_2$ vertices,
then $n^*_2=\Oh(n^2)$.

Note that each variable cycle needs at least $3$ of its vertices in the vertex
cover, each size-$3$ clause needs at least two vertices of its triangle in the
vertex cover, and each size-$2$ clause needs at least two internal vertices of
its path in the vertex cover. A size $3$-clause should be satisfied by some
literal, which would mean that the edge corresponding to this literal would be
covered from the variable cycle, and therefore it is sufficient to select the
other two vertices of the triangle. In a size-$2$ clause at least one of the
contained literals should be true, which exactly corresponds to the fact that
one of the endpoints of the corresponding edge has to be selected. Since there
is a path of length $5$ connecting these two vertices, we need to select the
two odd or even index internal vertices from it. With these at hand, the
following is a simple observation.

\begin{lemma}\label{lem:vc2dim}
The formula $\phi_2$ has a satisfying assignment if and only if $G^*_2$ has a
vertex cover of size $k^*_2=3n'+2m'=\Oh(n^2)$, where $n'$ and $m'$ are the
number of variables and size-$3$ clauses in $\phi_2$. If $G^*_2$ has a vertex cover of
size $k^*_2+t$, then $\phi_2$ has an assignment that satisfies all but at most
$9t$ clauses.
\end{lemma}

\begin{proof}
If $\phi_2$ is satisfiable, then we select the true or false vertices on the
variable cycles according to the assignment. In each clause, there is at least
one literal that is true; we select the vertices on the clause cycle that
correspond to the other two literals. The resulting set is clearly a vertex
cover. On the other hand, every vertex cover must have at least $3n'+2m'$
vertices, as in order to cover the variable cycles and the clause triangles
individually, one needs at least $3$ vertices per variable cycle and at least
$2$ vertices per clause triangle. Such a vertex cover in addition will select
only even or odd vertices from variable cycles, which yields a variable
assignment that satisfies $\phi_2$.

If $G^*_2$ has a vertex cover of size $k^*_2+t$, then there can be at most $t$
variable cycles or clause triangles where the number of vertices selected is
more than $3$ (respectively, $2$). Therefore we can mark "bad" any variable
whose cycle has more than $3$ vertices or that appears in size-$3$ clause
whose triangle has all vertices selected. Consequently, we have at most $3t$
bad variables. Since $3t$ variables can appear in at most $9t$ clauses, all
but at most $9t$ clause triangles will have exactly two vertices selected. One
can check that the assignment on the non-bad variables will then satisfy all
but $9t$ clauses.
\end{proof}

As a final step for the planar construction, we introduce the \emph{skeleton}
described by Garey and Johnson~\cite{garey1977rectilinear}; this again requires
that we refine the drawing by a constant factor. The procedure subdivides each
edge of the graph twice, using $n_{sub}$ new vertices, and also adds
$n_{skel}$ \emph{skeleton vertices}. An important property of the skeleton is
that the number of newly added vertices is
$n_{skel}+n_{sub}=\Theta(|V(G^*_2)|+|E(G^*_2|)=\Oh(n^2)$. The resulting graph is
drawn in an $\Oh(n)\times \Oh(n)$ grid. We use a final 4-refinement to ensure that
inside the $\ell_1$-disk of radius $4$ around each vertex $v$ the only grid edges
being used are on the horizontal or vertical line going through $v$. Let $G_2$
denote the resulting plane graph (i.e., the graph together with its embedding
in the plane).

\begin{lemma}\label{lem:cvc2dim}
The graph $G^*_2$ has a vertex cover of size $k^*$ if and only if $G_2$ has a
connected vertex cover of size $k_2\coloneqq k^*+(n_{skel}+n_{sub})/2$. If $G_2$ has a
connected vertex cover of size $k_2+t$, then $G^*$ has a vertex cover of size
$k^*+t$.
\end{lemma}

\begin{proof}
The construction of Garey and Johnson~\cite{garey1977rectilinear} has the
properties that (i) any connected vertex cover of $G_2$, when restricted to
the vertices of $G^*_2$, is a vertex cover of $G^*_2$, and (ii) one can add
$(n_{skel}+n_{sub})/2$ vertices among the subdivision and skeleton vertices to
any vertex cover of $G^*_2$ to get a connected vertex cover of $G_2$.  The
first claim follows directly from the above properties. For the second claim,
we note that the $n_{skel}$ skeleton vertices come in pairs, where one vertex
in the pair has degree one and is connected only to the other vertex.
Therefore, at least one vertex in each pair must be selected in every vertex
cover. Similarly, the vertices in $n_{sub}$ also come in pairs, each pair
being connected to each other, therefore one must select at least one vertex
from each pair into any vertex cover. Now consider a connected vertex cover of
size $k_2+t$ in $G_2$. Since there must be at least $(n_{skel}+n_{sub})/2$
vertices selected among the vertices newly introduced in $G_2$, there can be
at most $k_2+t-(n_{skel}+n_{sub})/2 = k^*+t$ vertices selected among the
original vertices in $G^*_2$.  It follows that $G^*_2$ has a vertex
cover of size $k^*+t$.
\end{proof}

\paragraph{Grid embedding in $\Reals^d$ for $d\geq 3$.}

We again start with the incidence graph $G$ of $\phi$, and let $L=|E(G)|$
denote the number of literal occurrences in $\phi$. Let $G^*$ be the graph
where variable vertices are replaced with variable cycles and clause vertices
by triangles (or for size $2$ clauses, paths) in the same manner as seen
in $G^*_2$. We will now define a different type of skeleton for these graphs.
First, we subdivide each edge of $G^*$ twice, that is, we remove the edge $vw$,
add the vertices $u',v'$, and add the edges $uu', u'v'$, and $v'v$. Let
$G^{**}$ be the resulting graph, and let $n^{**}$ denote the number of its
vertices. Consider the disjoint union of $G^{**}$ and the cycle graph $C^*$
with $n^{**}$ vertices. For each vertex $v$ of $G^{**}$, we can associate a
distinct vertex $v'$ in the cycle, and we also create a new vertex $v''$.
Finally, we add the edges $vv'$ and $v'v''$. It is routine to check that the
resulting graph has $\Oh(n)$ vertices and $\Oh(n)$ edges, and it has maximum
degree $4$. Therefore, we can apply the following theorem, which is
paraphrased from~\cite{frameworkpaperjournal}.

\begin{theorem}[Cube Wiring Theorem~\cite{frameworkpaperjournal}]
There is a constant $c$ such that for any $d\geq 3$ it holds that any graph
$G$ of maximum degree $2d$ on $n$ vertices can be drawn in a $d$-dimensional
grid cube of side length $cn^{1/{d-1}}$. Moreover, given $G$ and $d$ the
embedding can be constructed in polynomial time.
\end{theorem}

Since the graph has maximum degree $4< 2d$ and $\Oh(n)$ vertices and edges,
the resulting drawing is in a grid cube of side length $\Oh(n^{1/(d-1)})$. For a vertex $v$, let $\ell_1(v)$ and
$\ell_2(v)$ be the lines that are parallel to the first and second coordinate axis
respectively and pass through $v$. We
use a constant-refinement and reorganize the neighborhood of each vertex $v$
in the grid drawing so that the grid edges used by the drawing in the
$\ell_1$-ball of radius $4$ around $v$ all fall on $\ell_1$ and $\ell_2$.
Let $G_d$ denote the resulting graph
together with the obtained grid drawing.

One can prove the analogue of Lemma~\ref{lem:vc2dim} for $G^*$ (with $\phi$
instead of $\phi_2$), and also the analogue of Lemma~\ref{lem:cvc2dim} for
$G_d$. 

Putting Lemmas~\ref{lem:phi2dim}, \ref{lem:vc2dim} and~\ref{lem:cvc2dim} together,
and putting the higher dimensional analogues of
Lemmas~\ref{lem:vc2dim}, \ref{lem:cvc2dim} together, we get the following corollary, which is all that we will need from this subsection for the proof.

\begin{corollary}\label{cor:construction}
For each $d\geq 2$ the following holds. Given a $(3,3)$-CNF formula $\phi$ on
$n$ variables, we can generate a graph $G_d$ of degree at most $4$ drawn in a
$d$-dimensional grid of side length $\Oh(n^{1/(d-1)})$ in polynomial time such that (a) if $\phi$ is satisfiable then $G_d$
has a connected vertex cover of size $k_d=\Oh(n^{d/(d-1)})$ and (b) if $G_d$ has a
connected vertex cover of size $k_d+t$, then $\phi$ has an assignment
satisfying all but $\Oh(t)$ clauses.
\end{corollary}

\subsubsection{Construction and canonization}

Given any graph $G_d$ that is drawn in a $d$-dimensional grid, we construct a
point set $P_d$ the following way.

\begin{enumerate}
	\item Refine the drawing of $G_d$ by a factor of $137$. (The constant $137$ will be justified in Lemma~\ref{lem:canonize2}.)
	\item Add all grid points that are internal to edges of $G_d$ to $P_d$.
	\item Remove any point from $P_d$ that is at distance strictly less
	      than $2$ from a vertex of $G_d$.
\end{enumerate}

We call the set of points in $P_d$ that fall on an edge of $G_d$ \emph{edge
components}, and the $l_1$ balls of radius $2$ around
vertices of $G_d$ are the \emph{cutouts}. (The cutouts are cubes of diameter $4$ whose diagonals are axis-parallel; i.e., they are regular ``diamonds'' in $\Reals^2$.)

The resulting set $P_d$ is our construction for \textsc{Rectilinear Steiner
Tree}. The following lemma can be found in~\cite{garey1977rectilinear}, but we
provide a proof for completeness.

\begin{figure}[t]
    \centering
    \includegraphics{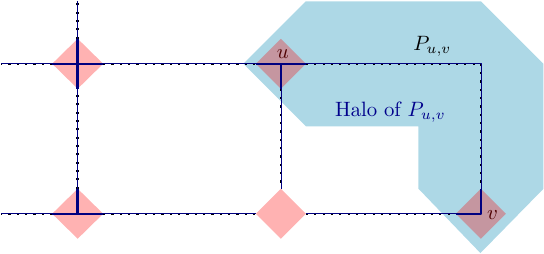}
    \caption{Construction of components from \cite{garey1977rectilinear} in $\Reals^2$. Dotted edges are edge
    components, and the red rotated squares are the cutouts. The blue object
    represents the halo of the edge component $P_{u,v}$ in the top-right corner. Blue (thin and thick) edges correspond to the edges of a canonical Steiner tree.}
    \label{fig:rectilinear}
\end{figure}

\begin{lemma}[Garey and Johnson~\cite{garey1977rectilinear}]\label{lem:Steiner_cvc_easy}
If $G_d$ has a connected vertex cover of size $k_d$, then $P_d$ has a
rectilinear Steiner tree of total length $\ell_d\coloneqq L+2|E(G_d)| +
2(k_d-1)$, where $L$ is the total length of the edge components. 
\end{lemma}

\begin{proof}
Given a connected vertex cover $S$ of size $k_d$, we construct a Steiner tree
$T$ the following way. First, we add all length $1$ edges connecting
neighboring vertices in edge components to $T$; these have total length $L$.
Since $S$ induces a connected graph, it has a spanning tree that has $k_d-1$
edges. The total length of these edges is $4(k_d -1)$. On each such edge
component, we add the length $2$ edges that connect to the incident cutouts on
both ends. Furthermore, on each edge $e$ of $G_d$ that is not an edge of this
spanning tree, there is at least one endpoint $s$ of $e$ that is in $S$ as $S$
is a vertex cover. We add a length $2$ cutout edge connecting the edge
component of $e$ to the center of the cutout of $s$. These edges have a total
length of $2(|E(G_d)| - k_d +1)$. The result is a tree $T$ on $P_d$ that has
total length
\[L+4(k_d -1) + 2(|E(G_d)| - k_d +1) = L+2|E(G_d)| + 2(k_d-1).\qedhere\]
\end{proof}

In the remainder of this subsection, we will canonize an approximate Steiner tree of $P_d$ in order to prove the following lemma.

\begin{lemma}\label{lem:Steiner_cvc_hard}
If $P_d$ has a rectilinear Steiner tree of length $\ell_d+\ell'$, then $G_d$ has a connected vertex cover of size $k_d + \ell'/2$.
\end{lemma}

\paragraph{Canonization}

We say that a Steiner tree of $P_d$ is \emph{canonical} if (i) every length
$1$ edge in edge components is included in $T$ and (ii) all other segments in
$T$ have length $2$ and connect the center of some cutout to the nearest
point in one of the incident edge components. 

A Steiner point of $T$ is a point that has degree at least $3$. The vertices
of $T$ are its Steiner points and $P_d$. An edge of $T$ is a curve connecting
two vertices. It follows that $T$ has at most $|P_d|-2$ Steiner points.

For a fixed constant $d\ge 2$ we can change each edge of
$T$ so that it is a minimum length path in the $\ell_1$ norm between these two
vertices, by moving parallel to the $x_1$ axis, then to the $x_2$ axis, etc.
until we arrive at the destination. The resulting tree consists of $\Oh(|P_d|)$
axis-parallel segments.

\begin{lemma}~\label{lem:canonize1}
For any approximate Steiner tree $T$ on $P_d$ there is a Steiner
tree $T'$ so that it contains all length-$1$ edges in edge components and all
of its edges have length at most $4$, and $T'$ is no longer than $T$.
\end{lemma}

\begin{proof}
Recall that the Hanan-grid~\cite{hanan1966steiner} of $P_d$ is the set of
points $H_d\subset \Reals^d$ that can be defined as the intersection of $d$ distinct
axis-parallel hyperplanes incident to $d$ (not necessarily distinct) points of
$P_d$. Snyder~\cite{Snyder92} shows that there exists a minimum rectilinear
Steiner Tree whose Steiner points are on the Hanan-grid. In fact, he proposes
modifications that are local, and can be applied also to a non-optimal Steiner
tree, i.e., affect only a vertex of a Steiner tree and its neighbors, and
result in a tree whose Steiner points lie on the Hanan-grid. None of the local
modifications lengthen the tree. Moreover, each type of local modification
either shortens $T$ by removing at least one segment, or it does not shorten
the tree but it can be repeated no more than $\Oh(|P_d|)$ times. Consequently,
given a tree $T$, we can use the local modifications of Snyder exhaustively to
get a tree that is not longer than $T$, and whose Steiner points lie in the
Hanan-grid of $P_d$.

Suppose now that $T$ is a rectilinear Steiner tree of $P_d$ whose Steiner
points are in $H_d$. Notice that the minimum distance between points of $H_d$
is $1$, and the minimum distance between points from two distinct edge
components is at least $4$. Suppose that there is an edge $uv$ of length $1$
in an edge component that is not in $T$. Then adding $uv$ to $T$ creates a
cycle, and that cycle has at least one edge $e$ that has either length more
than $1$, or $e$ has length exactly one but its endpoints are not covered by
any edge component. Replacing $e$ with $uv$ therefore results in a tree that is
no longer than $T$ and has one more length-$1$ edge that is inside an edge
component. Repeating the above check $\Oh(|P_d|)$ times results in a tree $T$
that contains all length-$1$ edges in edge components. Suppose now that $T$
has an edge $uv$ of length more than $4$. Removing the edge $uv$ from $T$
creates a forest of two trees. Suppose that one tree spans $P\subset P_d$ and
the other spans $Q\subset P_d$. Note that $P$ and $Q$ are non-empty, disjoint,
and their union is $P_d$. Observe that there exists $p\in P$ and $q\in Q$ such
that $\lVert p-q\rVert_1 \leq 4$. Now the shortest edge connecting $p$ and $q$
is shorter than $e$ was, therefore by adding this edge the created tree is not
longer than the original. As each such modification decreases the number of
edges of length more than $4$ and there are only $\Oh(|P_d|)$ edges in $T$, we
can remove all edges longer than $4$ in $\Oh(|P_d|^2)$ time. The resulting tree
$T'$ satisfies the required properties.
\end{proof}

\begin{lemma}~\label{lem:canonize2}
For any approximate Steiner tree $T'$ on $P_d$ that contains all length-$1$
edges in edge components and no edges longer than $4$, there is a canonical
tree~$T''$ that is no longer than $T'$.
\end{lemma}

\begin{proof} A \emph{full Steiner subtree} of $T$ is a subtree of $T$ whose internal vertices
are Steiner points of $T$, and whose leaves are points of $P_d$.
Let $F$ be a full Steiner subtree of $T'$ with $k$ leaves.  The \emph{halo} of an
edge component $P_{uv}$ is the set of points in $\Reals^d$ whose
$\ell_1$-distance from $P_{uv}$ is at most $68$. The refinement step of the
construction of $P_d$ (see step (i)) ensures that two halos intersect if and
only if the corresponding edges are incident to the same vertex. 

Since $T'$ contains all length-$1$ edges in edge components, the role of $F$ is
to connect a certain set of edge components; each edge component contains at
most one leaf of $F$ (as otherwise there would be a cycle). Let $\gamma$ be
the number of edge components adjacent to $F$ in the tree, that is, $F$ has
$\gamma$ leaves. Let $\beta$ be the number of Steiner points in $F$. Notice
that $\beta < \gamma$.

First, we show that the leaves of $F$ are in edge components that are incident
to the same cutout. Let $\mu$ be the number of pairs $(p,H)$ where $p$ is a
Steiner point of $F$, and $H$ is a halo for one of the edge components
connected to $F$, and moreover $p\in H$. On one hand, every Steiner point can
be contained in at most $4$ halos, since that is the maximum overlap achieved
by the halos. This is a consequence of the fact that halos corresponding to
non-incident edge components are disjoint so any set of intersecting halos
must correspond to the neighborhood of a single vertex, and the maximum degree
of a vertex in $G_d$ is $4$. Therefore, we have that $\mu \leq 4 \beta$.

On the other hand, since the maximum degree of $G_d$ is $4$, there is a set
$E$ of at least $\gamma/4$ pairwise non-incident edge components; in
particular, there is a point $r$ in $F$ that is outside all the closed halos
corresponding to edge components in $E$. Let $r$ be the root of $F$, and for
an edge component $P_{uv}\in E$, consider the unique path in $F$ from the leaf
in $P_{u,v}$ to $r$. This path must intersect the halo of $P_{uv}$ at some
point, and the portion of this path within the halo of $P_{uv}$ has length at
least $68$. Notice that these paths are disjoint for each edge component of
$F$. Since edges of $F$ have length at most $4$, there must be at least
$68/4-1=  16$ Steiner points on each such path, so altogether we have $\mu
\geq (\gamma/4) \cdot 16 =4\gamma$. Putting the upper and lower bound on $\mu$
together, we have that $4\gamma \leq 4 \beta$. But this contradicts the fact
that $\beta < \gamma$.

Consider now a full Steiner subtree $F$ that is connecting $\gamma \in \{2,3,4\}$
adjacent edge components. Then $F$ could have length $2\gamma$, as we can connect
the point associated with the common vertex $v$ from the point set $P_1$ with
length two segments to the nearest vertex of all the edge components connected
by $F$. We show that this is the shortest possible tree for $\gamma$ edge
components.  Notice that any pair of edge components have $\ell_1$-distance
exactly four, so the shortest path in the rectilinear Steiner tree between any
pair of components is at least four. Now consider the geometric graph that we
get by doubling every segment in the tree, so that we get parallel edges
everywhere. This graph has an Euler tour (since every degree is even); on such
an Euler tour, the length required between any pair of tree leaf vertices is
at least the length of a shortest path between them, which is at least $4$.
Since the length of the tour is exactly twice the length of the tree, we get
that the tree has total length at least $4\gamma/2 = 2\gamma$, as claimed.
Therefore, we can exchange each full Steiner subtree with a canonical connection.
The resulting tree $T''$ is canonical and no longer than $T'$.
\end{proof}

We can now prove the correspondence between an approximate rectilinear Steiner
tree and a connected vertex cover.

\begin{proof}[Proof of Lemma~\ref{lem:Steiner_cvc_hard}]
Let $T$ be a Steiner tree of $P_d$ of length $\ell_d+\ell'$. Using
Lemma~\ref{lem:canonize1} and Lemma~\ref{lem:canonize2}, we can create a
canonical tree $T''$ in $n^{\Oh(1)}$ time of length at most $\ell_d+\ell'$. Let
$k$ be the number of non-empty cutouts. Observe that the vertices of $G_d$
corresponding to the non-empty cutouts form a connected vertex cover.

Consider the subtree $U$ of $T''$ that is spanned by the centers of the
non-empty cutouts. Every edge component in $U$ must be connected to the
centers of both neighboring cutouts, and $U$ contains $k-1$ edge components.
Furthermore, every edge component outside $U$ must be connected to the center
of at least one of the neighboring cutouts. Consequently, the length of $T''$
is at least $L+ 4(k-1) + 2(|E(G_d)|-k+1) = L + 2|E(G_d)| + 2(k-1)$. Therefore,
we have \[L + 2|E(G_d)| + 2(k-1)\leq \ell_d+\ell' = L + 2|E(G_d)| + 2(k_d-1) +
\ell',\] and thus $k\leq k_d + \ell'/2$, as required.
\end{proof}

\subsubsection{Concluding the proof of Theorem~\ref{thm:rectsteinerlower}}

\begin{proof}[Proof of Theorem~\ref{thm:rectsteinerlower}]
Putting Corollary~\ref{cor:construction} and Lemmas~\ref{lem:Steiner_cvc_easy}
and~\ref{lem:Steiner_cvc_hard} together, we have that if a $(3,3)$-CNF
formula $\phi$ on $n$ variables has a satisfying assignment then $P_d$ has a
rectilinear Steiner tree of length $\ell_d=\Oh(n^{d/(d-1)})$. Let $c_1$ be such
that $\ell_d=c_1n^{d/(d-1)}$. Additionally, if $P_d$ has a rectilinear Steiner
tree of length $\ell_d+\ell'$, then $\phi$ has an assignment that satisfies
all but $c_2\ell'$ clauses, where $c_2$ is a constant.

Suppose that there is a $2^{\gamma/\eps^{d-1}}\poly(n)$ algorithm for
\textsc{Rectilinear Steiner Tree} for all $\gamma>0$. Given a formula $\phi$,
we create the set $P_d$ in polynomial time, and we run the above algorithm with
$\eps=\frac{\delta m}{c_1 c_2 n^{d/(d-1)}}$, where $m$ is the number of
clauses in $\phi$. Since $m=\Theta(n)$, we have that
$\eps=\Theta(1/n^{1/(d-1)})$. We can now distinguish between a satisfiable
formula (when a rectilinear Steiner tree on $P_d$ is of length $\ell_d$) and a formula in which all
assignments violate at least $\delta m$ clauses (when any rectilinear
Steiner tree on $P_d$ has length greater than $(1+\eps)\ell_d$). Since the
construction time of $P_d$ is polynomial in $n$, the total running time of this algorithm is
$2^{\gamma/\eps^{d-1}}\poly(n)= 2^{\gamma c n}$ for some constant $c$. The
existence of such algorithms for all $\gamma>0$ would therefore violate
Gap-ETH by Corollary~\ref{cor:max33sat}.
\end{proof} 

\section{Conclusion and Open Problems}
\label{conclusion}

In this article we gave randomized $(1+\eps)$-approximation algorithms for \textsc{Euclidean TSP},
\textsc{Euclidean Steiner Tree} and \textsc{Rectilinear Steiner Tree} that run in
$2^{\Oh(1/\eps)^{d-1}}n + \poly(1/\eps) n\log n$ time. In case of \textsc{Euclidean TSP} and
\textsc{Rectilinear Steiner Tree}, we have shown that there are no
$2^{o(1/\eps)^{d-1}}\poly(n)$ algorithms under Gap-ETH. 
We achieved the improved algorithms by extending Arora's method~\cite{Arora98} with a new technique: \emph{Sparsity-Sensitive Patching}. 

As mentioned in the beginning of this paper, the methods from~\cite{Arora98,Mitchell99} have been greatly generalized and extended to various other problems by several authors. A natural direction for further research would be to see whether Sparsity-Sensitive Patching can also be employed to obtain improved (and possibly, Gap-ETH-tight) approximation schemes for these problems. Examples of problems where such a question can be studied include 
\begin{itemize}
	\item Euclidean versions of \textsc{Matching}, \textsc{$k$-TSP} and \textsc{$k$-Steiner Tree}~\cite{Arora98}, \textsc{Steiner Forest}~\cite{euclidean-steiner-forest}, $k$-\textsc{Connectivity}~\cite{euclidean-mincostconnectivity}, $k$-\textsc{Median}~\cite{ekmedian,kmedian} and \textsc{Survivable Network Design}~\cite{survivable},
	\item versions of some of the above problems in other metric spaces (e.g., doubling, hyperbolic), and in planar, surface-embedded and minor free graphs (see Section~\ref{subsec:relatedwork} for such studies for TSP).
\end{itemize} 

Since the publication of the conference version of this paper, our sparsity-sensitive patching was already used in the follow up papers~\cite{DBLP:conf/soda/Dross0WZ23, vanwijland} on a certain separation problem and \textsc{$k$-TSP}. For both these problems, spanners do not seem to be useable and therefore the sparsity-sensitive technique was especially instrumental to obtain the results.

There are several open questions worth exploring further. The ideal algorithm
for \textsc{Euclidean TSP} would have a running time of
$2^{\Oh(1/\eps)^{d-1}}n$, and it would be \emph{deterministic}. However, achieving this
running time with a randomized algorithm is already a challenging question. The
most natural way to pursue this would be to try and unify Bartal and Gottlieb's
techniques~\cite{BartalG13} with ours. Is it possible to get this running time
without spanners by using some new ideas to handle singletons (e.g., crossings that appear on their own on a facet of a quadtree cell)?

One could also pursue a $(1+\eps)$-approximation algorithm that uses $f(1/\eps)n^{\Oh(1)}$ time and only $\poly(1/\eps,n)$ space, but
this would likely require an algorithm that is not based on dynamic programming.
Is such an algorithm possible (say, for $d=2$)?

\section*{Acknowledgement}

We thank Hans L. Bodlaender for his double counting argument in
Lemma~\ref{lem:canonize2}, Artur Czumaj for sharing the full version
of~\cite{survivable} with us, and Hung Le for answering multiple questions regarding
spanners. 

\bibliographystyle{abbrv}
\bibliography{bib}

\begin{thebibliography}{10}

\bibitem{Arora98}
S.~Arora.
\newblock {P}olynomial {T}ime {A}pproximation {S}chemes for {E}uclidean
  {T}raveling {S}alesman and other {G}eometric {P}roblems.
\newblock {\em J. {ACM}}, 45(5):753--782, 1998.

\bibitem{Arora03}
S.~Arora.
\newblock Approximation schemes for {NP}-hard geometric optimization problems:
  a survey.
\newblock {\em Math. Program.}, 97(1-2):43--69, 2003.

\bibitem{tsp-planar-2}
S.~Arora, M.~Grigni, D.~R. Karger, P.~N. Klein, and A.~Woloszyn.
\newblock A polynomial-time approximation scheme for weighted planar graph
  {TSP}.
\newblock In {\em Proceedings of the Ninth Annual {ACM-SIAM} Symposium on
  Discrete Algorithms, 25-27 January 1998, San Francisco, California, {USA}},
  pages 33--41, 1998.

\bibitem{weighted-planar-graphs}
S.~Arora, M.~Grigni, D.~R. Karger, P.~N. Klein, and A.~Woloszyn.
\newblock A polynomial-time approximation scheme for weighted planar graph
  {TSP}.
\newblock In {\em Proceedings of the Ninth Annual {ACM-SIAM} Symposium on
  Discrete Algorithms ({SODA} 1998)}, pages 33--41, 1998.

\bibitem{kmedian}
S.~Arora, P.~Raghavan, and S.~Rao.
\newblock {A}pproximation {S}chemes for {E}uclidean \emph{k}-{M}edians and
  {R}elated {P}roblems.
\newblock In {\em Proceedings of the Thirtieth Annual {ACM} Symposium on the
  Theory of Computing, 1998}, pages 106--113, 1998.

\bibitem{baker}
B.~S. Baker.
\newblock {A}pproximation {A}lgorithms for {NP}-{C}omplete {P}roblems on
  {P}lanar {G}raphs.
\newblock {\em J. {ACM}}, 41(1):153--180, 1994.

\bibitem{BartalG13}
Y.~Bartal and L.~Gottlieb.
\newblock {A} {L}inear {T}ime {A}pproximation {S}cheme for {E}uclidean {TSP}.
\newblock In {\em 54th Annual {IEEE} Symposium on Foundations of Computer
  Science ({FOCS} 2013)}, pages 698--706. {IEEE} Computer Society, 2013.

\bibitem{BartalGK16}
Y.~Bartal, L.~Gottlieb, and R.~Krauthgamer.
\newblock {T}he {T}raveling {S}alesman {P}roblem: {L}ow-{D}imensionality
  {I}mplies a {P}olynomial {T}ime {A}pproximation {S}cheme.
\newblock {\em {SIAM} J. Comput.}, 45(4):1563--1581, 2016.

\bibitem{steiner-spanner-scg21}
S.~Bhore and C.~D. T\'{o}th.
\newblock {{L}ight {E}uclidean {S}teiner {S}panners in the {P}lane}.
\newblock In {\em 37th International Symposium on Computational Geometry (SoCG
  2021)}, volume 189, pages 15:1--15:17, 2021.

\bibitem{rank-based}
H.~L. Bodlaender, M.~Cygan, S.~Kratsch, and J.~Nederlof.
\newblock Deterministic single exponential time algorithms for connectivity
  problems parameterized by treewidth.
\newblock {\em Inf. Comput.}, 243:86--111, 2015.

\bibitem{euclidean-steiner-forest}
G.~Borradaile, P.~N. Klein, and C.~Mathieu.
\newblock {A} {P}olynomial-{T}ime {A}pproximation {S}cheme for {E}uclidean
  {S}teiner {F}orest.
\newblock {\em {ACM} Trans. Algorithms}, 11(3):19:1--19:20, 2015.

\bibitem{BorradaileLW17}
G.~Borradaile, H.~Le, and C.~Wulff{-}Nilsen.
\newblock Minor-free graphs have light spanners.
\newblock In {\em 58th {IEEE} Annual Symposium on Foundations of Computer
  Science ({FOCS} 2017)}, pages 767--778. {IEEE} Computer Society, 2017.

\bibitem{ChanJ18}
T.~H. Chan and S.~H. Jiang.
\newblock Reducing curse of dimensionality: Improved {PTAS} for {TSP} (with
  neighborhoods) in doubling metrics.
\newblock {\em {ACM} Trans. Algorithms}, 14(1):9:1--9:18, 2018.

\bibitem{DBLP:journals/ipl/Chan08}
T.~M. Chan.
\newblock Well-separated pair decomposition in linear time?
\newblock {\em Inf. Process. Lett.}, 107(5):138--141, 2008.

\bibitem{jansen-scheduling}
L.~Chen, K.~Jansen, and G.~Zhang.
\newblock On the optimality of exact and approximation algorithms for
  scheduling problems.
\newblock {\em J. Comput. Syst. Sci.}, 96:1--32, 2018.

\bibitem{fast-hamiltonicity-jacm}
M.~Cygan, S.~Kratsch, and J.~Nederlof.
\newblock {F}ast {H}amiltonicity {C}hecking {V}ia {B}ases of {P}erfect
  {M}atchings.
\newblock {\em J. {ACM}}, 65(3):12:1--12:46, 2018.

\bibitem{euclidean-mincostconnectivity}
A.~Czumaj and A.~Lingas.
\newblock {A} {P}olynomial {T}ime {A}pproximation {S}cheme for {E}uclidean
  {M}inimum {C}ost k-{C}onnectivity.
\newblock In {\em Automata, Languages and Programming, 25th International
  Colloquium ({ICALP} 1998)}, pages 682--694, 1998.

\bibitem{survivable}
A.~Czumaj, A.~Lingas, and H.~Zhao.
\newblock {P}olynomial-{T}ime {A}pproximation {S}chemes for the {E}uclidean
  {S}urvivable {N}etwork {D}esign {P}roblem.
\newblock In {\em Automata, Languages and Programming, 29th International
  Colloquium ({ICALP} 2002)}, pages 973--984, 2002.

\bibitem{Das97}
G.~Das, S.~Kapoor, and M.~H.~M. Smid.
\newblock {O}n the {C}omplexity of {A}pproximating {E}uclidean {T}raveling
  {S}alesman {T}ours and {M}inimum {S}panning {T}rees.
\newblock {\em Algorithmica}, 19(4):447--460, 1997.

\bibitem{BergBKK18}
M.~de~Berg, H.~L. Bodlaender, S.~Kisfaludi{-}Bak, and S.~Kolay.
\newblock An {ETH}-tight exact algorithm for {E}uclidean {TSP}.
\newblock In {\em Proceedings of the 59th {IEEE} Annual Symposium on
  Foundations of Computer Science ({FOCS} 2018)}, pages 450--461. {IEEE}
  Computer Society, 2018.

\bibitem{frameworkpaperjournal}
M.~de~Berg, H.~L. Bodlaender, S.~Kisfaludi{-}Bak, D.~Marx, and T.~C. van~der
  Zanden.
\newblock {A} {F}ramework for {E}xponential-{T}ime-{H}ypothesis-{T}ight
  {A}lgorithms and {L}ower {B}ounds in {G}eometric {I}ntersection {G}raphs.
\newblock {\em {SIAM} J. Comput.}, 49(6):1291--1331, 2020.

\bibitem{DemaineHK11}
E.~D. Demaine, M.~Hajiaghayi, and K.~Kawarabayashi.
\newblock {C}ontraction {D}ecomposition in {H}-{M}inor-{F}ree {G}raphs and
  {A}lgorithmic {A}pplications.
\newblock In {\em Proceedings of the 43rd {ACM} Symposium on Theory of
  Computing ({STOC} 2011)}, pages 441--450, 2011.

\bibitem{Dinur16}
I.~Dinur.
\newblock Mildly exponential reduction from gap 3{SAT} to polynomial-gap
  label-cover.
\newblock {\em Electron. Colloquium Comput. Complex.}, 23:128, 2016.

\bibitem{dw71}
S.~E. Dreyfus and R.~A. Wagner.
\newblock {T}he {S}teiner problem in graphs.
\newblock {\em Networks}, 1(3):195--207, 1971.

\bibitem{DBLP:conf/soda/Dross0WZ23}
F.~Dross, K.~Fleszar, K.~Wegrzycki, and A.~Zych{-}Pawlewicz.
\newblock Gap-eth-tight approximation schemes for red-green-blue separation and
  bicolored noncrossing euclidean travelling salesman tours.
\newblock In N.~Bansal and V.~Nagarajan, editors, {\em Proceedings of the 2023
  {ACM-SIAM} Symposium on Discrete Algorithms, {SODA} 2023, Florence, Italy,
  January 22-25, 2023}, pages 1433--1463. {SIAM}, 2023.

\bibitem{FeldmannSLM20}
A.~E. Feldmann, {Karthik {C. S.}}, E.~Lee, and P.~Manurangsi.
\newblock A survey on approximation in parameterized complexity: Hardness and
  algorithms.
\newblock {\em Algorithms}, 13(6):146, 2020.

\bibitem{spanner-doubling-metrics}
A.~Filtser and S.~Solomon.
\newblock The greedy spanner is existentially optimal.
\newblock {\em {SIAM} J. Comput.}, 49(2):429--447, 2020.

\bibitem{garey1977rectilinear}
M.~R. Garey and D.~S. Johnson.
\newblock The rectilinear steiner tree problem is {NP}-complete.
\newblock {\em SIAM Journal on Applied Mathematics}, 32(4):826--834, 1977.

\bibitem{bartal-steiner-tree}
L.~Gottlieb and Y.~Bartal.
\newblock {N}ear-linear time approximation schemes for {S}teiner tree and
  forest in low-dimensional spaces.
\newblock {\em Accepted to STOC 2021}.

\bibitem{tsp-planar-1}
M.~Grigni, E.~Koutsoupias, and C.~H. Papadimitriou.
\newblock An approximation scheme for planar graph {TSP}.
\newblock In {\em 36th Annual Symposium on Foundations of Computer Science,
  Milwaukee, Wisconsin, USA, 23-25 October 1995}, pages 640--645, 1995.

\bibitem{grigni95}
M.~Grigni, E.~Koutsoupias, and C.~H. Papadimitriou.
\newblock An approximation scheme for planar graph {TSP}.
\newblock In {\em 36th Annual Symposium on Foundations of Computer Science
  ({FOCS} 1995)}, pages 640--645, 1995.

\bibitem{gln02}
J.~Gudmundsson, C.~Levcopoulos, and G.~Narasimhan.
\newblock Fast greedy algorithms for constructing sparse geometric spanners.
\newblock {\em {SIAM} J. Comput.}, 31(5):1479--1500, 2002.

\bibitem{hanan1966steiner}
M.~Hanan.
\newblock {O}n {S}teiner’s problem with rectilinear distance.
\newblock {\em SIAM Journal on Applied Mathematics}, 14(2):255--265, 1966.

\bibitem{Har-Peled11}
S.~Har-Peled.
\newblock {\em Geometric Approximation Algorithms}.
\newblock American Mathematical Society, USA, 2011.

\bibitem{itai1982hamilton}
A.~Itai, C.~H. Papadimitriou, and J.~L. Szwarcfiter.
\newblock {H}amilton {P}aths in {G}rid {G}raphs.
\newblock {\em SIAM Journal on Computing}, 11(4):676--686, 1982.

\bibitem{SKBdoktori}
S.~Kisfaludi-Bak.
\newblock {\em ETH-Tight Algorithms for Geometric Network Problems}.
\newblock PhD thesis, Technische Universiteit Eindhoven, Department of
  Mathematics and Computer Science, June 2019.

\bibitem{Klein06}
P.~N. Klein.
\newblock A {S}ubset {S}panner for {P}lanar {G}raphs, with {A}pplication to
  {S}ubset {TSP}.
\newblock In {\em Proceedings of the 38th Annual {ACM} Symposium on Theory of
  Computing ({STOC} 2006)}, pages 749--756, 2006.

\bibitem{Klein08}
P.~N. Klein.
\newblock A linear-time approximation scheme for {TSP} in undirected planar
  graphs with edge-weights.
\newblock {\em {SIAM} J. Comput.}, 37(6):1926--1952, 2008.

\bibitem{ekmedian}
S.~G. Kolliopoulos and S.~Rao.
\newblock {A} {N}early {L}inear-{T}ime {A}pproximation {S}cheme for the
  {E}uclidean k-{M}edian {P}roblem.
\newblock {\em {SIAM} J. Comput.}, 37(3):757--782, 2007.

\bibitem{kortevygen12}
B.~Korte and J.~Vygen.
\newblock {\em Combinatorial Optimization: Theory and Algorithms}.
\newblock Springer Publishing Company, Incorporated, 5th edition, 2012.

\bibitem{KrauthgamerL06}
R.~Krauthgamer and J.~R. Lee.
\newblock Algorithms on negatively curved spaces.
\newblock In {\em 47th Annual {IEEE} Symposium on Foundations of Computer
  Science ({FOCS} 2006)}, pages 119--132. {IEEE} Computer Society, 2006.

\bibitem{hung-le-soda20}
H.~Le.
\newblock A {PTAS} for subset {TSP} in minor-free graphs.
\newblock In {\em Proceedings of the 2020 {ACM-SIAM} Symposium on Discrete
  Algorithms ({SODA} 2020)}, pages 2279--2298, 2020.

\bibitem{truly-optimal-spanners}
H.~Le and S.~Solomon.
\newblock Truly optimal euclidean spanners.
\newblock In {\em 60th {IEEE} Annual Symposium on Foundations of Computer
  Science ({FOCS} 2019)}, pages 1078--1100, 2019.

\bibitem{steiner-spanners}
H.~Le and S.~Solomon.
\newblock {L}ight {E}uclidean {S}panners with {S}teiner {P}oints.
\newblock In {\em 28th Annual European Symposium on Algorithms ({ESA} 2020)},
  pages 67:1--67:22, 2020.

\bibitem{Lichtenstein82}
D.~Lichtenstein.
\newblock Planar formulae and their uses.
\newblock {\em {SIAM} J. Comput.}, 11(2):329--343, 1982.

\bibitem{ManurangsiR17}
P.~Manurangsi and P.~Raghavendra.
\newblock {A} {B}irthday {R}epetition {T}heorem and {C}omplexity of
  {A}pproximating {D}ense {CSP}s.
\newblock In {\em 44th International Colloquium on Automata, Languages, and
  Programming ({ICALP} 2017)}, pages 78:1--78:15, 2017.

\bibitem{marx07}
D.~Marx.
\newblock On the optimality of planar and geometric approximation schemes.
\newblock In {\em 48th Annual {IEEE} Symposium on Foundations of Computer
  Science {(FOCS 2007)}}, pages 338--348, 2007.

\bibitem{Mitchell99}
J.~S.~B. Mitchell.
\newblock Guillotine subdivisions approximate polygonal subdivisions: {A}
  simple polynomial-time approximation scheme for geometric {TSP}, k-{MST}, and
  related problems.
\newblock {\em {SIAM} Journal on Computing}, 28(4):1298--1309, 1999.

\bibitem{geomspannet}
G.~Narasimhan and M.~H.~M. Smid.
\newblock {\em Geometric Spanner Networks}.
\newblock Cambridge University Press, 2007.

\bibitem{PapadimitriouBook}
C.~H. Papadimitriou.
\newblock {\em Computational Complexity}.
\newblock Addison-Wesley, 1994.

\bibitem{plesn1979np}
J.~Plesn{\'i}k.
\newblock {T}he {NP}-{C}ompleteness of the {H}amiltonian {C}ycle {P}roblem in
  {P}lanar {D}iagraphs with {D}egree {B}ound {T}wo.
\newblock {\em Information Processing Letters}, 8(4):199--201, 1979.

\bibitem{RaoS98}
S.~Rao and W.~D. Smith.
\newblock {A}pproximating {G}eometrical {G}raphs via "{S}panners" and
  "{B}anyans".
\newblock In {\em Proceedings of the Thirtieth Annual {ACM} Symposium on the
  Theory of Computing ({STOC} 1998)}, pages 540--550. {ACM}, 1998.

\bibitem{RobinsS94}
G.~Robins and J.~S. Salowe.
\newblock On the maximum degree of minimum spanning trees.
\newblock In K.~Mehlhorn, editor, {\em Proceedings of the Tenth Annual
  Symposium on Computational Geometry, Stony Brook, New York, USA, June 6-8,
  1994}, pages 250--258. {ACM}, 1994.

\bibitem{Snyder92}
T.~L. Snyder.
\newblock {O}n the {E}xact {L}ocation of {S}teiner {P}oints in {G}eneral
  {D}imension.
\newblock {\em {SIAM} J. Comput.}, 21(1):163--180, 1992.

\bibitem{Trevisan00}
L.~Trevisan.
\newblock {W}hen {H}amming {M}eets {E}uclid: {T}he {A}pproximability of
  {G}eometric {TSP} and {S}teiner {T}ree.
\newblock {\em {SIAM} J. Comput.}, 30(2):475--485, 2000.

\bibitem{vanwijland}
E.~van Wijland and H.~Zhou.
\newblock {Faster Approximation Scheme for Euclidean k-TSP}.
\newblock In W.~Mulzer and J.~M. Phillips, editors, {\em 40th International
  Symposium on Computational Geometry (SoCG 2024)}, volume 293 of {\em Leibniz
  International Proceedings in Informatics (LIPIcs)}, pages 81:1--81:12,
  Dagstuhl, Germany, 2024. Schloss Dagstuhl -- Leibniz-Zentrum f{\"u}r
  Informatik.

\bibitem{vazirani-book}
V.~V. Vazirani.
\newblock {\em Approximation Algorithms}.
\newblock Springer, 2004.

\bibitem{williamson2011design}
D.~P. Williamson and D.~B. Shmoys.
\newblock {\em The Design of Approximation Algorithms}.
\newblock Cambridge University Press, 2011.

\end{thebibliography}

\pagebreak

\begin{appendix}
\section{Filtering Algorithm}
\label{sec:filtering}

The goal of this section is to prove Lemma~\ref{lem:banyan}. \textbf{Most of this section is taken almost verbatim from the unpublished full version of~\cite{survivable}.} We include it for four reasons:
\begin{itemize}
    \itemsep0em 
\item to make this paper self-contained,
\item we need to analyze the running time dependence in a slightly different setting,
\item we need to extend the result to the $\ell_1$ metric in addition to the $\ell_2$ metric,
\item we are able to simplify some of the arguments due to our special setting.
\end{itemize}

For a given set of points $P \subseteq \mathbb{R}^d$ let $\smt(P)$ be the
minimum Steiner tree\footnote{\label{stunique}For the simplicity of notation we
act as if the minimum Steiner trees and minimum spanning trees are unique; one
can check that all our arguments hold if there are multiple minima.} with
terminals $P$ and $\mst(P)$ be the minimum spanning tree of $P$. For any $X
\subseteq \mathbb{R}^d$ let $\smt(P;X)$ be the\textsuperscript{\ref{stunique}}
minimum length Steiner tree that connects the terminal set $P$ and is only allowed to
use a subset of $X$ as Steiner vertices (therefore $\smt(P;\mathbb{R}^d) =
\smt(P)$). We use $\cB(x,r)$ to denote the ball centered at $x$ of
radius~$r$.

We will use some named constants in the following definitions, which all depend only on the dimension $d$. The constant $k$ will be defined later in Claim~\ref{claim:ball}; it will be set so that $k = \Theta(d^2)$.
    Let $\Delta$ be the maximum degree of any minimum spanning tree of points in
    $\mathbb{R}^d$. It is well-known that $\Delta \le 3^d$~\cite{RobinsS94}. We furthermore define: 
\begin{alignat*}{3}  
    \gamma &:= 4k\Delta^k/\eps &&= 2^{\Oh(d^3)}/\eps,\\
    \phi &:= \frac{\eps}{80d\Delta\gamma} &&= \eps^2/2^{\Oh(d^3)}.
\end{alignat*}
Finally, let $c^\ast$ be a universal constant to be defined later (completely independent of $d$ and $\eps$).

\begin{definition}[Steiner filter]  
    \label{smt-definition}
	For point sets $P_0,P_1 \subseteq \mathbb{R}^d$ and $\eps > 0$ we say that $X \subseteq
    P_1$ is a \emph{Steiner filter} of $P_1$ with respect to $P_0$ if:
    \begin{enumerate}[font=\normalfont, label=(\roman*)]
        \item $\displaystyle{|X| \le \frac{2^{c^*d^4}}{\eps^{c^*d}} |P_0|}$, and
        \item $\displaystyle{\wt(\smt(P_0;X)) \le (1+c^*\eps) \cdot \wt(\smt(P_0; P_1))}$, and
        \item $\displaystyle{\wt(\mst(P_0 \cup X)) \le \frac{2^{c^*d^4}}{\eps^{c^*d}} \cdot \wt(\mst(P_0))}$.
    \end{enumerate}
\end{definition}

In this section we prove the following filtering Lemma.

\begin{lemma}
    \label{ref:filtering-lemma}
    For any point set $P\subset \mathbb{R}^d$ and any $\eps > 0$, there is an
    algorithm that:
    
    \begin{enumerate}[label=(\alph*)]
        \item finds a $X \subseteq \mathbb{R}^d$ that is a Steiner filter of $P$ with respect to $\Reals^d$, and 
        \item runs in $\frac{2^{\Oh(d^4)}}{\eps^{\Oh(d)}}\cdot |P|\log(|P|)$ time.
    \end{enumerate}
\end{lemma}

The proof of our Lemma is based on the following theorem proved by Czumaj et
al.~\cite{survivable}.

\begin{theorem}[Lemma 4.1 from full version of \cite{survivable}]
    \label{czumaj-filtering}
    For any point sets $P_0$ and $P_1$ in $\mathbb{R}^d$ and any $\eps>0$, there is an algorithm that finds a subset $X$ of $P_0$ that
    is a Steiner filter of $P_1$ with respect to $P_0$ and runs in
    $(d/\eps)^{\Oh(d)} n \log n$ time, where $n =
    |P_0 \cup P_1|$.
\end{theorem}

\begin{algorithm}
	\SetAlgoLined
	\DontPrintSemicolon
	\SetKwInOut{Input}{Algorithm}
	\SetKwInOut{Output}{Output}
    \Input{$\mathtt{Steiner filtering}(P)$, points $P \subseteq \mathbb{R}^d$ snapped to $L^d$ grid, where $L = \Oh(n/\eps)$}
    Build a light $(1+\eps)$-spanner $G$ on $P$\tcp*{Theorem 15.3.20 in~\cite{geomspannet}}
    \ForEach{edge $e$ of $G$}{
        $r_e \coloneqq 20\gamma |e|$\\
        Set $\mathrm{grid}(e):=$ $d$-dimensional grid of side length $\phi \cdot  |e|$ and $z_e:=$ midpoint of $e$\\
        $X_e \coloneqq \mathrm{grid}(e) \cap \cB(z_e,r_e)$
    }
    \Return $X \coloneqq \bigcup_{e \in E[G]} X_e$
    \caption{Pseudocode for the Steiner filtering algorithm of Czumaj~et
    al.~\cite{survivable} when $P_0 = \mathbb{R}^d$.}
	\label{alg:filtering}
\end{algorithm}

Observe that if we were to plug naively the set $P_1 = \mathbb{R}^d$ into
Theorem~\ref{czumaj-filtering} we would already prove
Lemma~\ref{ref:filtering-lemma}. Unfortunately, the running time of
Theorem~\ref{czumaj-filtering} depends on $|P_0 \cup P_1|$. For our purposes we
only need to analyze their algorithm in the special case of $P_1 = \mathbb{R}^d$. In
Algorithm~\ref{alg:filtering} we present (a simplified version of) the Steiner filtering procedure
of~\cite{survivable} in the special case of $P_1 = \mathbb{R}^d$. First a
light $(1+\eps)$-spanner $G$ of $P$ is computed (this step already takes
$\eps^{-\Oh(d)} |P| \log(|P|)$ time). Next, for every edge~$e$
of~$G$ we consider the $d$-dimensional axis-parallel grid 
of cell side length $\phi|e|=\eps^2 |e|/2^{\Oh(d^3)}$. We add to~$X$ all the grid points within distance $20\gamma|e|=2^{\Oh(d^3)} |e|/\eps$ from the midpoint of~$e$. After processing all edges of~$G$ in this manner we return the resulting point set~$X$.

Now, we show that Algorithm~\ref{alg:filtering} runs in $\frac{2^{\Oh(d^4)}}{\eps^{\Oh(d)}} |P|$ time.

\begin{proof}[Proof of Lemma~\ref{ref:filtering-lemma} (b)]
    To construct a light spanner $G$ we need $\eps^{-\Oh(d)} |P| \log{|P|}$ time (note that
    Algorithm~\ref{alg:filtering} computes the spanner on $P$). Known
    constructions guarantee that the number of edges of such a spanner is bounded
    by $|P|/\eps^{\Oh(d)}$~\cite{survivable,gln02,geomspannet}. For any edge $e
    \in E[G]$, the number of grid points that we add to $X$ is bounded by
    $2^{\Oh(d^4)}/\eps^{\Oh(d)}$. The time needed to construct the
    grid around edge $e$ is bounded by $2^{\Oh(d^4)}/\eps^{\Oh(d)}$. Therefore the
    total runtime of Algorithm~\ref{alg:filtering} is bounded by
    $\Tspan(|P|,\eps) + \frac{2^{\Oh(d^4)}}{\eps^{\Oh(d)}} |P|$.
\end{proof}

Lemma~\ref{ref:filtering-lemma} (a) follows from Theorem~\ref{czumaj-filtering}.  For completeness we provide the detailed arguments in Section~\ref{smt-filter-correctness}, making some simplifications that come from our specific setting.

\subsection{Properties of Spanners and Steiner Trees}

\begin{claim}[Lemma 4.6 in~\cite{survivable}]\label{cl:shorteredge}
    Let $P \subseteq \mathbb{R}^d$ and let $x,y$ be any pair of distinct points in $P$.
    Let $uv$ be any edge in $\smt(P)$ that separates $x$ and $y$ in $T$.

    Suppose that $G$ is a connected graph on $P$ and $p_{x \leadsto y}$ is the shortest path in $G$ from $x$ to $y$.
    Then at least one edge of $p_{x \leadsto y}$ is not shorter than $\dist(u,v)$.
\end{claim}

\begin{proof}
    Suppose for the sake of contradiction that all edges on $p_{x \leadsto y}$ are shorter
    than $\dist(u,v)$, and let $T_1,T_2$ be the two connected components of $\smt(P)-uv$. Since $x$ and $y$ are in two different components there must
    exist an edge $e \in p_{x \leadsto y}$ such that its endpoints are in different
    components $T_i$ and $|e| < \dist(u,v)$. This means that $T_1 \cup T_2 \cup
    \{e\}$ is a Steiner tree of $P$ and has smaller length than $\smt
     (P)$, which contradicts the minimality of $\smt(P)$.
\end{proof}

\begin{claim}[Lemma 4.7 in~\cite{survivable}]
    \label{claim:steiner}
    For any $P,X \subseteq \mathbb{R}^d$ let $T$ be any subtree of $\smt(P;X)$ and let $L_T$ be the set of leaves of
    $T$. Then $T=\smt((P\cap T) \cup L_T;X)$.\footnote{More precisely, $T$ is a minimum Steiner tree of $(P\cap T) \cup L_T$ that uses Steiner points only from $X$.}
\end{claim}
\begin{proof}
    For the sake of contradiction assume that $T^\ast = \smt((P\cap T) \cup L_T,X)$ and
    $\wt(T^\ast)$ is smaller than $\wt(T)$. We replace the tree $T$ by the
    tree $T^\ast$ in $\smt(P,X)$. We get a connected graph of weight smaller than
    $\smt(P;X)$ that contains a Steiner Tree of $P$ and uses only $X$ as
    Steiner points. However, this graph has smaller weight than $\smt(P;X)$
    which implies a contradiction.
\end{proof}

\begin{claim}[Lemma 4.8 in~\cite{survivable}]
    \label{claim:ball}
    For any $P,X \subseteq \mathbb{R}^d$ and $\rho_0 \in \mathbb{R}^d$ if $\mathcal{B}(\rho_0,1)$ has at most $1$
    terminal from $P$, then for any $\alpha \in (0,1/2]$ the tree $\smt(P;X)$ contains less than $2^k$ Steiner points
    from $X\cap \mathcal{B}(\rho_0,\alpha)$ where $k=\Oh(d^2)$ is an integer multiple of $d^2$. The statement is true for both $\ell_1$ and $\ell_2$ norms.
\end{claim}

The proof of this claim follows~\cite{survivable}. An analogous statement to~Claim~\ref{claim:ball} also appears in~\cite[Lemma
36]{RaoS98}.

\begin{proof}
    For brevity of notation let $\mathcal{B}_r$ be $\mathcal{B}(\rho_0,r)$. Let $T_r$ be the subforest of $\smt(P;X)$ consisting of edges with 1 or 2 endpoints in $\mathcal{B}_r$. 
    Let $s_r$ be the number of internal vertices in $T_r$. Note that when $r\leq 1$, then all but one internal vertex of $T_r$ has to be a Steiner vertex. Finally, let $n_r$ be the number of
    edges of $\smt(P;X)$ with exactly one endpoint inside $B_r$.
    Observe that for any $r \in (0,1]$, the value $s_r$ counts the internal
    points of $T_r$ plus perhaps the single terminal from $P$, while $n_r$ counts the number of leaf edges of $T_r$.
    Our goal is therefore to show $s_r \le 2^{\Oh(d^2)}$. 

    Observe that we can assume that the degree of Steiner vertices is at least
    $3$ and at most $\Delta$. Therefore $s_r+1 \le n_r \le \Delta\cdot (s_r+1)
    \le 2\Delta s_r$ (we use $s_r \ge 1$, as otherwise the claim is obviously
    true). Therefore, $T_r$ has at most $(s_r+1)+n_r \le 2n_r$ vertices.

    Consider arbitrary $r$ and $r^\ast$, such that $\alpha \le r < r^\ast \le
    1$. If $n_\alpha < 2^d$, then $s_\alpha < n_\alpha < 2^d$ and the claim is
    correct. Hence, we assume from now that $n_{r^\ast} \ge n_r \ge n_\alpha \ge
    2^d$.

    Next, we show that the length of $T_{r^\ast}$ is at least $(r^\ast - r)
    \cdot s_r$. Observe that each edge counted by $n_r$ is on at least one
    path to a leaf of the tree $T_r$. Moreover, there is at most one non-Steiner point
    inside $\mathcal{B}_1$. Hence every edge (except perhaps one) counted by $n_r$ is on at least one
    path from $\mathcal{B}_r$ to the outside of $\mathcal{B}_1$. Hence the
    length of $T_{r^\ast}$ is at least $(r^\ast - r) (n_r-1) \ge (r^\ast - r)
    s_r$.

    Next, we upper bound the length of $T_{r^\ast}$.  Observe that $T_{r^\ast}$
    contains at most $2n_{r^\ast}$ points.  It is well known that for any $P \subseteq \mathcal{B}_1$ of
    at least $2^d$ points we have $\mst(P) \le 8 r |P|^{1-1/d}$~\cite[Exercise
    6.3]{geomspannet}, \cite[Lemma 4.4]{survivable}.
    Therefore,
    \[
        (r^\ast - r) \cdot s_r \le \wt(T_{r^\ast}) \le 16 \cdot (n_{r^\ast})^{1-1/d}.
    \]
    We combine this with $n_{r^\ast} \le 2\Delta s_{r^\ast}$ and set $r^\ast=r+\eps$. We get that for every $\eps \in (0,1)$ and $r \in (0,1 - \eps)$ it
    holds that:
    \begin{equation}\label{apx:recursion}
        s_{r} \leq \frac{16}{\eps}(n_{r+\eps})^{1-1/d} < \frac{32 \Delta}{\eps} \cdot (s_{r+\eps})^{1-1/d}.
    \end{equation}
    It remains to show that $s_\alpha$ must be bounded by $2^{\Oh(d^2)}$ for $\alpha \le 1/2$.
    Let $\tau \coloneqq \frac{1}{\ln\ln(s_1)}$ (we can assume that $\tau < 1/2$,
    as otherwise $s_1$ is constant-bounded). Fix some $r\in (0,1-\tau]$ and set $\eps = \frac{1}{d (\ln\ln s_{r+\tau})^2}$. Observe that $s_{r+\tau}< s_1$ and thus $\tau/\eps<d\ln\ln s_{r+\tau}$. We iterate Inequality~\eqref{apx:recursion} $\floor{\tau/\eps}$ times:
    \begin{equation}\label{eq:iter1}
\begin{split}
        s_r &< 32\Delta d (\ln\ln s_{r+\tau})^2 (s_{r+\eps}^{1-1/d})< \dots \\
        &< (32\Delta d (\ln \ln s_{r+\tau})^2)^{\sum_{j>0} (1-1/d)^j} \cdot s_{r+\tau}^{(1-1/d)^{\tau/\eps}}\\
        &< (100 d \cdot \Delta \cdot \ln \ln s_{r+\tau})^{2d},
\end{split}
\end{equation}
where we used that $\sum_{j>0} (1-1/d)^j = d$ and the following bound:
\[s_{r+\tau}^{(1-1/d)^{\tau/\eps}}<s_{r+\tau}^{(1-1/d)^{d\ln\ln s_{r+\tau}}}<s_{r+\tau}^{(1/e)^{\ln\ln s_{r+\tau}}}=s_{r+\tau}^{1/\ln s_{r+\tau}}=e.\]
    Next, we iterate Inequality~\eqref{eq:iter1} $\ell = \ln^\ast(s_1)$  times for $r \in \{\alpha,\alpha+\tau,\dots\}$. Assuming $\alpha \in [0,1-\ell\tau)$, we have that $\ln^{(\ell)} s_{\alpha+\ell\tau}<\ln^{(\ell)} s_1\leq 1$, where $\ln^{(j)}(.)$ denotes the $j$-times iterated natural logarithm. Therefore we can bound $s_\alpha$ as follows.
\begin{align*}
        s_\alpha &< (100 d \Delta \cdot \ln \ln s_{\alpha+\tau})^{2d} \\
        &<\Big(100 d \Delta \cdot \ln \ln \big(100 d \cdot \Delta \cdot \ln \ln s_{\alpha+2\tau})^{2d}\big)\Big)^{2d}\\
        &= \Big(100 d \Delta \cdot \big( \ln(2d) + \ln(\ln(100 d)+\ln \Delta + \ln^{(3)} s_{\alpha+2\tau})\big)\Big)^{2d}\\
&< \Big(100d \Delta \cdot \big(4\ln d + \max\{\ln\ln(100d), \ln\ln\Delta, \ln^{(4)} s_{\alpha+2\tau}\})\big)\Big)^{2d} \\
		&< \Big(100d \Delta \cdot \big(4\ln d+ \max_{j\geq 1}\{\ln^{(2j+1)} (2d), \ln^{(2j)}(100d), \ln^{(2j)}\Delta, \ln^{(2\ell)} s_{\alpha+\ell\tau}\}\big)\Big)^{2d}\\
        &< (100d \Delta \cdot 20\ln d)^{2d}\\
        &< (2000 \Delta d\ln d)^{2d}\\
        &=2^{\Oh(d^2)}
\end{align*}
    Note that when $\ell\tau \ge 1/2$, then $s_1$
    is constant-bounded. Hence, for every $\alpha \in (0,1/2]$, it holds that $s_\alpha \le 2^{\Oh(d^2)}$.
\end{proof}

\subsection{Proving that \texorpdfstring{$X$}{\textit{X}} is a Steiner filter of \texorpdfstring{$P$}{\textit{P}}}
\label{smt-filter-correctness}

\subsubsection{Property (i) of \texorpdfstring{$X$}{\textit{X}}:} The set $X$ returned by Algorithm~\ref{alg:filtering} satisfies property (i)
of Definition~\ref{smt-definition} because there are at most
$|P|/\eps^{\Oh(d)}$ edges in $G$ and for each edge we add at most
\begin{equation}\label{eq:edggridsize}
|X_e| =\Oh\left(\frac{20\gamma}{\phi}\right)^d=2^{\Oh(d^4)}/\eps^{\Oh(d)}
\end{equation}
points to $X$.

\subsubsection{Property (ii) of \texorpdfstring{$X$}{\textit{X}}:} We prove that $\wt(\smt(P; X)) \le
(1+2\eps) \cdot \wt(\smt(P;\mathbb{R}^d))$. Let
$\mathcal{T}^\ast$ be the minimum Steiner tree of $P$ that can use any point in $\mathbb{R}^d$ as
a Steiner point. We build a graph $H$ in three steps, consisting of edge sets $H_1,H_2$ and $H_3$. The set $H_1$ consists of all
the edges $(u,v) \in \mathcal{T}^\ast$ with $u,v \in P$. The set $H_2$ will be defined later; these edges are created from edges
$(u,v) \in \mathcal{T}^\ast$ by moving the endpoints $(u,v)$ to $X$. Finally, the set $H_3$ makes $H$ connected by adding
additional edges incident to the endpoints of $H_1\cup H_2$ that have total minimum weight.

Let $\mathcal{E}_1$ denote the edges of
$\mathcal{T}^\ast$ with both endpoints in $P$. Moreover, $\mathcal{E}_2$ are the
edges of $\mathcal{T}^\ast$ that are transformed in the second phase into $H_2$ and $\mathcal{E}_3 = \mathcal{T}^\ast \setminus (\mathcal{E}_1 \cup
\mathcal{E}_2)$.

Clearly the graph $H$ is spanning $P$ and uses only points from $X\cup P$ as endpoints; consequently, it contains a Steiner tree spanning $P$ that uses only $X$ as Steiner vertices. We need to bound the total weight of $H$. Czumaj et
al.~\cite{survivable} show that $\wt(H_2) \le \wt(\mathcal{E}_2) + \eps \cdot
\wt(\mathcal{T}^\ast)$ and $\wt(H_3) \le \wt(\mathcal{E}_3) + \eps \cdot
\wt(\mathcal{T}^\ast)$. Since $\wt(H_1) = \wt(\mathcal{E}_1)$, this leads to:
\begin{align*}
    \wt(H) & \le \wt(H_1) + \wt(H_2) + \wt(H_3) \\
    & \le \wt(\mathcal{E}_1)+ (\wt(\mathcal{E}_2) + \eps \wt(\mathcal{T}^\ast)) +
    (\wt(\mathcal{E}_3) + \eps \wt(\mathcal{T}^\ast)) \\
            & \le (1+2\eps)\wt(\mathcal{T}^\ast).
\end{align*}

We now expand on the construction of $H_2$ and $H_3$.

\paragraph{Construction of \texorpdfstring{$H_2$}{\textit{H2}}:}

We begin by constructing the graph $H_2$. First, we define its set of edges
$\mathcal{E}_2$.  Recall that we defined $k$ to be the constant
from Claim~\ref{claim:ball} and $k = \Theta(d^2)$ and $\gamma = 4k \Delta^k/\eps$.

Since the edge $(u,v) \in \mathcal{T}^\ast$ is not in $\mathcal{E}_1$ either
$u$ or $v$ are not in $P$. Let $t$ be the midpoint of $(u,v)$. 
For a tree $\mathcal{T}^\ast$, let $\smt_u$ and $\smt_v$ be the two subtrees of
$\mathcal{T}^\ast$ that arise after removing an edge $(u,v)$ from
$\mathcal{T}^\ast$ and $u \in \smt_u$ and $v \in \smt_v$. Let $T_u^k$ be
the subtree of $\smt_u$ induced by the vertices within hop-distance $k$ from $u$
(analogously $T_v^k$ is a subtree of $\smt_v$ induced by the vertices within
hop-distance $k$ from $v$). Let $\ell = \dist(u,v)$. Then, Czumaj et
al~\cite{survivable} add an edge $(u,v)$ to $\mathcal{E}_2$ if all the following
conditions hold:
\begin{enumerate}[label=(C\arabic*)]
    \item\label{enum1} every edge in $T_u^k$ and $T_v^k$ is shorter than $2\gamma\ell/k$.
    \item\label{enum2} $\smt_u$ has at least one point (call it $x$) from $P$ that is contained in the ball $\mathcal{B}(t,4\gamma\ell)$.
    \item\label{enum3} $\smt_v$ has at least one point (call it $y$) from $P$ that is contained in the ball $\mathcal{B}(t,4\gamma\ell)$.
\end{enumerate}

This concludes the construction of $\mathcal{E}_2$. Let $\mathcal{E}_3$ be
$\mathcal{T}^\ast \setminus (\mathcal{E}_1 \cup \mathcal{E}_2)$.
It remains to analyse the total length of these edges.

\paragraph{Bound on \texorpdfstring{$\wt(H_2)$}{wt(\textit{H2})}:} 

Since $T_u^k$ contains vertices within hop-distance $k$ from $u$, Condition \ref{enum1} implies that $T_u^k\subset \cB(t,4\gamma\ell)$, and the same holds for $T_v^k$. Conditions \ref{enum2} and \ref{enum3} ensure that $x$ and $y$ are disconnected in $T^*$ after the removal of the edge $uv$. Claim~\ref{cl:shorteredge} shows that the spanner $G$ has at least one edge $e$ on 
$p_{x\leadsto y}$ whose length is at least $|e|\geq \ell$.

Next, we prove that $p_{x\leadsto y}$ is contained in $\cB(t,(20\gamma\ell)$. First note that $x,y\in \cB(t,4\gamma\ell)$, and thus $\dist(x,y)\leq 8\gamma\ell$. On the other hand, we have $|p_{x\leadsto y}|\leq (1+\eps)\dist(x,y)$ by the spanner property of $G$. Thus for any point $z$ of $p_{x\leadsto y}$ we have $\dist(x,z)\leq (1+\eps)\dist(x,y)<2\dist(x,y)\leq 16\gamma\ell$, and by the triangle inequality $\dist(t,z)\leq \dist(t,x)+\dist(x,z)<20\gamma\ell$.

Now since $p_{x\leadsto y}$ is contained in $B(t,20\gamma\ell)$ and $e$ is in $p_{x\leadsto y}$, we have that both endpoints of $e$ are contained in $B(t,20\gamma\ell)$.
By Claim~\ref{cl:shorteredge} and the construction of $X$, there exist two points $u^\ast$ and $v^\ast$ in $P
\cup X$, such that $\dist(u,u^\ast)$ and $\dist(v,v^\ast)$ are at most
$d\phi |e|$; this is because the axis-parallel square grid of side-length $\phi|e|$ has cells of $\ell_2$-diameter $\sqrt{d}\phi |e|$ and $\ell_1$-diameter $d\phi |e|$. We modify $\mathcal{T}^\ast$ by moving
all edges incident to $u$ and $v$ to have their endpoints at $u^\ast$ and
$v^\ast$. Since the degree of both $u$ and $v$ is bounded by $\Delta$, the operations at $u$ and $v$ will increase the cost of $T^*$ by at most
an additive term $2\Delta \cdot d\phi|e|$. Because $e$ is completely contained in $\cB(t,20\gamma\ell)$,
we have $|e| \leq 2 \cdot 20\gamma \ell \leq 40 \gamma\ell$.
Thus, the cost of $T^*$ will increase by at most an additive term

\[2\Delta d\phi |e| \leq 2\Delta d\phi \cdot 40 \gamma\ell= \eps\ell = \eps \dist(u,v).\]

To summarize the total cost of the edges in $H_2$ is bounded by
$\wt(\mathcal{E}_2) +
\sum_{(u,v)\in \mathcal{T}^\ast} \eps \cdot \dist(u,v) \le \wt(\mathcal{E}_2) + \eps
\wt(\mathcal{T}^\ast)$.

\paragraph{Construction of \texorpdfstring{$H_3$}{\textit{H3}}:}

Observe that $\mathcal{E}_3$ induces a forest. For each tree $T'$ in this forest let $V'$ be the set of vertices in $T'$ belonging to $P$.
Take any minimum spanning tree on $V'$ and add it to $H_3$. This concludes the construction of $H_3$.

\paragraph{Bound on \texorpdfstring{$\wt(H_3)$}{wt(\textit{H3})}:} 

Consider removing edge $(u,v) \in \smt(P)$ from $\smt(P)$. Let $\smt_u(P)$ and
$\smt_v(P)$ denote the resulting trees containing $u$ and $v$ respectively.
Recall that $T_u^k$ is the subtree of $\smt_u(P)$ induced by the vertices that
is at a hop distance at most $k$ from $u$ for some parameter $k = \Theta(d^2)$.

Consider an edge $(u,v) \in \mathcal{E}_3$.
Because $(u,v) \notin \mathcal{E}_1 \cup \mathcal{E}_2$ at least one of the following
conditions hold: (i) either $T_u^k$ or $T_v^k$ contains an edge of length greater
than $2\gamma\ell/k$, or (ii) $\smt_u(P)$ or $\smt_v(P)$ has no
endpoint in the ball $\mathcal{B}(t,4\gamma\ell)$ (where $t$ is
the midpoint of edge $(u,v)$) (see conditions~\ref{enum1},~\ref{enum2} and \ref{enum3}).

Next, we will show that when $(u,v) \in \mathcal{E}_3$ then (ii) cannot hold.

\begin{claim}
    \label{cl:long-edge}
    If $(u,v) \in \mathcal{E}_3$ then there is an edge $(u^\ast,v^\ast) \in T_u^k \cup T_v^k$ of length at least
    $2\gamma\ell/k$. 
\end{claim}
\begin{proof}
    Fix $R \coloneqq 4\gamma\ell$. Suppose for the sake of contradiction that all edges in $T_u^k$ are smaller than
    $R/2k$.
    Since $(u,v)\not\in \mathcal{E}_1 \cup \mathcal{E}_2$, we have that $\smt_u(P) \cap \mathcal{B}(t,R)$ or $\smt_v(P) \cap \mathcal{B}(t,R)$ is empty of terminals (other than $u,v$). Without loss of generality, assume that $P \cap \smt_u(P) \cap \mathcal{B}(t,R) =
    \{u\}$.  Then, $T_u^k$ must by fully contained in
    $\mathcal{B}(t,R/2)$. Since there are no
    terminals in $\mathcal{B}(t,R)$ other than $u$, the tree $T_u^k$ must
    contain at least $2^k$ Steiner points. Note that $\smt_u(P)$ is a Steiner tree of its leaves by Claim~\ref{claim:steiner}. Now applying
    Claim~\ref{claim:ball} to $\smt_u(P)$ and the balls $\mathcal{B}(t,R)$ and $\mathcal{B}(t,R/2)$ leads to a contradiction, as $\mathcal{B}(t,R/2)$ has at least $2^k$ Steiner points.
\end{proof}

Czumaj et al.~\cite{survivable} charge the cost of the edge $(u,v) \in \mathcal{E}_3$  to the edge
$(u^\ast,v^\ast) \in T_u \cup T_v$ guaranteed by Claim~\ref{cl:long-edge} in order to bound the cost of edges in $H_3$. We will show that the total cost
of all edges charged to $e$ is upper bounded by $\eps |e|$.

To analyze the charging scheme observe that an edge can be charged to $e \in
\mathcal{T}^\ast$ only if it is at most $(k-1)$ hops from one of the
endpoints of $e$. The number of such edges is at most $4 \Delta^k$ (where
$\Delta \le 3^d$ is the maximum degree of any minimum spanning tree
in $\mathbb{R}^d$). Moreover, the lengths of such edges are upper bounded
by $\eps|e|/(4\Delta^k)$: by Claim~\ref{cl:long-edge} an edge $(u,v)$
is charged to $(u^\ast,v^\ast)$ only if $\dist(u,v) \ge
2\gamma\ell/k$.  This shows that $\wt(\mathcal{E}_3) \le \eps \wt(T^\ast)$.

Recall that for each tree $T'$ of the forest in $\cE_3$ with vertex set $V'=V(T')$ there is a spanning tree of $V'$ in $H_3$. Because the cost of the minimum spanning tree on $V'$ is at most twice
the cost of $\wt(\smt(V')) \le \wt(T')$ it shows that $\wt(H_3) \le
2\wt(\mathcal{E}_3)$ and therefore $\wt(H_3) \le \wt(\mathcal{E}_3) + \eps
\wt(T^\ast)$.
This concludes the proof of Property (ii).

\subsubsection{Property (iii) of \texorpdfstring{$X$}{\textit{X}}:} Now, we prove property (iii) of $X$,
namely that $\wt(\mst(P \cup X)) \le
2^{\Oh(d^5)}/\eps^{\Oh(d)} \cdot \wt(\mst(P))$ (see Lemma~4.10 in the full version
of~\cite{survivable}). We construct a spanning graph $\mathcal{T}$ of $P \cup
X$. First we take $\mathcal{T}$ to be the minimum spanning tree of $P$. Next for
every edge $e \in G$ we find a minimum spanning tree $\mathcal{T}_e$ of $X_e$
and connect it to any of the endpoints of $e$. Such a graph is a spanning
graph of $P \cup X$. Now we focus on estimating the cost of $\mathcal{T}$.

Fix an arbitrarily edge $e$ in $G$. By \eqref{eq:edggridsize} we have $|X_e| = \frac{2^{\Oh(d^4)}}{\eps^{\Oh(d)}}$, and each edge has length at most $d\phi |e|=\eps^2/2^{\Oh(d^3)}$. Therefore the total length of the minimum spanning tree $\mathcal{T}_e$ is:

\[
    \Oh(|X_e|\cdot \mu)=\frac{2^{\Oh(d^4)}}{\eps^{\Oh(d)}} |e|.
\]

Hence, the total cost of $\mathcal{T}$ is bounded by

\[
    \wt(\mst(P)) + \sum_{e \in E[G]}  \frac{2^{\Oh(d^4)}}{\eps^{\Oh(d)}} \cdot |e| \le
    \frac{2^{\Oh(d^4)}}{\eps^{\Oh(d)}} \cdot  \wt(G) \le
    \frac{2^{\Oh(d^4)}}{\eps^{\Oh(d)}} \cdot \wt(\mst(P))
\]

This concludes the proof of Property (iii) of the set $X$ and the proof of
Lemma~\ref{ref:filtering-lemma}. Note that the constant $c^\ast$ used
in~\cref{smt-definition} can be properly defined to be global constant greater
than any constant hidden behind the $\Oh(\cdot)$ notation in this section.

\subsection{Proof of Lemma~\ref{lem:banyan}}

Now we proceed with the proof of Lemma~\ref{lem:banyan}. Recall that we are
given a point set $P$ and a random offset $\mathbf{a}$. The task is to find a set
$\tilde{S}$ of segments with the property that (i) there is a Steiner tree that uses
$\tilde{S}$ which is a $(1+\eps)$-approximation of the optimum
solution and (ii) it has at most $(1/\eps)^{\Oh(d)}$ crossings with each facet of
the quadtree.

We use the Lemma~\ref{ref:filtering-lemma} and get a set $X \subseteq
\mathbb{R}^d$, with the property that $\smt(P;X)$ is a $(1+\eps)$-approximation
of $\smt(P)$ and $\mst(P \cup X) \le \frac{2^{\Oh(d^4)}}{\eps^{\Oh(d)}} \cdot
\mst(P)$. Next, we compute a graph $G$ that is a light $(1+\eps)$-spanner of $P
\cup X$. Here, we use a spanner construction due
to~\cite{spanner-doubling-metrics} that guarantees a light spanner in doubling
metrics (to make it work both in Euclidean and Rectilinear Space) and works in
$\Oh(n \log n)$ time.

The graph $G$ has weight $\poly(1/\eps) \mst(P)$ and there is a
$(1+\Oh(\eps))$-approximate Steiner tree on $P$
that uses only the edges of $G$. Now, we need to guarantee that edges of $G$
are crossing each facet of a quadtree at most $1/\eps^{\Oh(d)}$ times. To
achieve that we use the lightening procedure of Rao and Smith~\cite{RaoS98} (see
Lemma~\ref{lem:spann}). Their lightening procedure works for any connected
graph (see also the lightening procedures
in~\cite{survivable,euclidean-mincostconnectivity} and \cite[Lemma 19.3.2]{geomspannet}).

This gives us a graph $\tilde{G}$ with $\wt(\tilde{G}) - \wt(G) \le \eps
\wt(\smt(P))$, it contains a Steiner tree on $P$ that is only
$(1+\Oh(\eps))$ times heavier than $\smt(P)$, and each shared facet of sibling cells in the quadtree is
crossed by $\tilde{G}$ at most $1/\eps^{\Oh(d)}$ times. The set $\tilde{S}$
consists of the edges of $\tilde{G}$. This concludes the proof of Lemma~\ref{lem:banyan}.

\end{appendix}

\end{document}